\RequirePackage{boolean}
\input{suspended.cfg}

\documentclass[acmsmall,screen,nonacm%
\Anonymous{,anonymous}{}%
\Final{}{,review}%
]{acmart}

\usepackage{suspended}

\author{Alistair O'Brien}
\orcid{0009-0009-0055-7793}
\email{ajo41@cam.ac.uk}
\affiliation{
  \institution{University of Cambridge}
  \city{Cambridge}
  \country{United Kingdom}
}

\author{Didier R\'emy}
\orcid{0000-0002-0693-6278}
\email{didier.remy@inria.fr}
\affiliation{
  \institution{INRIA}
  \city{INRIA}
  \city{Inria Paris Center}
  \country{France}
}

\author{Gabriel Scherer}
\orcid{0000-0003-1758-3938}
\email{gabriel.scherer@inria.fr}
\affiliation{
  \institution{INRIA \& IRIF, Universit\'e Paris Cit\'e}
  \city{Paris}
  \country{France}
}



\copyrightyear{2025}
\setcopyright{cc}
\setcctype{by}
\acmJournal{TOPLAS}
\acmDOI{XXXXXXX.XXXXXXX}

\received{November 2025}
\received[revised]{April 2026}
\received[accepted]{September 2026}

\newcommand{\BBemlti}{Pottier-Remy/emlti}
\newcommand{\BBpolyml}{Garrigue-Remy/poly-ml}
\newcommand{\BBpolyparams}{White/polyparams@ml2023}
\newcommand{\BBoutsidein}
   {Vytiniotis-Peyton-Jones-Schrijvers-Sulzmann/outsidein@jfp2011}

\title{Omnidirectional type inference for \ML: principality any way}

\begin{document}

\begin{abstract}

  The Damas-Hindley-Milner (\ML) type system owes its success to
  \emph{principality}, the property that every well-typed expression has a
  unique most general type. This makes inference predictable and efficient.
  Yet, principality is \emph{fragile}: many extensions of \ML---GADTs,
  higher-rank polymorphism, and static overloading---break it by introducing
  \emph{fragile} constructs that resist principal inference. Existing
  approaches recover principality through \emph{directional} inference
  algorithms, which propagate \emph{known} type information in a fixed (or
  \emph{static}) order (\eg as in bidirectional typing) to disambiguate such
  constructs. However, the rigidity of a \emph{static} inference order often
  causes otherwise well-typed programs to be rejected.


  We propose \emph{omnidirectional} type inference, where type information
  flows in a \emph{dynamic} order. Typing constraints may be solved in any
  order, suspending when progress requires known type information and resuming
  once it becomes available, using \emph{suspended match constraints}. This
  approach is straightforward for simply typed systems, but extending it to \ML
  is challenging due to \emph{let-generalization}. Existing \ML inference
  algorithms type let-bindings in a fixed order---type the let-bound expression 
  first, generalize its type, and then type the let-body. To overcome this, we
  introduce \emph{incremental instantiation}, allowing partially solved type
  schemes containing suspended constraints to be instantiated, with a mechanism
  to incrementally update instances as the scheme is refined.
  %
  %
  Omnidirectionality provides a \emph{general framework} for restoring
  principality in
  the presence of fragile features. We demonstrate its versatility on two
  fundamentally different features of \OCaml: static overloading of record
  labels and datatype constructors and semi-explicit first-class polymorphism.
  In both cases, we obtain a \emph{principal} type inference algorithm that is
  more expressive than \OCaml's current typechecker.

\end{abstract}
\maketitle

\section{Introduction}
\label{sec/introduction}

\parcomment {Introduction. What is \ML, what is principality?}

The Damas-Hindley-Milner (\ML) \citep*{Damas-Milner/W@popl82,
hindley1969principal} type system has long occupied a sweet spot in the design
space of strongly typed programming languages, as it enjoys the \emph{principal
types property}: every well-typed expression $\e$ has a most general type $\ts$
from which all other valid types for $\e$ are instances of $\ts$. For example,
the identity function $\efun \x \x$ has the principal type $\tfor \tv \tv \to
\tv$, generalizing types like $\tint \to \tint$ and $\tbool \to \tbool$.

\parcomment {Benefits of principality: practical implications}

The existence of principal types in \ML has important practical benefits. It
makes inference predictable, compositional, and efficient: since every
expression has a most general type, local typing decisions are always optimal,
with no need for guessing or backtracking. Beyond inference, principality
ensures that well-typedness is stable under common program transformations such
as let-contraction (or inlining) and argument reordering.

\parcomment {Principality is fragile. Extensions often break it}

Principality, however, is fragile. Many extensions of \ML---such as extensible
records with row-polymorphism \citep*{Remy/popl89, conf/popl/RemyV97,
conf/lics/Wand89, journals/toplas/Ohori95, garrigue1998programming} and
higher-kinded types \citep*{journals/jfp/Jones95}---are \emph{robust}: they
preserve principality. Others, including GADTs
\citep*{conf/icfp/SchrijversJSV09, conf/aplas/GarrigueR13}, higher-rank and
first-class polymorphism \citep*{conf/popl/OderskyL96,
\BBpolyml, journals/pacmpl/SerranoHJV20}, and static overloading
\citep*{Chargueraud-Bodin-Dunfield-Riboulet/jfla2025}, are \emph{fragile}: they
break principality under their \emph{natural} typing rules.

\parcomment {Example}

This fragility can already be observed in \OCaml through impredicative
higher-rank (\ie first-class) polymorphism, exposed by \emph{polymorphic
methods}~\citep{\BBpolyml} (\smashcolorbox{welltyped}{green (\cmark)} indicates
typechecking success and \smashcolorbox{illtyped}{red (\xmark)} indicates
failure):

\vspace{1ex}
\begin{program}[input]
let self x = x#f x °\Ocamlcomment{\ocamlFlag {\OCaml}1}°
\end{program}

In \OCaml, objects are defined as a collection of methods within
\code{object ... end}, and accessed using $\e \esend m$. Unlike Java or C++,
\OCaml uses \emph{structural typing} for objects: object types are a list of
method types between two chevrons \eg \ocaml[angles]{<f : 'a. 'a -> 'a>},
where the method \code{f} has the polymorphic identity function type
$\all \tv {\tv \to \tv}$ (the $\forall$ being omitted in \OCaml syntax).
When typing \code{self} in the example above, one could \emph{guess} the type
of \code{x} to be either $\epolymeth f \tv {\tv \to \tv}$ or $\epolymeth f \tv
{\tv \to \tv \to \tv}$---neither of which is strictly more general than the
other, violating principality.

\parcomment {Recovering principality}

Principality can be recovered through explicit type annotations. The return
type of overloaded datatype constructors may be annotated; polymorphic
expressions can be annotated with a type scheme; and for GADTs, both the type
of the \texttt{match} scrutinee and return type can be annotated with a rigid
type, which is refined by type equalities introduced in each branch.
In this example, the binding of \code{x} should be annotated with the
higher-rank type $\epolymeth f \tv {\tv \to \tv}$.

\parcomment {Fragile constructs can be viewed as an elaboration into robust ones}

Every fragile construct has a corresponding \emph{robust} form where the type
annotation is mandatory---for instance, $\eannot {\e \esend m} {} \ts$ is the
robust form of $\e \esend  m$ for polymorphic method invocation.  Robust forms
preserve principality, but at the cost of being significantly more cumbersome
to use. Fragile forms relieve this burden, but can only be elaborated into
their robust counterpart if sufficient type information is already available
from the context. Thus, typing fragile constructs ultimately amounts to
identifying when type information is \emph{known}.

\parcomment {The hard part of the problem---defining known type information.}

Intuitively, by known information, we mean typing constraints that
must hold---either from typing rules (\eg application requires the
function to have an arrow type) or programmer supplied type annotations. Yet,
formulating a declarative specification for \emph{when} type information is
known is difficult: most specifications are often twisted with some direct or
indirect algorithmic flavor in order to preserve principality and completeness.

\parcomment {Existing approaches: all rely on a static order of inference}

The two dominant approaches thus far are \emph{bidirectional} typechecking
\citep*{conf/popl/PierceT98} and \emph{\geninst-directional} type inference
\citep*{\BBpolyml}. Each impose some \emph{static} ordering of inference, using
it to propagate inferred types and user-provided annotations as \emph{known}
information.

\parcomment {Static orders have limitations}

While effective in many settings, the rigidity of a static ordering causes even
simple examples whose type could easily be inferred to be rejected. For
instance, \OCaml accepts or rejects the following expression, depending on the
position of the annotation:

\begin{program}[input,angles]
let self_1_1 (x : <f : 'a. 'a -> 'a>) = if true then x#f x else x °
\Ocamlcomment{\ocamlFlag {\OCaml}0}°
let self_1_2 x = if true then x#f x else (x : <f : 'a. 'a -> 'a>) °
\Ocamlcomment{\ocamlFlag {\OCaml}1}°
\end{program}

\parcomment {Our solution: a dynamic order}

We propose \emph{omnidirectional} type inference, a global inference approach
that relies on a \emph{dynamic} order of inference. The solving of inference
constraints may proceed in any order, suspending whenever progress requires
\emph{known} type information. Other constraints may continue to be solved;
once the missing information becomes available (typically via unification), the
suspended typing constraints are resumed.

\parcomment {Omnidirectionality is a framework}

We present omnidirectionality as a general framework for inference in the
presence of fragile features. While this paper instantiates our framework for
two concrete features in \OCaml---static overloading of constructors and record
fields, and polymorphic object methods---its scope is broader: we expect it to
extend to richer features such as GADTs, polymorphic parameters
\citep{\BBpolyparams}, and more general forms of static overloading \ala Swift.

Despite the name, omnidirectionality is not an extension of bidirectional
typechecking. Fragile features are often handled by \emph{local} inference
methods, most notably bidirectional typechecking, because they rely on the
propagation of known type information; omnidirectionality addresses this same
need using \emph{global} inference.

\parcomment {Scaling to ML / let-generalization}

Because our long-term goal is to integrate omnidirectional inference into
\OCaml's typechecker, it must scale to \ML-style polymorphism. While the idea
of suspending constraints is not new (see \cref{sec:related-work}), we show how
suspended constraints can coexist with \ML \emph{local let-generalization}---an
indispensable feature of \OCaml\footnote{In contrast, \Haskell only supports
top-level implicit let-generalization.}---but one that makes suspended
constraints uniquely difficult to implement and specify declaratively.

\subsection* {Contributions}

Section \cref{sec/overview} introduces our setting: \OCaml's static overloading
of datatype constructors and record labels, and polymorphic methods. We review
directional inference, its limitations, and motivate omnidirectional inference.
To this end, we introduce our three key ideas: a new characterization of
\emph{known} type information, \emph{suspended match constraints}, and
\emph{incremental instantiation}, and demonstrate how together they enable
principal type inference for these fragile features. Before turning to
technical developments, we also discuss the \emph{limitations and trade-offs}
of omnidirectionality.

The subsequent sections present our main contributions:
\begin{enumerate}

  \item[(\cref{sec/oml})]

    The \OML calculus, an extension of \ML featuring \OCaml's static
    overloading of record labels and semi-explicit first-class polymorphism
    (\cref{sec/overview/polytypes}).
    We give typing rules with a new declarative
    characterization of \emph{known} type information.

  \item[(\cref{sec:constraints})]

    A novel constraint language for omnidirectional inference, equipped with a
    semantics for suspended constraints.
    We describe the translation of \OML programs to constraints
    representing typing problems, and establish the expected metatheoretic
    properties: soundness, completeness, and principality of inference.

  \item[(\cref{sec:solving})]

    A formal definition of our constraint solver as a series of
    non-deterministic rewriting rules, proved correct with respect to
    the constraint semantics. The rewriting rules detail our treatment
    for the interaction of let-generalization with suspended
    constraints via \emph{incremental instantiation}, and the formal
    description can be directly related to an efficient
    implementation.

  \item[(\cref{sec:implementation})]

    A description of an efficient implementation of our solver, including our
    treatment of suspended constraints and partial type schemes. Validating
    that omnidirectional inference for \ML is practical.

\end{enumerate}
Section \cref{sec/discussion} considers additional practical and peripheral
aspects of our work. Finally, \cref{sec:related-work} compares related work,
and \cref{sec:future-work} concludes with future work. Appendix
\cref{app:full-reference} contains a complete technical reference, collecting
key definitions and figures for convenient lookup. All proofs are deferred to
the appendices.

\section{Overview}
\label{sec/overview}

We ground our work in two fragile features of \OCaml: \emph{static overloading}
of record labels and constructors, and \emph{polymorphic object methods}. Both
are useful in practice: static overloading is widely relied upon in large
programs, and polymorphic methods make first-class polymorphism available
within \OCaml.

\subsection{Static overloading of constructors and record labels}
\label{sec/overview/overloading}

\parcomment{What do we mean by static overloading?}

\emph{Static overloading} denotes a form of overloading in which resolution is
performed entirely at compile time, enabling the compiler to select a unique
implementation without relying on runtime information---in contrast to
\emph{dynamic overloading}, which defers resolution to runtime via mechanisms
such as dictionary-passing or dynamic dispatch. Many mainstream languages, such
as C++, Rust, and Java, use static overloading; its appeal is that it provides
a \emph{zero-cost} abstraction.

\parcomment{\OCaml's static overloading}

\OCaml supports a limited yet useful form of static overloading for record
labels and datatype constructors. Ambiguity is resolved using \emph{known
type information} under its directional inference algorithm (discussed in
\cref{sec/overview/directional}).
To illustrate static overloading in \OCaml, consider two nominal
record types with overlapping field names:
\begin{program}[input]
type point      = { x : int; y : int }
type gray_point = { x : int; y : int; color : int } °
\vadjust {\penalty 10000}°
\end{program}
With both definitions in scope, \OCaml must statically disambiguate each
field usage:
\begin{program}[input,checkocaml]
let one = { x = 42; y = 1337 }                           °\ocamlflags 00°
let ex_1 r = r.x                                         °\ocamlflags 21°
let ex_2 (r : point) = r.x + r.y                         °\ocamlflags 00°
let ex_3 r = (r.x, (r : point).y)                        °\ocamlflags 10°
\end{program}
The type of expression \ocaml!one! has the unambiguous type \code{point},
even though both \code{point} and \code{gray_point} define the fields
\code{x} and \code{y}.
This is because \OCaml performs \emph{closed-world reasoning}: the
typechecker is able to unambiguously
infer the type of \ocaml{one} as \code{point}, since it is the only record type
whose domain is $\set{\mathtt{x}, \mathtt{y}}$. Similarly,
\code{r.color} necessarily infers \code{gray_point} for the type of \code{r}.

\parcomment{Examples}

By contrast, \code{r.x} is ambiguous unless the type of \code{r} is
\emph{known}.
In \ocaml{ex_1}, the type of \code{r} is unconstrained, so disambiguation
fails.\footnote{\let \code \code
  In fact, \OCaml does not fail on ambiguous types, but instead applies a
  default resolution strategy: it emits a warning and selects the last
  matching type definition in scope. Here, this will amount to choosing the
  type \code{gray_point} for \code{r}. 
}\footnote{\let \code \code
  To check all our examples, use the options
  \code{-principal -w +41+18 -warn-error +41+18}, which enables principal
  type inference and escalates the associated warnings to errors.
  \label{fn/principal}
}
%
In \code{ex_2}, the annotation fixes the type of \code{r}, thus \code{r}'s type
is \emph{known} and resolves \code{r.x} and \code{r.y} unambiguously. In
\code{ex_3}, the type of \code{r} can only be \code{point}: considering the
second projection first, we learn that \code{r} must have the type
\code{point}, and since it is $\lambda$-bound, this should make the first
projection unambiguous. However, \OCaml still rejects this example due to its
\emph{static order} of inference (\cref{sec/overview/directional}).

\paragraph{Default rules}
\label{sec/default-rules}

If local type information and closed-world reasoning are insufficient, \OCaml
falls back to a syntactic default: it selects the most recently defined
compatible type. For example, \OCaml accepts the following expression, when
\smashcolorbox{warning}{\strut warnings (\dmark)} are not turned into
{\lower 0pt\hbox{\strut}}\smashcolorbox{illtyped}{\strut errors (\xmark)},
\ie without the \code{-warn-error} flag\footnotecref{fn/principal}:
%
\begin{program}[input,checkocaml=true]
let ex_1 r = r.x                                      °\ocamlflags 21°
\end{program}
The expression is compatible with both \code{point} and \code{gray_point},
since each defines a field \code{x}. But \code{gray_point} is chosen simply
because it appears later in the source.

\parcomment {Default rules are inherently non-principal}

Such fallback behavior is inherently \emph{non-principal}: it reflects the
typechecker's decision to abandon principal inference and arbitrarily select a
syntactic default when no unique type can be inferred. We therefore give no
formal account of such ``default rules''.

\parcomment{Default rules interact poorly with directional inference}

Defaulting also interacts poorly with \OCaml's directional inference. Once the
compiler selects a type, it commits to it---even if that choice causes errors
downstream. Consider \code{ex_3}: when typing \code{r.x}, \OCaml defaults
\code{r} to \code{gray_point}, and subsequently fails on \code{(r : point).y}.
\OML succeeds by suspending the resolution of \code{r.x} until it learns from
\code{(r : point).y} that \code{r} has type \code{point}.

\parcomment{Variants are supported, but not discussed}

Since overloaded datatype constructors are analogous to record fields, we focus
only on record fields in this work. Our prototype implementation
(\cref{sec:implementation}), however, supports both.

\subsection{Polymorphic methods}

\parcomment{Why polymorphic methods?}

Polymorphic methods \citep*{\BBpolyml} bring a form of
System-$\F$-like expressiveness to \OCaml by supporting first-class
polymorphism (impredicative higher-rank polymorphism) while preserving
principal type inference.

\paragraph{From polymorphic methods to \polytypes}
\label{sec/overview/polymethods-reduction}

\parcomment{Idea and syntax}

Polymorphic methods can be translated into ordinary methods that carry a
\emph{\polytype}.
A \polytype is
a \emph{boxed} type scheme $\tpoly \ts$ that
must be explicitly unboxed at use sites. The syntax of \polytypes is given in
\cref{fig:polyml-syntax}.
We write $\epoly[\ts] \e$ to box a term~$\e$ with the scheme $\ts$, yielding
an expression whose type is the \polytype $\tpoly \ts$. Since polytypes are
boxed, they cannot be be freely instantiated, unlike ordinary type schemes. An
expression~$\e$ of a \polytype $\tpoly \ts$ must first be unboxed using the
construct $\einst \e$ before it can be typed at any instance $\t$ of $\ts$. 
Because they are boxed, polytypes can be treated as regular
(mono)types, thereby enabling impredicativity.

\begin{bnffig}[tb]
  {fig:polyml-syntax}
  {Selected syntax for \polytypes \citep*{\BBpolyml}.}
     \entry[Terms]{\e}{
       \epoly[\ts] \e
       \and \einst \e
       \and \eannot \e {} \t
       \and \dots
     }\\
     \entry[Types]{\t}{
        \tpoly \ts
        \and \dots
     }
\end{bnffig}

\parcomment{Example}

Concretely, the polymorphic method of
\ocaml{object method id : 'a. 'a -> 'a = fun x -> x end}
is translated to \ocaml{object method id = [ fun x -> x : 'a. 'a -> 'a ] end}.
Method invocation implicitly unboxes the polytype \eg $\ttlab \x \esend
\ttlab {id}$ becomes $\einst {\ttlab \x \esend \ttlab {id}}$.

\parcomment{Why bother?}

This reduction is useful for two reasons:
\begin{enumerate*}
  \item Inference for \OCaml's object layer is largely governed by row-polymorphism,
    which is \emph{robust} and does not threaten principality; it is therefore
    orthogonal to our concerns. In contrast, \polytypes are \emph{fragile}.

  \item \Polytypes underpin other features in \OCaml, notably the recent
    addition of polymorphic function parameters \citep*{\BBpolyparams}.

\end{enumerate*}

\paragraph{Default rules}
\label{sec/default-rules/poly}


As with static overloading in \cref{sec/default-rules}, \OCaml employs a
default rule for \polytypes: when a polytype cannot be inferred as
\emph{known}, its polymorphic scheme is defaulted to a monomorphic type. 
This is particularly useful, as it allows method calls such as \ocaml{fun bag x -> bag#mem x} 
to typecheck without annotating \ocaml{bag}. Here, the method
invocation $\ttlab {bag} \esend \ttlab {mem}$ introduces an unboxing operation,
becoming $\einst {\texttt{bag} \esend \texttt{mem}}$. Since the type of this
expression is not yet known, the default rule applies, assigning \ocaml{bag}'s
\ocaml{mem} field the \emph{monomorphic} polytype $\tpoly {\tva \to \tvb}$.


Defaulting is not considered in our work.  In practice, many users restrict
their use of objects to method calls on objects of known types. In such
cases, defaulting never arises.
By contrast, defaulting appears essential for polymorphic parameters
\citep*{\BBpolyparams}: they introduce polytype boxing and unboxing at every
function application, with the default corresponding to the common monomorphic
case. We therefore defer them to future work. 

\paragraph{Semi-explicit first-class polymorphism}
\label{sec/overview/polytypes}


\Polytypes, introduced by \citet{\BBpolyml}, allow so-called
\emph{semi-explicit first-class polymorphism}, requiring explicit boxing an
unboxing, but permit instantiations to be implicit. They therefore stand
in contrast to \emph{implicit} higher-rank or first-class polymorphism,
provided by systems such as
\DK~\citep*{dunfield-krishnaswami-bidirectional-poly},
\QuickLook~\citep*{journals/pacmpl/SerranoHJV20}, and \Frost~\citep*{journals/corr/abs-2607-16061},
where polymorphic values are unboxed. At the same time, they are less explicit
than \emph{fully explicit} polymorphism (\eg System-$\F$), where both
quantification and instantiation must be explicitly and fully specified.


We focus on \polytypes in the remainder of this work for two reasons. First,
our work is grounded in the features of \OCaml, where \polytypes are used for
polymorphic methods and polymorphic function parameters. Second, \polytypes
expose the subtle interaction with principality that is of interest. We discuss
approaches for implicit first-class polymorphism in \cref{sec:related-work}.

\subsection{Directional type inference}
\label{sec/overview/directional}

We now discuss the two main directional inference approaches:
\geninst-directional and bidirectional, illustrated using \polytypes as a
running example. We then discuss limitations of both approaches, providing us
with the motivation for omnidirectional type inference.

\paragraph{\Geninst-directional type inference}

\parcomment{\ML (and \geninst) have a fixed order}

Most \ML type inference algorithms proceed in a fixed order when typechecking
let-bindings $\elet \x \ea \eb$: first typecheck the definition $\ea$,
generalize its type, and then typecheck the body $\eb$ under the extended
environment. \Geninst-directionality leverages this ordering to resolve
ambiguous constructs in a \emph{principal} way. The key idea is that this
ordering reveals which types are \emph{known}, namely types that are fixed by
generalization---and therefore stable enough to guide disambiguation. We call
this approach \geninst-directional (read as ``\textbf{pi}-directional'') type
inference, reflecting that \textbf polymorphic expressions must be typed
before their \textbf instances. 

\parcomment{Annotation variables. Why?}

To make this notion precise, we annotate types with \emph{annotation variables}
$\av$. We write $\tannot \t \av$ for a type~$\t$ annotated with the variable
$\av$. Annotation variables record the origins of types and may themselves be
generalized, yielding type schemes such as $\tfor \av \t$, where $\av$ appears
free in $\t$; in particular, we may have a scheme $\tfor \av {\tannot \t \av}$,
where $\av$ annotates the topmost structure of $\t$.
\parcomment{What is a known type?}
A type $\tannot \t \av$ is considered \emph{known} when its (topmost)
annotation variable $\av$ is eligible for generalization (\eg $\tfor \av
{\tannot \t \av}$). Conversely, a type whose topmost annotation variable is
monomorphic---that is, it cannot be generalized in the current context---is
considered \emph{not-yet-known} and cannot be relied on for disambiguation.

\parcomment{\PolyML typing rules}

Using this approach, we can give typing rules for polytypes; this
yields precisely the formulation introduced in
\PolyML~\citep*{\BBpolyml}, whose most important rules are given in
\cref{fig:polyml}. 
\begin{mathparfig}[htpb!]
  {fig:polyml}
  {\Geninst-directional syntax and typing rules for polytypes---excerpt from
  \PolyML~\citep*{\BBpolyml}.} 
\begin{bnfgrammar}[\def\bnfsep{\;}]
     \entry[Types]{\t}{
        \tannot \t \av
        \and \tpoly \ts
        \and \dots
     }\\
     \entry[Type schemes]{\ts}{
       \t \and \tfor \tv \ts \and \tfor \av \ts
     }\\
     \entry[\llap{Annotation variables}]{\av}{}
\end{bnfgrammar}

  \infer[PolyML-Poly]
    { \G \th \e : \tsa \\ \annot \ts \tsa \\ \annot \ts \tsb}
    {\G \th \epoly[\ts] \e : \tapoly {\tsb} \av}

  \infer[PolyML-Inst]
    {\G \th \e : \tfor \av {\tapoly \ts \av}}
    {\G \th \einst \e : \ts}

  \infer[PolyML-Annot]
    {\G \th \e : \ta \\ \annot \t \ta \\ \annot \t \tb}
    {\G \th \eannot \e {} \t : \tb}

  \infer[Annot-Poly]
    {\annot \ts \tsp}
    {\annot {\tpoly \ts} {\tannot {\tpoly \tsp} \av}}
\end{mathparfig}
\parcomment {Boxing}
The introduction form (\Rule{PolyML-Poly}) for \polytypes is a boxing operator
$\epoly[\ts] \e$ with an explicit \emph{closed} polytype annotation%
\footnote{%
  For simplicity, we here restrict \PolyML annotations to closed types. Both
  \citet{\BBpolyml} and our work support a more general form with
  existentially quantified type variables, as commonly used in \OCaml.
  \label{fn/polyml-closed-annotations} 
}
$\ts$ (\ie $\ts$ contains no free type or annotation variables).
The body $\e$ is checked against $\tsa$, a \emph{freshened copy} of $\ts$,
\ie a variant of $\ts$ with fresh annotation variables placed at \polytypes
(\Rule{Annot-Poly}---and other omitted congruence rules). The resulting
expression has type $\tapoly {\tsb} \av$ where $\av$ is an arbitrary
(typically fresh) annotation variable and $\tsb$ is another copy of
$\ts$. Because $\ts$ is supplied by the programmer, the polytype is treated
as known: $\epoly[\ts] \e$ also has the generalized type scheme $\all \av
{\tapoly {\tsb} \av}$. This is by design---the explicit annotation in
$\epoly[\ts] \e$ records that the polytype is known.

\parcomment {Unboxing (principality restriction and \geninst-directionality)}

Conversely, to instantiate a \polytype expression (\Rule{PolyML-Inst}), one must
use an explicit unboxing operator $\einst \e$, which requires no accompanying
type annotation. However, the operator requires~$\e$ to have a \emph{known}
polytype scheme 
of the form $\tfor \av {\tapoly \ts \av}$ and then assigns $\einst \e$ the
type~$\ts$. If, by contrast, $\e$ has the type $\tapoly \ts \av$ for some 
non-generalizable annotation variable $\av$, then $\e$ is considered as a
\emph{not-yet-known} \polytype, and therefore $\einst \e$ is ill-typed. This
restriction enforces principality, preventing instantiation on \emph{guessed}
polytypes.

\parcomment {Annotations}

\Rule{PolyML-Annot} can be used to introduce fresh annotation variables.
Intuitively, $\ta$ and $\tb$ are two \emph{fresh copies} of the same
closed annotation\footnotecref{fn/polyml-closed-annotations} $\t$, preventing
unwanted sharing of annotation variables that could otherwise block
generalization.  For example, $\efun {\x : \tpoly \ts} {\einst \x}$, which is
syntactic sugar for 
$\efun \x {\elet \x {\eannot \x {} {\tpoly \ts}}} {\einst \x}$, is well-typed
because the explicit annotation introduces a fresh annotation variable~$\av$
for $\eannot \x {} {\tpoly \ts}$, which can then be generalized, yielding
$\tfor \av {\tannot {\tpoly \ts} \av}$ for the type of the let-bound variable
$\x$.

\parcomment {Who uses this?}

Building on its introduction in \PolyML, \geninst-directionality has since been
adopted in other systems, notably \MLF~\citep*{LeBotlan-Remy/recasting-mlf},
and in \OCaml for features such as polymorphic object methods and the
overloading of record fields and variant constructors.  

\parcomment{Good example}

We illustrate \geninst-directionality using the following \OCaml examples:%
\footnote{
In \OCaml, these examples can be typechecked by translating
$\epoly[\ts] \e$ to
\Code{object method f} $: \ts = \e$ \Code{end}
and $\einst \e$ to $\e \esend \texttt{f}$. This translation is the inverse of
that discussed in \cref{sec/overview/polymethods-reduction}.
}
\begin{program}[input,angles,checkpolyml]
let pid = [ fun x -> x : 'a. 'a -> 'a ] °\ocamlflags 00°
let ex_5 = let p = pid in <p> °\ocamlflags 00°
let ex_6 = (fun p -> <p>) pid °\ocamlflags 10°
\end{program}
At first glance, \code{ex_5} and \code{ex_6} appear equivalent: both simply
instantiate the \polytype bound to~\code{p}. Yet, \OCaml accepts \code{ex_5} and
rejects \code{ex_6}. This is because the let-binding in \code{ex_5}
assigns \code{p} the type scheme $\tfor \av {\tapoly {\tfor \tv \tv \to
\tv} \av}$, and thus its type is considered \emph{known}---permitting unboxing
(\Rule{PolyML-Inst}).
In \code{ex_6}, by contrast, \code{p} is monomorphic at the point of
instantiation as it is $\lambda$-bound, and unboxing is therefore
forbidden.

To emphasize that this behavior is specification-driven and not an artifact of
\OCaml's inference algorithm, consider two equivalent versions of
\ocaml{ex_6}:\footnote {\let \code \Code
\code{app} and \code{rev_app} are the application
function \Code{fun f x -> f x} and the reverse application function
\Code{fun x f -> f x}, respectively.}
\begin{program}[input,angles,checkpolyml]
let ex_6_2 = app (fun p -> <p>) pid                        °\ocamlflags 10°
let ex_6_3 = rev_app pid (fun p -> <p>)                    °\ocamlflags 20°
\end{program}
While these terms are semantically equivalent, they highlight a potential
hazard: their typability may vary under a directionally biased inference
algorithm, depending on whether the function or argument is typed first.  To
limit such implementation-dependent behavior, \OCaml infers all
subexpressions in an \emph{order-independent} manner until they are let-bound.
Consequently, \OCaml does not make any distinction between
\ocaml[indices]{ex_6}, \ocaml[indices]{ex_6_2}, and \ocaml[indices]{ex_6_3}.

\parcomment{Short comings of pi-directionality}

Treating both examples uniformly is in one sense a strength of
\geninst-directionality, but it also reveals a limitation:
annotability is fragile, in that well-typedness depends on the
\emph{precise} placement of annotations, often forcing the programmer
to introduce annotations that would otherwise be unnecessary.
For instance, the following two terms differ only in the position of
the annotation, yet only \code{self_2_1} is well-typed in
\OCaml---while they are both well-typed in \OML:
\begin{program}[input,angles,checkpolyml]
let self_2_1 x = <(x : [ 'a. 'a -> 'a ])> x °\ocamlflags 00°
let self_2_2 x = <x> (x : [ 'a. 'a -> 'a ]) °\ocamlflags 10°
\end{program}

\paragraph{Bidirectional typechecking}

Bidirectional typechecking is a standard alternative to unification for
propagating type information. It is typically formulated by splitting typing
rules into two modes: \emph{checking mode} ($\G \th \e \Leftarrow \t$), which
typechecks a term $\e$ against a type $\t$ in a given context, and
\emph{inference mode} which infers $\e$'s type from the context alone ($\G \th \e
\Rightarrow \t$).

\parcomment{Assignment of modes}

The type system designer assigns modes---checking or inference---to each
language construct. For instance, one can decide to typecheck function
applications $\eapp \ea \eb$ by first \emph{inferring} that $\ea$ has some
function type $\ta \tarrow \tb$, and then \emph{checking} $\eb$ against
$\ta$ (\Rule{Syn-App}); 
but the opposite, mode-correct choice (\Rule{Chk-App}) is also possible:
\begin{mathpar}
  \infer[Syn-App]
    {\G \th \ea \Rightarrow \ta \to \tb \\ \G \th \eb \Leftarrow \ta}
    {\G \th \eapp \ea \eb \Rightarrow \tb}

  \infer[Chk-App]
    {\G \th \ea \Leftarrow \ta \to \tb \\ \G \th \eb \Rightarrow \ta}
    {\G \th \eapp \ea \eb \Leftarrow \tb}
\end{mathpar}
Following the standard bidirectional recipe of checking introduction forms and
inferring elimination forms---often referred to as the \emph{Pfenning
recipe}~\citep*{conf/popl/DunfieldP04}---the former rule (\Rule{Syn-App}) is
typically preferred. 

\begin{mathparfig}[htpb!]
  {fig:bidir-poly}
  {Bidirectional typing rules for \polytypes.}
  \infer[Syn-Var]
    {\x : \t \in \G}
    {\G \th \x \Rightarrow \t}

  \infer[Chk-Fun]
    {\G, \x : \ta \th \e \Leftarrow \tb}
    {\G \th \efun \x \e \Leftarrow \ta \to \tb}

  \infer[Syn-App]
    {\G \th \ea \Rightarrow \ta \to \tb \\ \G \th \eb \Leftarrow \ta}
    {\G \th \eapp \ea \eb \Rightarrow \tb}

  \infer[Syn-Annot]
    {\G \th \e \Leftarrow \t}
    {\G \th \eannot \e {} \t \Rightarrow \t}

  \infer[Chk-Forall]
    {\G \th \e \Leftarrow \ts \\ \tv \notin \fvs \G}
    {\G \th \e \Leftarrow \tfor \tv \ts}

  \infer[Chk-Inst]
    {\G \th \e \Rightarrow \ts \\ \ts \leq \t}
    {\G \th \e \Leftarrow \t}

  \infer[Chk-Poly]
    {\G \th \e \Leftarrow \ts}
    {\G \th \epoly \e \Leftarrow \tpoly \ts}

  \infer[Syn-Inst]
    {\G \th \e \Rightarrow \tpoly \ts}
    {\G \th \einst \e \Rightarrow \ts}
\end{mathparfig}

\parcomment{Known type information}

Crucially, the assignment of modes determines how type information flows
through the program. In particular, types become available either as inputs to
checking judgments or as outputs of synthesizing judgments. We call such
types \emph{known}. 
Formally, a type $\t$ is \emph{known} when it is either:
\begin{enumerate*}
  \item part of an annotation,
  \item supplied as input to a checking judgment in the conclusion
  ($\G \th \e \Leftarrow \t$), or
  \item produced by a synthesizing premise ($\G \th \e \Rightarrow \t$)
\end{enumerate*}.

\parcomment{\Polytype rules}

Using this notion of known types, we can give a \emph{principal} type system
for \polytypes, shown in \cref{fig:bidir-poly}.\footnote{%
  The rules in \cref{fig:bidir-poly} deliberately make no \emph{guesses}, in
  order to isolate the role of \emph{known} type information. Consequently, the
  system is rather weak. In particular, \Rule{Chk-Inst} permits polymorphic
  instantiation only in checking mode: $\t$ in $\ts \leq \t$ must already be
  known, rather than guessed. 
  This restriction is not intended to stack the deck against bidirectional
  typing: relaxing it would require a bidirectional account of inferred impredicative
  instantiation, which remains an active area of research
  \citep{journals/pacmpl/SerranoHJV20, journals/pacmpl/ParreauxBFC24,
  journals/corr/abs-2607-16061}. The restriction is moreover orthogonal to the
  limitation illustrated later in \cref{example:bidir-poly}.
}
This system eliminates two
artifacts needed in the \geninst-directional presentation: explicit annotations
on boxing and annotation variables. Since $\eannot \e {} \t$ already propagates
known information, $\epoly \e$ requires no attached annotation; and because
``known-ness'' now follows from inference modes rather than polymorphism,
annotation variables are unnecessary. 

\parcomment{Limitations}

While this approach yields a remarkably clean account of \polytypes, it suffers
from the fundamental limitation of bidirectional typing: there is generally no
optimal assignment of modes. For any choice of modes, some programs will
typecheck successfully, while others will fail unnecessarily. Yet, the typing
rules must irrevocably commit to a fixed set of modes, after which, principal
types exist, but only with respect to a specification that made non-principal
choices to begin with. 

For instance, \code{ex_6} is ill-typed under the rules in
\cref{fig:bidir-poly}, as the following example demonstrates.

\begin{example}
\label{example:bidir-poly}
\newcommand{\pid}{\texttt{pid}}\newcommand{\tpid}{\tau_\pid}
Recall that \code{ex_6}, transcribed into the \PolyML calculus, is defined as:
\begin{mathpar}
\texttt{ex}_6 \eqdef
\eapp {\parens{\efun p {\einst p}}}
      {\eannot {\epoly {\efun x x}} {} {\tpoly {\tfor \tva {\tva \to \tva}}}} 

\pid \eqdef
\eannot {\epoly {\efun x x}} {} {\tpoly {\tfor \tva {\tva \to \tva}}} 

\tpid \eqdef
\tfor \tva {\tva \to \tva}
\end{mathpar}
Suppose we try to infer $\tvb \to \tvb$ for \code{ex_6}. Using
\Rule{Syn-App} would require a derivation of the following form:
\begin{mathpar}
\colorlet{ebox}{illtyped}
  \infer*[right=\ebox{Syn-App}]
    {
      \infer*[right=??]
        {\color{red}\text{No rule applies}}
        {\th \efun p {\einst p} \Rightarrow
                {\tpid \to {\tvb \to \tvb}}}
      \\
      \infer*[Right=Chk-Inst]
        {
          \infer*[right=Syn-Annot]
            {\vdots}
            {\th \pid \Rightarrow \tpid}
          \\
          \tpid \leq \tpid
         }
        {\th \pid  \Leftarrow \tpid}
        \hskip -5em
    }
    {\th \eapp {\parens{\efun p {\einst p}}} \pid \Rightarrow \tvb \to \tvb}
   \hskip 5em\null
\end{mathpar}
The first premise cannot be derived: functions are introduction forms and
\Rule{Chk-Fun} only permits them to be \emph{checked}. Although the annotation
in \code{pid} determines its type, \Rule{Syn-App} cannot use this information
to typecheck the function, since the information flows from the function to its
argument rather than conversely. 

Indeed, taking the opposite, equally mode-correct rule \Rule{Chk-App} for
applications (instead of Rule \Rule {Syn-App}, \ie using a dual version
of the system proposed in \cref{fig:bidir-poly}) accepts the program.
With \Rule{Chk-App}, we can check \code{ex_6} against $\tvb \to \tvb$:
\begin{mathpar}
\colorlet{ebox}{welltyped}
\hskip 9em
\infer*[right=\ebox{Chk-App}]
    {\hskip -9em
     \infer*[left=Chk-Inst]
            {
              \infer*[Left=Syn-Inst]
                {
                  \infer*[Left=Syn-Var]
                    { }
                    {p : \tpid \th p \Leftarrow \tpid}
                }
                {p : \tpid \th \einst p \Rightarrow \tfor \tva \tva \to \tva}
                \\
                \tfor \tva \tva \to \tva \leq \tvb \to \tvb
            }
        {\infer*[Left=Chk-Fun]
            {p : \tpid \th \einst p \Leftarrow \tvb \to \tvb}
        {\th \efun p {\einst p} \Leftarrow \tpid \to \tvb \to \tvb}
        }
      \\
      \infer*[Right=Syn-Annot]
        {\vdots}
        {\th \pid \Rightarrow \tpid}
    }
    {\th \eapp {\parens{\efun p {\einst p}}} \pid \Leftarrow \tvb \to \tvb}
\end{mathpar}
Here the annotated argument \code{pid} is inferred first, making its type
$\tpid$ available to check the function.  
Thus \code{ex_6} is not intrinsically difficult to type bidirectionally: it
is rejected only because \cref{fig:bidir-poly} has committed to
\Rule{Syn-App}. For this particular program, that fixed choice of  
modes causes a needless failure.
\end{example}

\paragraph{Limitations of directional type inference}

Bidirectional typechecking is lightweight, practical, and well-suited for
complex language features such as higher-rank polymorphism, dependent types,
or subtyping. It supports the propagation of type information with minimal
annotations. Its main downside lies in the need to fix an often arbitrary
flow of type information---as in the case of function applications discussed
above.

On the other hand, \emph{\geninst-directional} type inference appears better
suited for \ML: %
\begin{enumerate*}
    \item thanks to its use of polymorphism\penalty 50---the essence of \ML; and
    \item its ability to be retrofitted easily onto existing typecheckers.
\end{enumerate*}
But it remains surprisingly weak in some cases: it does not even allow the
propagation of user-provided type annotations from a function to its argument!
This weakness is sometimes counter-intuitive to the user. For example, the
following would be rejected as ambiguous using \geninst-directional type
inference alone:
\begin{program}[input,angles,checkpolyml]
let ex_7 = °\ocamlflags 00°
  let g (f : [ 'a. 'a -> 'a ] -> int) = f pid in
  g (fun p -> <p> 42)
\end{program}
Here, \code{p} is $\lambda$-bound and therefore monomorphic. Without further
propagation, the term \ocaml[angles]{<p>} would be ambiguous, as no polymorphic (and
thus \emph{known}) type can be ascribed to \code{p}. \OCaml resolves this by
supplementing \geninst-directional inference with a form of bidirectional
propagation: the expected type of \code{g}'s parameter (\code{[ 'a. 'a -> 'a ] -> int}) is
bidirectionally propagated to the application of~\code{g}, assigning
\code{p} to have the \emph{known} type \code{[ 'a. 'a -> 'a ]} and thereby
disambiguating \ocaml[angles]{<p>}.

\subsection{Omnidirectional type inference}
\label{sec/overview/omni}

%
%
%

\parcomment{Idea: inference has a dynamic order}

Omnidirectional inference infers typing constraints in any order. Constraints
advance \emph{dynamically}; those that require \emph{known} type information
suspend, and resume when other constraints supply it. This stands in contrast
to the fixed \emph{static} order of bidirectional and \geninst-directional
inference.

\parcomment{Example: dynamic order types more programs}
Consider again \code{ex_6_2} and \code{ex_6_3} from
\cref{sec/overview/directional}: 
\begin{program}[input,angles]
let ex_6_2 = app (fun p -> <p>) pid     °\ocamlflags 10°
let ex_6_3 = rev_app pid (fun p -> <p>) °\ocamlflags 00°
\end{program}
Under \geninst-directionality, both terms are ill-typed; under \OCaml's
bidirectional approach,\footnote {In this example, we use the option
\code{-no-principal} to enable bidirectional propagation in 
addition to \geninst-directionality} only
\code{ex_6_2} is rejected. Yet both terms have a principal type---it is
merely a question of propagating type information in the right
order. Omnidirectional type inference typechecks both: since it allows the
typing of either side to proceed first, suspending the other until the
relevant type information is \emph{known}.

\paragraph{To be or not to be known}
\label{sec/overview/omni/unicity}

\parcomment{Idea: known = contextual uniqueness}

That is the question, indeed. To specify omnidirectionality declaratively, we
must say when a type is considered \emph{known} without relying on any fixed
directional order. Our key idea is that a type is \emph{known} when it is the
\emph{unique} type that can be inferred within some surrounding term context
$\E$.

\parcomment{Example: ex_62}

Take \code{ex_6_2} as an example, using the following term context where
$\hole$ denotes the hole:
\begin{program}[input, mathescape=true]
let pid = [ fun x -> x : 'a. 'a -> 'a ]
let ex_6_2 = app (fun p -> $\hole$) pid
\end{program}
Here, \code{p} has a uniquely inferrable type $\tpoly {\tfor \tv
{\tv \to \tv}}$, no other type can be inferred for \code{p}. As a result, we
can consider \code{p}'s type \emph{known}.

\parcomment{Elaboration}

\begin{wraphbox}{}{}
\begin{mathpar}[inline]
  \infer[Use-I]
    {\eshape \E \e {\tpoly \ts} \\ \G \th \E\where{\exinst \e {} \ts} : \t}
    {\G \th \E\where{\einst \e} : \t}
\end{mathpar}
\end{wraphbox}
Once the type is known, the fragile implicit term can be elaborated to a
robust explicit counterpart. For example, the unboxing \ocaml[angles]{<p>} can
be elaborated into the explicitly annotated form \ocaml[angles]{<p : 'a. 'a -> 'a>}.
Consequently, to typecheck $\einst \e$ under the context $\E$, it suffices to
assert that $\e$ has the known polytype $\tpoly \ts$ and elaborate
$\einst \e$ into an annotated unboxing $\exinst \e {} \ts$, as captured by
\Rule{Use-I}.
The predicate $\eshape \E \e {\tpoly \ts}$---our \emph{unicity
condition}---formalizes precisely what it means for a type to be \emph{known}.
We defer its technical definition, which is rather subtle, to \cref{sec/oml/typing/I}.

\paragraph{Suspension in action}
\label{sec/overview/omni/match}

\parcomment{Suspended constraints}

Suspension is the mechanism that allows inference to proceed in any order,
in spite of constructs that require \emph{known} type information.  In our
framework, we realize this through our novel primitive: \emph{suspended
match constraints}.

\parcomment{Syntax + informal semantics}

A match constraint $(\cmatch \t {{\overline{\cbranch \cpat \c}}})$ pairs a
(typically unknown) matchee type $\t$ with a finite series of shape-pattern
branches $\overline{\cbranch \cpat \c}$. Such constraints remain
\emph{suspended} until the \textit{shape} of~$\t$ (\ie its top-level
constructor) is known. Then, they are \emph{discharged}: a unique branch is
selected and its associated constraint has to be solved. A match constraint
that is never discharged is considered unsatisfiable.

\parcomment{A note on shape patterns}

For now, it suffices to think of shapes as parts of types (\eg the record type constructor 
$\T$ in ($\trcd \T \tys$)%
\footnote{
  For readability in shape patterns, type constructors are written in 
  prefix form and record type constructors are prefixed with
  \textsf{rcd}. This prefix is omitted in \OCaml code,  
  where type constructors appear in postfix position.
})
while shape patterns $\cpat$ act as type `destructors', binding parts of the type
(\eg the constructor name $\T$) to meta-variables (\eg the pattern variable
$\ct$) used in $\cs$. This will be made precise in
\cref{sec/constraints/shapes}.

\parcomment{Why our solution makes things easier?}

We now illustrate the role of suspended constraints on our running
\emph{fragile} features: static overloading of records (and variants) and
semi-explicit first-class polymorphism.
Each feature translates the typability of the term into constraints, formalized
using a constraint generation function of the form
$\cinfer \e \t$, which,
given a term $\e$ and expected type $\t$, produces a constraint $\c$ which is
satisfiable if and only if $\e$ has the type $\t$.
As we will see, once we adopt the suspended constraint machinery developed in
this paper, much of the complexity of these typing fragile constructs
vanishes---suspended constraints do most of the heavy lifting.

\parcomment{Records}

For an ambiguous record projection $\efield \e \elab$,
we generate the typing constraint:
\begin{mathpar}
\cinfer {\efield \e \elab} \t \wide\eqdef
  \cexists \tv \cinfer \e \tv
  \cand
  \cmatch \tv
      {{\cbranch {\cpatrcd \ct}
  {{\labfrom \elab \ct \leq (\tlab \tv \t)}}}
      }
\end{mathpar}
This constraint introduces the unification variable $\tv$, unifying it
with the type of $\e$ (via $\cinfer \e \tv$), and suspends resolution
of the return type $\t$ until the type $\tv$ of $\e$ becomes
\emph{known} to be some non-variable type $\tz$. The branch then
matches $\tz$ against the record type pattern $(\cpatrcd \ct)$. If $\tz$
is a record type $(\trcd {\T} \tyas)$,
then the pattern binds the record name variable $\ct$ to the record name $\T$,
otherwise the whole constraint fails. When $\ct$ is bound to $\T$, the
right-hand-side constraint becomes
$\labfrom \elab \T \leq (\tlab \tv \t)$, which requires that the type
of the projection of the label $\elab$ at type $\T$ in the global
record-declaration environment can be instantiated into
$(\tlab \tv \t)$.

\parcomment{\Polytypes}

When typechecking the polytype unboxing operator $\einst \e$, if $\e$ is
already known to have the type $\tpoly \ts$, then we can simply
instantiate $\ts$.  However, if the type of $\e$ is not yet known---\ie  it is a
(possibly constrained) type variable $\tv$---then we must defer until more
information is available. We capture this behavior with a suspended match
constraint:
\begin{mathpar}
\cinfer {\einst \e} \t \Wide\eqdef
    \cexists \tv \cinfer \e \tv
\cand
    \cmatch  \tv \cbranch {\tpoly \cscm} \cscm \leq \t
\end{mathpar}
The match remains suspended until $\tv$ resolves to some type $\tz$. If, upon
resolution, $\tz$ is $\tpoly \ts$, the pattern $\tpoly \cscm$ matches
successfully, binding the polytype variable $\cscm$ to $\ts$ and performing the
instantiation $\cleq \ts \t$. Otherwise, the
pattern does not match and the constraint fails.

\paragraph{Scaling to \ML}

\parcomment{It's easy without polymorphism}

In the absence of (implicit) polymorphism, type inference is solely based on
unification constraints which can be solved in any order; omnidirectional
inference with suspended match constraints is then natural and easy to
implement.

\parcomment{\ML generalization has a fixed order}

The difficulty originates from \ML \emph{implicit} \texttt{let}-polymorphism for which
all known implementations follow the \geninst-order: first typing the binding,
generalizing it into a type scheme, and finally typing the body under the
extended typing environment that binds the generalized scheme. The
Hindley-Milner algorithm $\mathcal{J}$, its variants $\mathcal{W}$ or
$\mathcal{M}$~\citep* {Lee_Yi/algoM@toplas1998}, or more flexible
constraint-based type inference implementations~\citep*
{Remy/mleth,Remy/thesis, Odersky-Sulzmann-Wehr@tpos, \BBemlti} all
follow this strategy, to the best of our knowledge.

\parcomment{Example of why this problematic}

Consider the following program:
\begin{program}[input]
type 'a gpoint = { x : 'a; y : 'a }                     
let diag (n : 'a) : 'a gpoint = { x = n; y = n }
°\halfline°
let ex_8 gp =                                           °\ocamlflags 20°
  let getx p = p.x in getx (diag 42), (getx gp : float)
\end{program}
We introduce a new parameterized record type \code{'a gpoint}, whose
fields \code{x} and \code{y} are overloaded---recall that both \code{point} and
\code{gray_point} already define these fields---but here they have a
\emph{polymorphic} projection type $\tfor \tv {}$ \code{'a gpoint -> 'a}. The
function \code{diag} constructs a diagonal point, a point lying on the
diagonal $x = y$, and has the type $\tfor \tv {}$ \code{'a -> 'a gpoint}.

When typechecking \code{ex_8}, we cannot infer the type of \code{getx} first,
since the type of \code{p.x} is still not yet known. Instead, we must typecheck
the body first, where the call \code{getx (diag 42)} reveals that \code{p} has
type \code{'a gpoint}.
Failing to do so would make inference incomplete, since the program is
clearly well-typed. Nor can we treat the let-binding as monomorphic, since
both calls to \code{getx} use different instantiations of \code{'a}: indeed,
$\tv$ is $\tint$ in \code{getx (diag 42)} while $\tv$ is $\tfloat$ in
\code{(getx gp : float)}.

\parcomment{Our solution}

We solve this by introducing \emph{incremental instantiation}, \ie the ability
to instantiate type schemes that are not yet fully determined (so-called
\emph{partial type schemes}) and consequently revisit their instances when they
are being refined, \emph{incrementally}. This allows inferring parts of a
let-body to disambiguate its definition, without duplicating
constraint-solving work.


\paragraph {The forest, not the trees.}

Suspended match constraints offer a \emph{general framework} to typing the
features we have considered so far, and more. Some of those features can be
handled using more specialized approaches: for example, \SML employs row
variables to support overloaded fields for structural records, while \GHC uses
qualified types to allow overloading of nominal record fields with a
simple-enough type.


In contrast to these specialized mechanisms, suspended constraints are more
expressive. They can handle cases where the typing rule to use on the subterms
depends on the outcome of disambiguation, such as overloaded polymorphic record
fields or overloaded GADT constructors in patterns.
Moreover, these simpler approaches typically lack a declarative semantics that
justify rejecting programs with unresolved disambiguation choices.

\subsection{Limitations}
\label{sec/overview/limitations}


\parcomment{Some programs even though they have a principal type are ill-typed}

Omnidirectional inference is powerful, but not omnipotent:
\begin{enumerate*}

  \item some programs are still rejected as ambiguous, even though a unique
    elaboration could in principle be chosen---a deliberate ``Goldilocks''
    compromise;

  \item our current formalization cannot disambiguate based on the \emph{return
    type} of overloaded projections;

  \item it presently omits a formalization of \emph{default rules}
    (\cref{sec/default-rules}); and

  \item our framework entails a higher conceptual and implementation complexity
    than static directional approaches (\eg bidirectional typechecking).

\end{enumerate*}

\paragraph{Not too hot, not too cold} Some expressions
must be rejected even though their elaboration would be unambiguous. Consider:
\begin{program}[input,checkocaml]
type cie_color = { x : int; y : int; z : int }
type cie_point = { x : int; y : int; color : cie_color }
let ex_1_0 r = r.color.x °\ocamlflags 11°
\end{program}
Neither field projections in \code{ex_1_0} can individually be
disambiguated. However, if one were allowed to combine the constraints, they
would jointly determine that \code{r} must have type \code{cie_point}: in
\code{gray_point}, the field \code{color} has type \code{int}, and hence
cannot itself be projected, leaving a unique consistent elaboration:
\code{r.cie_point.color.cie_color.x}.%
\footnote{
The syntax $\exfield \e \T \elab$ qualifies the label $\elab$ to unambiguously
belong to the type $\T$ in the projection.
}

Our framework nonetheless rejects \code{ex_1_0}. This is intentional:
overloaded projections must be elaborated \emph{sequentially}, each in isolation,
rather than jointly with others. We view this restriction as a
``Goldilocks'' compromise: it rules out examples like the above, but avoids
the intractability of full general overloading which is NP-hard, even
without let-polymorphism, as shown by a reduction from 3-SAT~\citep*
{Chargueraud-Bodin-Dunfield-Riboulet/jfla2025}.

\paragraph{No returns accepted}
At present, our system also resolves overloaded projections without taking their
\emph{return type} into account.
\begin{program}[input,checkocaml]
let ex_1_1 r = (r.color : cie_color) °\ocamlflags 11°
\end{program}
Under the rules defined in this work, programs such as \code{ex_1_1} are
rejected. This limitation is shared with existing languages such as \OCaml,
\Haskell, and \SML, none of which use the return type of a projection to guide
disambiguation. Our framework already goes beyond these systems, but we believe
that omnidirectionality provides the right foundation to incorporate return
type disambiguation. We leave this refinement to future work.

\paragraph{No defaults, by default}

\parcomment{No default rules formalism}

We do not yet provide a formal account of \emph{default rules} mentioned in
\cref{sec/default-rules} and \cref{sec/default-rules/poly}.  However, our
prototype implementation, discussed in \cref{sec:implementation}, does support
(optionally) attaching a default strategy to each suspended constraint.
\Draft{More details can be found in Appendix \cref {app/default-rules}.}{} In
practice, defaulting proves useful and appears essential for richer features
such as implicit first-class polymorphism.
%
Developing a formal treatment of defaulting within our framework is an
important direction for future work.

\paragraph{On complexity budgets}

Omnidirectional type inference is conceptually straightforward but technically
challenging. It follows a simple key idea: solving constraints in any order,
suspending when known type information is required. However, realizing this
idea precisely and efficiently comes at a higher complexity cost than
bidirectional or $\Geninst$-directional type inference.

\begin{enumerate}
\item[(\cref{sec/oml})]
  Giving a declarative characterization of \emph{known} type information
  without statically relying on directionality is hard. Our contextual rules
  nicely solve this problem, but their meta-theory is unsurprisingly more
  complex than local rules.

\item[(\cref{sec:implementation})]

  Implementing efficient incremental instantiation is harder than \ML
    instantiation, as refinements of partial type schemes must trigger
    re-instantiations while avoiding solving the same constraints 
    repeatedly across successive instantiations.

\end{enumerate}
We hope to pay off some of this complexity in future work; in particular, we
believe omnidirectionality is the missing piece to unlock modular
implicits~\citep*{White-Bour-Yallop/Modular_Implicits/ml2014}---an approach to
generalized overloading and a long-anticipated feature within the \OCaml
community.

\section{The \OML calculus}
\label{sec/oml}

\begin{mathparfig}[t]
  {fig/oml/syntax-xtyping}
  {Syntax and explicit, robust typing rules of \OML.}
  \begin{bnfgrammar}
  \entryset[Type variables]{\tva, \tvb, \tvc}{\TyVars}{}
  \\
  \entry[Types]{\t}{
      \tv \and
      \tunit \and
      \ta \to \tb \and
      \trcd \T \tys \and
      \tpoly \ts
  }
  \\
  \entry[Type schemes]{\ts}{
      \t \and
      \all \tv \ts
  }
  \\
  \entry[Record name]{\T}{}{}
  \\
  \entryset[Ground types]{\gt}{\Ground}
  \\[1ex]
  \entry[Terms]{\e}{
    x \and
    () \and
    \efun x e \and
    \eapp \ea \eb \and
    \elet x \ea \eb \and
    \eannot \e \tvs \t \andcr
     \epoly e \and
     \expoly e \tvs \ts \and
     \einst e \and
     \exinst e \tvs \ts \andcr
    \erecord {\overline{\elab = \e} } \and
    \efield e \elab \and
    \exrecord \T {\overline{\elab = \e}} \and
    \exfield e \T \elab
     }
  \\[1ex]
  \entry[Contexts]{\G}{
     \eset \and
     \G, x : \ts
  }
  \\
  \entry[Label contexts]{\labenv}{
    \eset \and
    \labenv, \labfrom \elab \T : \all \tvs \trcd \T \tvs \to \t
  }
  \\[1ex]
  \entry[Shapes] {\Sh} {\any \tvcs \t \qquad\;\;\ (\Sh \in \Shapes)}
  \\
  \entryset[Canonical principal shapes] {\sh} {\CanonicalShapes \qquad (\CanonicalShapes \subset \Shapes)}
  \end{bnfgrammar}
  \par
  \inferrule[Var]
    {x : \sigma \in \G}
    {\G \th x : \sigma}

  \inferrule[Fun]
    {\G, x : \ta \th e : \tb }
    {\G \th \efun x e : \ta \to \tb}

  \inferrule[App]
    {\G \th \ea : \ta \to \tb \\
     \G \th \eb : \ta}
    {\G \th \eapp \ea \eb : \tb}

  \inferrule[Unit]
    { }
    {\G \th () : 1}

  \inferrule[Gen]
    {\G \th e : \sigma \\ \tv \disjoint \fvs \G}
    {\G \th e : \tfor \tv \sigma}

  \inferrule[Inst]
    {\G \th e : \tfor \tv \ts}
    {\G \th e : \ts \where{\tv \is \t}}

  \inferrule[Let]
    {\G \th \ea : \sigma \\
     \G, x : \sigma \th \eb : \t}
    {\G \th \elet x \ea \eb : \t}

  \inferrule[Annot]
    {\G \th e : \t\where {\tvs \is \tys}}
    {\G \th (e : \exi \tvs \t) : \t\where {\tvs \is \tys}}

  \inferrule [Poly-X]
    {\G \th \e : \ts\where {\tvs \is \tys}}
    {\G \th \expoly \e \tvs \ts : \tpoly {\ts \where {\tvs \is \tys}}}

  \inferrule [Use-X]
    {\G \th \e : \tpoly \ts \where {\tvs \is \tys}}
    {\G \th \exinst e \tvs \ts : \ts \where {\tvs \is \tys}}

  \inferrule[Rcd-X]
    {\dom {(\Labenv[\T])} = \elabs \\
     \parens{{\labfrom \elabi \T} \leq \t \to \ti}\iton \\
     \parens{\G \th \ei : \ti}\iton }
    {\G \th \exrecord \T {\elaba = \ea; \ldots; \elab_n = \en} : \t}

  \inferrule[Rcd-Closed]
    {\labsuni \elabs \T \\
     \G \th \exrecord \T {\overline{\elab = \e}} : \t }
    {\G \th \erecord {\overline{\elab = \e}} : \t}

  \inferrule[Rcd-Proj-X]
    {{\labfrom \elab \T} \leq \tya \to \tyb \\ \G \th \e : \tya }
    {\G \th \exfield \e \T \elab : \tyb}

  \inferrule[Rcd-Proj-Closed]
    {\labuni \elab \T \\
     \G \th \exfield \e \T \elab : \t }
    {\G \th \efield \e \elab : \t}

  \inferrule[Lab-Inst]
    {\labenv(\labfrom \elab\T) = \tfor \tvs {\trcd \T \tvs} \to \t}
    {{\labfrom \elab \T} \leq (\trcd \T \tys \to \t\where{\tvs \is \tys})}
\end{mathparfig}

\parcomment {Running example: \polytypes}

\parcomment {We need a spec, but this itself is hard}

To prove correctness of type inference, we must define a language and
its type system. Identifying an appropriate declarative type system to use as a
specification is itself a challenging problem. In particular, natural
specifications for fragile features often fail to preserve principality.

\parcomment {Why do natural approaches not guarantee principal types.}

Consider records, for instance. We can ask the user to provide a type
annotation by using an explicit record projection $\exfield \e \T \elab$, which
has a simple typing rule (\Rule{Rcd-Proj-X} in \cref{fig/oml/syntax-xtyping}).
By contrast, the natural typing rule for the overloaded projections $\efield \e
\elab$ (\Rule{Rcd-Proj-I-Nat}) breaks principality.
\begin{mathpar}
  \inferrule*[right=Rcd-Proj-I-Nat]
    {\labfrom \elab \T \leq \ta \to \tb \\ \G \th \e : \ta}
    {\G \th \efield \e \elab : \tb}
\end{mathpar}
For example, term $\efun r
\efield r {\ttlab x}$ admits two incompatible types, \code{point -> int} and
\code{gray_point -> int}, as explained in \cref{sec/overview}.

\parcomment{How we restore principal types (briefly)}

To restore principality for overloaded projections, we require that the type of
$\e$ be \emph{known}---that is, determined to be a specific record type $\trcd
\T \tys$, rather than merely \emph{guessed}. As discussed earlier
(\cref{sec/overview/omni/unicity}), this requirement is expressed via a
\emph{unicity condition}, giving rise to the contextual rule \Rule{Rcd-Proj-I}
(\cref{fig/oml/typing/I}), which elaborates the fragile overloaded projection
$\efield \e \elab$ into its robust explicit counterpart $\exfield \e \T \elab$.

\subsection{Syntax}

\OML (\cref{fig/oml/syntax-xtyping}) extends \ML with two fragile constructs:
\polytypes and nominal records.  Variants are not treated formally in \OML, but
behave analogously to records.

\paragraph{Notation for collections}

We write $\overline X$ for a (possibly empty) set of elements $\set {X_1,
\ldots, X_n}$ and a (possibly empty) sequence $X_1, \ldots, X_n$. The
interpretation of whether $\overline X$ is a set or a sequence is often
implicit. We write $\overline{X} \disjoint \overline{X'}$ as a shorthand for
when $\overline X \cap \overline {X'} = \emptyset$. We write $\overline X,
\overline {X'}$ as the union or concatenation (depending on the interpretation)
of $\overline X$ and $\overline X'$. We often write $X$ for the singleton set
(or sequence).

\paragraph{Types}

Monotypes (or just types) include, as usual, type variables $\tv$, the unit
type $\tunit$, arrow types, but also nominal record types $\trcd \T \tys$, and
\polytypes $\tpoly \ts$. We use a non-standard syntax $\trcd \T \tys$ for record
types, where $\T$ is a \emph{record name} and $\tys$ a list of type
parameters.%
\footnote{
  Using $\T$ as an argument of $(\trcd \T \tys)$ rather a
  head constructor $(\Fapp[\T] \tys)$ will make record
  shape patterns easier to understand.
}
Type schemes $\ts$ are of the form $\all \tvs \t$, they are equal up to the
reordering of binders and removal of useless variables. Standard notions such
as the set of free variables $\fvs \ts$ and capture-avoiding substitutions
$\ts\where{\tv \is \t}$ are defined in the usual way. We write $\TyVars$ for
the set of type variables, and use $\tvs \disjoint \ts$ as a short-hand for
$\tvs \disjoint \fvs \ts$. Finally, $\gt \in \Ground$ denotes a \emph{ground
type}---a type with no free variables.\footnote{$\gt$ can be pronounced
``Fraktur g'' or just ``ground g''.}

\paragraph{Terms}

Terms of \OML are variables~$\x$, the unit literal $\eunit$,
lambda-abstractions $\efun \x \e$, applications $\eapp \ea \eb$, annotations
$\eannot \e \tvs \t$, and let-bindings $\elet x \ea \eb$, extended with the
following expressions:
\begin{itemize}

\item
  For \polytypes, we introduce implicit and explicit boxing and unboxing
  forms: $\epoly \e$, $\expoly \e \tvs \ts$, and $\einst \e$, $\exinst \e
  \tvs \ts$ respectively.

\item
  Overloaded record labels include record literals $\erecord { \elaba = \ea;
  \ldots; \elab_n = \en }$ and field projections $\efield \e \elab$. Both
  constructs have explicit counterparts: $\exrecord \T {\elaba = \ea;
  \ldots; \elab_n = \en }$ and $\exfield \e \T \elab$, where the explicit
    record annotation $\labfrom {\_} \T$ indicates that the labels unambiguously
  belong to the record type $\T$.

\end{itemize}
All annotations in \OML are closed \ie their existentially quantified
variables $\tvs$ are exactly the free variables of the type $\t$ or type
scheme $\ts$.
We use $\e^i$ to range over fragile, implicit terms (\eg $\efield \e \elab$) and
$\e^x$ for their explicit counterparts (\eg $\exfield \e \T \elab$).

\parcomment {Continuation onto next sections}

Typing rules for explicit terms are mostly standard; nominal
records require a more intricate (yet largely folklore) treatment of
\emph{closed world} reasoning.
The crux of our work is the novel typing of the fragile constructs,
presented in~\cref {sec/oml/typing/I}.

\subsection{Typing rules for robust, explicit constructs}
\label{sec/oml/typing/X}

\parcomment {Simple typing rules explained}

As usual, the main typing judgment $\G \th \e : \ts$
(\cref{fig/oml/syntax-xtyping}) states that in context~$\G$, the expression~$\e$
has the type (scheme) $\ts$.
Rules \Rule{Var} through \Rule{Let} are standard. Annotations $\eannot \e \tvs
\t$ (\Rule{Annot}) ensures that the type of $\e$ is (an instance of) the type
$\t$. The type variables $\tvs$ are \emph{flexibly} (or existentially) bound in
$\t$, meaning they may be instantiated to some types $\tys$ so that the
resulting annotation matches the type of $\e$. For instance, the term $(\efun
\x \x + 1 : \exi \tv \tv \to \tv)$ is well-typed under \Rule{Annot} with the
substitution $\where {\tv \is \tint}$.

\paragraph {Explicit \polytypes}

\Rule{Poly-X} serves as the introduction rule: given the (closed) type
scheme~$\ts$, it forms a first-class \polytype $\tpoly \ts$, requiring the
expression
$\e$ to be at least as polymorphic as $\ts$. \Rule {Use-X} is the
corresponding
elimination rule, unpacking an expression of \polytype $\tpoly \ts$ into
one of polymorphic type $\ts$,
which may be freely instantiated (via \Rule{Inst}). Both rules also allow
polytype annotations to be partial, \ie $\ts$ may have free type variables
$\tvs$, which are existentially quantified to close the annotation, as in
\Rule{Annot}.

\paragraph{Explicit nominal records}
\label {app/typing/X/records}

\parcomment {Label types and contexts}

We assume a global label context $\labenv$ mapping labels to their projection
type, \ie type schemes of the form $\all \tvs \trcd \T \tvs \to \t$. A label
$\elab$ may belong to multiple record types, but is unique within each record
type $\T$. We write $\Labenv$ for the set of labels belonging to the record type $\T$:
${\Labenv} \eqdef \set {\elab : \labfrom \elab \T : \all \tvs \trcd \T \tvs \to \t \in \labenv}$.
We write $\labenv (\labfrom \elab \T)$ for the unique scheme
$\tfor \tvs {\trcd \T \tvs} \to \t$ associated with $\elab$ in $\T$ (if defined).

\parcomment{Typing rules}

Label instantiations are typed by an auxiliary judgment ${\labfrom \elab \T}
\leq \tya \to \tyb$ defined by \Rule{Lab-Inst} and
meaning that $\ta \to \tb$ is an instance of the projection scheme
$\labenv (\labfrom \ell \T)$. Explicit field projections
(\Rule{Rcd-Proj-X}) require that $\exfield \e \T \elab$ projects from a record
$\e$ of type $\tya$ to $\tyb$, provided ${\labfrom \elab \T}
\leq \tya \to \tyb$ holds.
Explicit records (\Rule {Rcd-X}) are typed similarly, checking that each field
has the appropriate type. In addition, the premise asserts that the fields
$\elabs$ appearing in the record expression exactly match the labels of $\T$
(\ie $\dom {(\Labenv)}$).

\parcomment {Closed world reasoning}

Following \OCaml, our explicit system also supports \emph{closed-world}
reasoning, which exploits the absence of ambiguity in the label context
$\labenv$ to infer record annotations.
In particular, in a record expression $\erecord {\elaba = \ea; \ldots; \elab_n
= \en}$, if the set of labels ${ \elaba, \ldots, \elab_n }$ uniquely identifies
a record type $\T$ in the context~$\labenv$, then the record has the
type of $\exrecord \T {\elaba = \ea; \ldots; \elab_n = \en}$ (via
\Rule{Rcd-Closed}).
Similarly, if the label $\elab$ is associated with exactly one record type $\T$
in $\labenv$, then the projection $\efield \e \elab$ has the type of $\exfield
\e \T \elab$ (by \Rule{Rcd-Proj-Closed}).

\parcomment {Closed sets}
These two forms of label uniqueness differ. A \emph{closed set} of labels
may uniquely identify a record type even if no individual label is unique.
Conversely, a unique label implies uniqueness of every closed set containing it.
For instance, recall \code{one} from \cref{sec/overview/overloading}:
\begin{program}[input]
type point  = { x : int; y : int }
type gray_point = { x : int; y : int; color : int }
let one = { x = 42; y = 1337 }           °\ocamlflags 00°
\end{program}
Here, the closed set $\set {\text{\code{x}}, \text{\code{y}}}$ in \code{one}
uniquely identifies \code{point}, even though the individual labels \code{x}
and \code{y} also appear in \code{gray_point}.

\parcomment {Formalization of closed-world reasoning}

We formalize this \emph{closed world uniqueness} using the predicates $\labuni \elab \T$ (a label
uniquely identifies $\T$) and $\labsuni \elabs \T$ (a closed label set uniquely
identifies $\T$):
\begin{mathpar}
  \begin{tabular}{C;C;L;C;L}
    \labuni \elab \T &\eqdef& \elab \in \dom {(\Labenv)} &\wedge& \forall \Tp, \; \parens {\elab \in \dom {(\Labenv[\Tp])} \implies {\T} = \Tp} \\
    \labsuni \elabs \T &\eqdef& \dom {(\Labenv)} = \elabs &\wedge& \forall \Tp, \;\parens {\dom {(\Labenv[\Tp])} = \elabs \implies {\T} = \Tp}
  \end{tabular}
\end{mathpar}
These predicates depend only on the global label environment $\labenv$:
they ignore field types and require no contextual type information.
The associated typing rules (\Rule{Rcd-Closed}, \Rule{Rcd-Proj-Closed}) are
therefore \emph{robust}, since disambiguation relies solely on globally known
label information rather than type-directed disambiguation.

\subsection{Shapes}
\label{sec/constraints/shapes}

\parcomment {Why shapes?}

We introduce \emph{shapes} as a generalization of type constructors.
They provide a uniform treatment of both
constructors and \polytypes, and are useful in defining polytype
unification (\Cref{sec:implementation}).

\parcomment {Definition of shape}

A shape $\Sh$ is a type with holes, written $\any \tvcs \t$, where $\tvcs$
denotes the set of type variables representing the holes.  By construction, we
require $\tvcs$ to be \emph{exactly} the free variables of $\t$.  Hence, shapes
are closed and do not contain useless binders.  We consider shapes equal up to
$\alpha$-conversion.  When $\t$ is a ground type, we omit the binder and
simply write $\t$ for the shape.
We use $\bot$ to denote the shape $\any \tvc \tvc$, which we call the
\emph{trivial} shape. Let $\Shapes$ denote the set of all shapes and
$\Shapesz \subset \Shapes$ the set of non-trivial shapes.

\parcomment {Instantiation order of shapes}

\begin{wraphbox}{-1}{}
\begin{mathpar}[inline]
  \infer[Inst-Shape]
    { }
    {\any {\tvcs_1} \t \preceq
     \any {\tvcs_2} \t \where {\tvcs_1 \is \tys_1}}
\end{mathpar}
\end{wraphbox}
Shapes are equipped with the standard instantiation ordering,
defined by \Rule{Inst-Shape}.
When writing $\Sh \preceq \Shp$, we say that $\Sh$ is more general
than~$\Shp$. When $\Sh$ and $\Shp$ are more general than one another, they
are actually equal. The trivial shape $\bot$ is the most general shape.
If $\Sh$ is $\any \tvcs \t$, the shape application $\shapp[\Sh] \tys$ is
defined as $\t \where {\tvcs \is \tys}$. We say that $\Sh$ is a shape of
$\t$ when there exists $\tys$ such that $\t = \shapp[\Sh] \tys$; in this
case, we call the pair $(\Sh, \tys)$ a decomposition of $\t$.

\begin{definition}
A non-trivial shape $\Sh \in \Shapesz$ is the principal shape of the type
$\t$ iff:
\begin{enumerate}
  \item
    $\exists \typs,\ \t = \shapp[\Sh] \typs$
  \item
    $\forall \Shp \in \Shapesz, \forall \typs,\ \t = \shapp[\Shp] \typs
    \implies \Sh \preceq \Shp$
\end{enumerate}
\end{definition}

\begin{restatable}[Principal shapes]{theorem}{principalShapesBIS}
  \label{thm:principal-shapes}
Any non-variable type $\t$ has a principal shape $\Sh$.
\end{restatable}

\parcomment {Define canonical shape}

A principal shape $\any \tvcs \t$ is \emph{canonical} if its free variables
appear in the sequence $\tvcs$ in the order in which they occur in $\t$. Canonical
principal shapes are written $\sh$, and we write $\CanonicalShapes$ for the set of
all such shapes.
Each non-variable type $\t$ has a unique canonical principal shape,
written $\shape \t$. For example, $\shape {\trcd \T \tys}$ is
$\any \tvcs \trcd \T \tvcs$.

\parcomment {\Polytypes are constructors}

Shapes are particularly interesting in the context of \polytypes,
which can be decomposed into shapes and thus treated analogously to
type constructors.
\begingroup
\newcommand {\tsh}[1]%
  {\def \tsi {\all \tvb {{\parens{\tva \to #1}}} \tprod \tvb}
  \tpoly {\all \tva {\parens {\tpoly \tsi} \to \tva} \to \tva}}%
For instance, consider the \polytype $\tsh {\tint \tlist}$. Since its 
principal shape $\sh$ is $\any \tvc {\tsh \tvc}$, the original \polytype
can therefore be represented as the shape application $\shapp (\tint \tlist)$.
\endgroup

\subsection {Typing rules for fragile, implicit constructs}
\label{sec/oml/typing/I}

\begin{mathparfig}[tb]
  {fig/oml/typing/I}
  {Typing rules for fragile, implicitly typed extensions.}
  \begin{bnfgrammar}
    \entry[Terms]{\e}{
      \dots
      \and \emagic \es
    }\\
    \entry[Term contexts]{\E}{
      \hole
      \and \eapp \E \e
      \and \eapp \e \E
      \and \dots
    }
  \end{bnfgrammar}
  \par
  \inferrule[Hole]
    {\parens{\G \th \ei : \ti}\iton}
    {\G \th \emagic \es : \tp}

  \inferrule [Use-I]
    {\eshape \E  \e {\any \tvcs \tpoly \ts} \\\\
     \G \th \E\where{\exinst \e \tvcs \ts} : \t}
    {\G \th \E\where{\einst \e} : \t}

  \inferrule [Poly-I]
    {\Eshape \E \e {{\any \tvcs \tpoly \ts}} \\\\
     \G \th \E \where{\expoly \e \tvcs \ts} : \t}
    {\G \th \E \where{\epoly \e} : \t}
\\
  \inferrule[Rcd-I]
    {\Eshape \E \es {\any \tvcs \trcd \T \tvcs} \\
     \G \th \E\where{\exrecord \T {\overline{\elab = \e}}} : \t }
    {\G \th \E\where{\erecord {\overline{\elab = \e}}} : \t}

  \inferrule[Rcd-Proj-I]
    {\eshape \E \e {\any \tvcs \trcd \T \tvcs} \\
     \G \th \E\where{\exfield \e \T \elab} : \t}
    {\G \th \E\where{\efield \e \elab} : \t}

\let \Eqdef \eqdef\def\eqdef{&\Eqdef&}
\let \Implies \implies\def\implies{&\Implies&}
\begin{tabular}{RCL;C;L}
\eshape[\es] \E \e \sh \eqdef
  \forall \G, \t, \gt, \uad
  \G \th \eerase {\E \where {\emagic {\es, \eannot \e {} \gt }}} : \t
      \implies \shape \gt = \sh
\\[.4ex]
\Eshape \E \es \sh \eqdef
  \forall \G, \t, \gt, \uad
      \G \th \eerase {\E\where{\eannotmagic \es {} \gt}} : \t
      \implies \shape \gt = \sh
\end{tabular}
\end{mathparfig}

\parcomment{Typing rules are contextual. Re-iterate the idea of unicity}

We now turn to the typing of fragile implicit constructs
(\cref{fig/oml/typing/I}). To prevent the kind of uncontrolled guessing
permitted by \emph{natural} typing rules (\eg \Rule{Rcd-Proj-I-Nat}), we adopt
\emph{contextual} typing rules: an implicit term $\e^i$ is typed \emph{within}
a surrounding one-hole term context $\E$ (\cref{fig/oml/typing/I}). The context
$\E$ is used to ensure that the relevant type information (\eg a shape $\sh$)
is \emph{known}---that is, it is the unique shape~$\sh$ that can be inferred
from the context, rather than an arbitrary one that could be \emph{guessed}.
Such rules are therefore inherently non-compositional.

\paragraph{Unicity}

\parcomment{What are unicity conditions? e.g. notation}

The key question for our contextual typing rules is whether the shape $\sh$ of
a term's type is \emph{uniquely determined} by its surrounding context $\E$. We
capture this with our \emph{unicity conditions} $\eshape[\es] \E \e \sh$ and
$\Eshape \E \es \sh$, which state that all valid typings of the (erased)
context $\E$ assign the same canonical shape $\sh$ to the subterm $\e$ and the
context's hole $\hole$, respectively.

In other words, $\eshape[\es] \E \e \sh$ holds when the type of the subterm
$\e$ has a unique shape $\sh$ fixed by its context $\E$ and sibling terms
$\es$, while $\Eshape \E \es \sh$ holds when the expected type of the hole
in $\E$ with subterms $\es$\footnote{%
  In this paper, $\es$ is empty in all uses of this unicity condition
  (\Rule{Use-I} and \Rule{Rcd-Proj-I}); in general, it may be non-empty for
  fragile eliminators with multiple subterms. 
}
has the unique shape $\sh$.
\parcomment{The formal definition / auxiliary notions}
\begin{mathpar}
\def \Eqdef {&\eqdef&}\let \Implies \implies\def\implies{&\Implies&}
{\begin{tabular}{RCLCL}
\eshape[\es] \E \e \sh \Eqdef
  \forall \G, \t, \gt, \uad
  \G \th \eerase {\E \where {\emagic {\es, \eannot \e {} \gt }}} : \t
  \Implies \shape \gt = \sh
\\[.4ex]
\Eshape \E \es \sh \Eqdef
  \forall \G, \t, \gt, \uad
      \G \th \eerase {\E\where{\eannotmagic \es {} \gt}} : \t
      \Implies \shape \gt = \sh
\end{tabular}}
\end{mathpar}
To make sense of these definitions, we rely on a few auxiliary notions:
\begin{enumerate}
  \item
    We write $(\emagic \es)$ for a \emph{typed hole} carrying subterms $\es$.
    The subterms $\es$ are required to be well-typed in the current
    environment (\Rule{Hole}), but their types are independent of the type of
    the hole: the hole itself may be assigned an arbitrary type.

  \item
    We introduce an \emph{erasure} function $\eerase \e$ that erases all
    not-yet-elaborated implicit constructs $\e^i$ in $\e$ with a typed hole
    around their subterms. For example, $\eerase {\efield \e \elab}$ is
    $\parens{\emagic {\cerase \e}}$. The full definition is given in
    \cref{fig/oml/erasure}.
\end{enumerate}

\begin{mathparfig}[tb]
  {fig/oml/erasure}
  {Selected cases of the erasure of $\e$. All other cases are homomorphic.}
\newcommand{\Erule}[2]{\eerase {#1} &\eqdef& {#2}}
  \begin{tabular}{RCL}
  \Erule{\epoly \e}{\emagic {\eerase \e}}\\
    \Erule{\expoly \e \tvs \ts}{\expoly {\eerase \e} \tvs \ts}\\
  \Erule{\einst \e}{\emagic {\eerase \e}}\\
  \Erule{\exinst \e \tvs \ts}{\exinst {\eerase \e} \tvs \ts}\\
    \Erule{\erecord {\elaba = \ea; \ldots; \elab_n = \en}}{\begin{cases}
      \erecord {\elaba = \eerase \ea; \ldots; \elab_n = \eerase \en} &\text{if } \labsuni \elabs \T \\
      \emagic {\eerase \ea, \ldots, \eerase \en} & \text{otherwise}
    \end{cases}}\\
  \Erule{\exrecord \T {\elaba = \ea; \ldots; \elab_n = \en}}{\exrecord \T {\elaba = \eerase \ea; \ldots; \elab_n = \eerase \en}}\\
    \Erule{\efield \e \elab}{\begin{cases}
      \efield {\eerase \e} \elab & \text{if } \labuni \elab \T \\
      \emagic {\eerase \e} & \text{otherwise}
    \end{cases}}\\
  \Erule{\exfield \e \T \elab}{\exfield {\eerase \e} \T \elab}\\
  \Erule{\emagic \es}{\emagic {\parens {\eerase \ei} \iton}}\\
\end{tabular}
\end{mathparfig}

\parcomment{Typed holes must preserve typability of subterms}

Typed holes ensure that subterms---such as type annotations---remain
present even when the implicit construct containing them is erased. This is
because they may introduce constraints that can still contribute to unicity.
For instance, $\efun r \efield {\eannot r {} {\ttlab {point}}} {\ttlab x}$
would be ill-typed if we erased $\eannot r {} {\ttlab {point}}$.

\parcomment{Erasure induces a causal order}

The erasure of $\eerase{\E\where{\emagic{\eannot \e {} \gt, \es}}}$ and
$\eerase{\E\where{\eannotmagic \es {} \gt}}$ replaces all
not-yet-elaborated implicit constructs with typed holes, ensuring that the
unicity of $\sh$ is determined \emph{only} by implicit terms that have been
elaborated thus far. This induces a causal, or partial, order among implicit
constructs: an implicit term can only be elaborated once all of its
dependencies---those providing the relevant \emph{known} type
information---have themselves been elaborated. This ordering prevents
self-justifying cycles in unicity and rules out ``out-of-thin-air'' guesses.

\parcomment{Well-foundedness}

The attentive reader may notice that our unicity conditions contain a
negative occurrences of the typing judgment. At first glance, this seems
problematic for well-foundedness. However, the issue is easily resolved by
noting that these occurrences arise only as typing assumptions on erased
terms, which do not themselves involve any implicit rules. Formally one
could make this explicit by introducing a restricted typing judgment $\G
\throbust \e : \t$ that excludes the implicit rules (\NoRule{*-I}), and by
using this strictly simpler judgment in the antecedent of unicity
conditions. We omit this distinction for simplicity.

\paragraph{The omnidirectional recipe}
All typing rules instantiate a common framework---the \emph{omnidirectional
recipe}---which ensures that certain omitted type annotations are uniquely
determined from the context. Each construct, however, requires a specific
instantiation of the framework.
We first describe the framework, then present each feature separately.

\begin{enumerate}
  \item[\emph{Step 1}:] \emph{Contextualize.}
    Each implicit fragile term $\e^i$ is typed relative to a surrounding one-hole term context $\E$:
    its rule asserts the typability of $\G \th \E \where{\e^i} : \t$ as the conclusion.

  \item[\emph{Step 2}:] \emph{Select a unicity condition.}
    This is the secret ingredient! The unicity condition ensures that the shape
    $\sh$ is fully determined by the surrounding context $\E$ and
    subexpressions $\es$ (\eg the subexpressions $\es$ in $\erecord
    {\overline{\elab = \e}}$).

    If the $\e^i$ is an introduction form, we infer the shape from the context's
    hole $\Eshape \E \es \sh$. If $\e^i$ is an elimination form, we infer
    the shape from the \emph{principal term} $\e$ (the term whose type
    contains the connective we're eliminating): $\eshape[\es] \E \e \sh$.

    This division of terms is reminiscent of the \emph{Pfenning recipe} for
    bidirectional typechecking~\citep*{conf/popl/DunfieldP04}, where $\Eshape
    \E \es \sh$ can be viewed as a checking mode, and $\eshape[\es] \E \e \sh$
    as an inference mode for shapes.

  \item[\emph{Step 3}:] \emph{Elaborate.}
    The uniquely inferred shape $\sh$ is used to elaborate $\e^i$ into its
    explicit counterpart~$\e^x$, and the rule asserts $\G \th \E\where{\e^x}
    : \t$ as a premise.
\end{enumerate}

\paragraph{Implicit \polytypes}

Unboxing a \polytype $\einst \e$ is an \emph{elimination form}. Following the
omnidirectional recipe, \Rule{Use-I} is a contextual rule (\emph{Step 1})
requiring that the principal term $\e$ have the unique shape $\any \tvcs \tpoly
\ts$ in the context $\E$ (\emph{Step 2}). In \emph{Step 3}, we then elaborate
$\einst \e$ into $\exinst \e \tvcs \ts$.
Conversely, boxing with $\epoly \e$ is an \emph{introduction form}. In
\Rule{Poly-I}, we require that the expected type of the context's hole $\E$ has
the shape $\any \tvcs \tpoly \ts$ (\emph{Step 2}). We then type $\epoly \e$ as
$\expoly \e \tvcs \ts$ (\emph{Step 3}).

\paragraph{Implicit nominal records}

Overloaded record labels are handled analogously.
Typing record projections in \Rule{Rcd-Proj-I} is an \emph{elimination
  form} for the record type $\trcd \T \tys$: the projection $\efield \e \elab$
is typed as $\exfield \e \T \elab$ (\emph{Step 3}) provided the
type of expression $\e$ in context $\E$ has record shape
$\any \tvcs \trcd \T \tvcs$ (\emph{Step 2}).
For record construction, $\erecord {\overline{\elab = \e}}$ is an
\emph{introduction form}. In \Rule{Rcd-I}, we type an overloaded record
$\erecord {\overline{\elab = \e}}$ as $\exrecord \T {\overline{\elab = \e}}$
(\emph{Step 3}), provided the context $\E$ with subterms $\es$ expects a
record type of shape $\any \tvcs \trcd \T \tvcs$ (\emph{Step 2}).

\parcomment {Illustrative examples}
We now illustrate the typing of implicit constructs with a few examples.

\begin{example}
  \locallabelreset
  Consider the term $\ttex {3} \eqdef \efun r {\etuple {\efield r \ttx,
  \efield {\eannot r {} \ttpoint} \tty}}$ from \cref{sec/overview/overloading}.%
  \footnote{
    The typing rules for tuples are standard and present in Appendix
    \cref{app:full-reference}.
  }
  In \ocaml{ex_3}, $r$ can only be of type \ocaml{point}. Indeed,
  considering the second projection first, we should learn that $r$ is of type
  \ocaml{point} (using \Rule{Annot}) and since it is
  $\lambda$-bound, this makes the first projection unambiguous. (For record
  types without parameters, we use $\T$ as a shorthand for $\trcd \T \emptyset$.)

  Formally, we derive:
\begin{mathpar}
\newcommand{\reinferonbraces}[4]
   {\phantomontop[-1.4ex]
     {#1 \underbrace{#2} {#4}}
     {#1 \overbrace {#3} {#4}}}

\infer* [right=Rcd-Proj-I]{
   \eshape \E {\eannot r {} \ttpoint} {\ttpoint} \\
   \topinfer [Left=Rcd-Proj-I]{
        \eshape \Ep r \ttpoint \\ 
          \eset \th \Ep \where{\exfield r \ttpoint \ttx} :
    \ttpoint \to \tint \tprod \tint
}{
    \reinferonbraces
            {\eset \th}
            {\Ep\where{\efield r \ttx}}
            {\E\where{\exfield r \ttpoint \tty}}
            {: \ttpoint \to \tint \tprod \tint}
   }
}{
      \eset \th \E\where{\efield {\eannot r {} \ttpoint} \tty} :
      \ttpoint \to \tint \tprod \tint
}
\end{mathpar}
  where the contexts $\E$ and $\Ep$ are:
  \begin{mathpar}
    \begin{array}{rcl}
      \E &\eqdef& \efun r {\etuple {\efield r \ttx, \square}} \\
      \Ep &\eqdef& \efun r {\etuple {\square, \exfield r \ttpoint \tty}}
    \end{array}
  \end{mathpar}

  We have $\eset \th \Ep\where{\exfield r \ttpoint \ttx} : \ttpoint \to \tint
  \tprod \tint$, indeed.
  It therefore remains to prove $\eshape \E {\eannot r {} \ttpoint}
  \ttpoint$~\llabel B and $\eshape \Ep r {\ttpoint}$~\llabel A.

  \sloppy
  For \lref B, we use the fact that unicity is \emph{composable}: if $\eshape
  \Ea \e \sh$, then $\eshape {\Eb\where\Ea} \e \sh$. Since ${\eshape \square
  {\eannot r {} \ttpoint} \ttpoint}$ holds trivially, we have $\eshape \E
  {\eannot r {} \ttpoint} \ttpoint$.
  
  For \lref A, assume $\eset \th \Ep\where{\emagic {\eannot r {} \gt}}
  : \t$. 
  Since $\exfield r \ttpoint \tty$ requires $r$ to have the type
  $\ttpoint$ (due to \Rule{Rcd-Proj-X}), it follows that there is no other choice
  but to take $\ttpoint$ for $\gt$, which proves~\lref A.
\end{example}

\begin{example}
To illustrate a simple case of non-typability, we reconsider the example $\ttex
1 \eqdef \efun r {\efield r \ttx}$ from \cref{sec/overview/overloading}.
If there is a derivation of $\ttex 1$, then there must be one of the form:
\begin{mathpar}
  \infer* [Right=Rcd-Proj-I]
    { \eshape \E r {\any \tvcs \trcd \T \tvcs} \\
      \eset \th \E\where{\exfield r \T \ttx} : \t}
    { \eset \th \E\where{\efield r \ttx} : \t}
\end{mathpar}
where $\E$ is the term $\efun r \hole$, which is the largest possible
context.
Unfortunately, $\eshape \E r {\any \tvcs \trcd \T \tvcs}$ does not hold for any
$\T$. Indeed, we have $\eset \th \E \where {\emagic {\eannot r {} {\gt}}} :
\gt \to \t$.  for any $\gt$ and $\t$.  Hence, $\ttpoint$ and $\ttgraypoint$
are both possible shapes for the type of $r$.
\end{example}

\def \ttciecolor {\ttlab {cie\_color}}
\def \ttciepoint {\ttlab {cie\_point}}

\begin{example}
Considering the example from \cref{sec/overview/limitations}.
\begin{program}[input,checkocaml]
type gray_point = { x : int; y : int; color : int }
type cie_color  = { x : int; y : int; z : int }
type cie_point  = { x : int; y : int; color : cie_color }
let ex_1_0 r = r.color.x °\ocamlflags 11°
\end{program}

\parcomment{Proof that this is ill-typed}

As explained in \cref{sec/overview/limitations}, \code{ex_1_0} is ill-typed
because our unicity conditions enforce that each implicit field projection
must be resolved individually and in a sequential fashion. We view this as a
``Goldilocks'' solution: we deliberately trade a small loss in expressivity
for the benefit of a tractable, backtracking-free inference algorithm.

To type $\ttex {10}$ one must eliminate the final implicit
projection in a context of the form $\E\where{\efield \e \elab}$.
Neither projection can be resolved by applying
\Rule{Rcd-Proj-Closed} since the labels (\code{color} and \code{x})
are both overloaded. Thus, the implicit projections must
be typed using \Rule{Rcd-Proj-I}.
Two cases arise, neither of which are possible:
\begin{proofcases}
\proofcase{$\E$ is $\efun r {\hole}$}
We should have a derivation that ends with
\begin{mathpar}
  \infer*[Right=Rcd-Proj-I]
    {\eshape \E {\efield r \ttcolor} {\ttciecolor} \\
     \eset \th \E \where {\exfield {\efield r \ttcolor} \ttciecolor \ttx} : \t}
    {\eset \th \E \where {\efield {\efield r \ttcolor} \ttx} : \t}
\end{mathpar}

However, $\eshape \E {\efield r \ttcolor} {\ttciecolor}$ does not hold. Indeed,
the judgment
  $\eset \th \eerase {\E \where {\eannot {\emagic {\efield r \ttcolor}} {}
  {\gt}}} : \gt \to \tp$
(\ie $\eset \th \efun r {\eannotmagic r {} {\gt}} : \gt \to \tp$) holds for
any $\gt$.  Hence, the shape of the type of $\efield {\efield r \ttcolor}
\ttx$ is not uniquely determined and this case cannot occur.

\proofcase{$\E$ is $\efun r {\efield {\hole} \ttx}$}
The derivation must end with:
\begin{mathpar}
  \infer*[Right=Rcd-Proj-I]
    {\eshape \E {r} {\ttciepoint} \\
     \eset \th \E\where{\exfield r \ttciepoint \tty} : \t}
    {\eset \th \E \where{\efield r \ttcolor} : \t}
\end{mathpar}

However, $\eshape \E {r} {\ttciepoint}$ does not hold
either. Again, the judgment
$\eset\th \eerase {\E\where{\eannotmagic r {} {\gt}}} : \tp \to \gt$,
(\ie $\eset\th \efun r {\eannotmagic r {} {\gt}} : \tp \to \gt$) holds for
any $\gt$.

\end{proofcases}
\end{example}

\subsection{Metatheory}
\label{sec/oml/metatheory}

\paragraph{Principality}

As expected, \OML enjoys principality. We prove this in
\cref{apx/principal-types} from soundness and completeness of constraint
generation (\cref{thm:constraint-gen-is-sound-and-complete}) and correctness of
constraint solving (\cref{corollary:correctness}), established in the
following two sections.

\begin{restatable}[Principal types]{theorem}{principalTypesBIS}
\label{thm:principal-types}
For any well-typed closed \OML term $\e$, there exists a type $\t$
such that:
\begin{enumerate*}[(\roman*)]
\item
  $\th \e : \t$.
\item
  For any other typing $\th \e : \tp$, then $\tp = \theta(\t)$ for some
  substitution $\theta$.
\end{enumerate*}
\end{restatable}

\paragraph{Subject reduction}

\OML does not enjoy subject reduction in the presence of fragile constructs.
Reduction may erase contextual information needed to establish unicity. 
The following program gives a counterexample:
\begin{program}[input]
let ex_1_2 r = if true then r.x else (r : point).y      °\ocamlflags 10°
\end{program}
It typechecks because the annotation of \code{(r : point)}
determines the shape of \code{r}, whereas its reduct \code{r.x} is ill-typed. 

This is a general consequence of the contextual typing rules introduced in
\cref{fig/oml/typing/I}. Subject reduction does, however, hold in the absence
of fragile constructs. We can therefore recover the usual property
operationally by first elaborating fragile constructs into their robust
counterparts, and only then reducing. 

\paragraph{Program transformations} 

Inlining and factorization may likewise fail to preserve typing because they
change what type information is shared. We write these transformations as:
\begin{mathpar}
\begin{tabular}{RCL<{\qquad}l}
  \elet \x \ea \eb &\longmapsto& \eb\where{\x\is\ea} & Inlining\\
  \eb\where{\x\is\ea} &\longmapsto& \elet \x \ea \eb & Factorization
\end{tabular}
\end{mathpar}
subject to the usual variable-capture conditions.

Inlining removes sharing. Recall example \code{ex_8} from
\cref{sec/overview/omni}:
\begin{program}[input]
type 'a gpoint = { x : 'a; y : 'a }                     
let diag (n : 'a) : 'a gpoint = { x = n; y = n }
°\halfline°
let ex_8 gp =                                         °\ocamlflags 20°
  let getx p = p.x in getx (diag 42), (getx gp : float)
\end{program}
Information from both uses of \code{getx} may \emph{backpropagate} back to its
body \code{fun p -> p.x}, providing enough type information to resolve the
overloaded projection \code{p.x}. After inlining, however, each copy of \code{fun p -> p.x}
is resolved independently. In particular, while \code{((fun p -> p.x) gp : float)}
constrains the result of the projection to have type \code{float}, it provides 
no information about the shape of \code{p}'s type. Consequently, the projection 
\code{p.x} cannot be resolved.

Dually, factorization introduces sharing. Suppose we inline \code{getx} from \code{ex_8}
and annotate \code{gp} with a distinct float-point type:
\begin{program}[input]
type fpoint = { x : float; y : float }
°\halfline°
let ex_9 (gp : fpoint) =                      °\ocamlflags 10°
  (fun p -> p.x) (diag 42), ((fun p -> p.x) gp : float)
\end{program}
Here the two copies of \code{fun p -> p.x} typecheck independently, with
\code{p} having types \code{int gpoint} and \code{fpoint}, respectively.
Factorizing them forces both uses to constrain a single typing of \code{p},
requiring a record shape compatible with both. Since these shapes differ, the
factorized term is ill-typed. 

These failures do not apply to all program transformations. In particular,
\emph{linear} let-bindings are preserved by inlining and factorization: if $\x$
occurs exactly once in $\eb$, then $\elet \x \ea \eb$ and $\eb\where{\x\is\ea}$
are equitypable. Backpropagation is precisely what makes linear factorization
possible. 

\OML also preserves the order independence of function application enjoyed by
\ML. In particular, $\eappp f \ea \eb$ and $\eappp {\eswap f} \eb \ea$ are
equitypable where $\eswap$ is defined as $\efun {f x y} \eappp f y x$. This
property is also preserved by \geninst-directional inference, but not by
bidirectional inference (\cref{sec/overview/directional}). Additionally,
typability of $\eta$-expansion is preserved \ie $\e$ and $\efun \x {\eapp \e
\x}$ are equitypable when $\e$ has a function type.

\section{Constraints}
\label{sec:constraints}


\begin{mathparfig}[t]%
  {fig:constraints}%
  {Selected syntax and semantics of constraints.}
\begin{bnfgrammar}
\entry[Constraints]{\c}{
        \ctrue
  \and  \cfalse
  \and  \ca \cand \cb
  \and  \cunif \ta \tb
  \and  \cexists \tv \c
  \and 	\cfor \tv \c
  \nextline
  \and  \clet \x \tv \ca \cb
  \and  \capp \x \t
  \and  \cmatch \t \cbrs
}\\[1ex]
\entry[Branches]{\cbr}{\cbranch \cpat \c} \\
\entry[Shape patterns]{\cpat}{
  \cpatwild
  \and \ldots
}{} \\
\entry[Constraint contexts]{\C}{
  \hole
  \and \C \cand \c
  \and \c \cand \C
  \and \cexists \tv \C
  \and \cfor \tv \C
  \nextline
  \and \clet \x \tv \C \c
  \and \clet \x \tv \c \C
} \\[1ex]
\entry[Semantic environments]{\semenv}{
  \emptyset
  \and \semenv\where{\tv \is \gt}
  \and \semenv\where{\x \is \gabs}
} \\
\entryset[Ground types]{\gt}{\Ground} \\
\entrysubseteq[Sets of ground types]{\gabs}{\Ground}{}
\end{bnfgrammar}

  \infer[True]
    { }
    {\semenv \th \ctrue}

  \infer[Conj]
    {\semenv \th \ca \\
     \semenv \th \cb}
    {\semenv \th \ca \cand \cb}

  \infer[Unif]
    {\semenv(\ta) = \semenv(\tb)}
    {\semenv \th \cunif \ta \tb}

  \infer[Exists]
    {\semenv\where{\tv \is \gt} \th \c}
    {\semenv \th \cexists \tv \c}

  \infer[Forall]
    {\forall \gt, ~ \semenv\where{\tv \is \gt} \th \c}
    {\semenv \th \tfor \tv \c}

  \infer[Let]
    {\gabs = \semenv(\cabs \tv \ca) \\
     \gabs \neq \eset \\\\
     \semenv\where{\x \is \gabs} \th \cb}
    {\semenv \th \clet \x \tv \ca \cb}

  \infer[App]
    {\semenv(\t) \in \semenv(\x)}
    {\semenv \th \capp x \t}

  {\let \Eqdef\eqdef \def \eqdef {&\Eqdef&}
  \begin{tabular}[c]{.R.C;.L.}
  \semenv(\cabs \tv \c) \eqdef
    \set {\gt \in \Ground : \semenv\where{\tv \is \gt} \th \c}
  \\
  \ca \centails \cb \eqdef
    \forall \semenv,\ \semenv \th \ca \implies \semenv \th \cb
  \\
  \ca \cequiv \cb \eqdef
    (\ca \centails \cb) \wide\wedge   (\ca \centails \cb)
  \end{tabular}}
\par
\begin{tabular}{.L.}
  \cmatched \t \sh {\cbranch \cpats \cs} \eqdef \\[.6ex]
\uad \begin{cases}
    \cexists \tvs \cunif \t \shapp \tvs \cand \theta(\ci) &
    \text{if } \cmatches \cpati \sh \tvs \theta\\
    \cfalse & \text{otherwise}
\end{cases}
\end{tabular}
\qquad\qquad
\hfil
{\let \Eqdef\eqdef \def \eqdef {&\Eqdef&}\def \EQ{&\partialmap&}
\begin{tabular}[c]{.R.C;.L.}
  \hline \vrule\;
  \cmatches[\EQ] \cpat \sh \tvcs \theta
  \;\vrule \\ \hline
  \noalign{\vskip .6ex}
  \cmatches[\EQ] \cpatwild \sh \tvcs \eset \\
  &\vdots&
\end{tabular}}

\infer[Match-Ctx]
    {\Cshape \C \t \sh \\
      \semenv \th \C\where{\cmatched \t \sh \cbrs}
    }
    {\semenv \th \C\where{\cmatch \t \cbrs}}
\hfill
\begin{tabular}{.L.}
  \Cshape \C \t \sh \eqdef \\[.6ex]
  \quad \forall \semenv, \gt. \uad \semenv \th \cerase {\C\where{\cunif \t \gt}} \implies \shape \gt = \sh
\end{tabular}
\end{mathparfig}

\parcomment{Why bother with a formal constraint language?}

To reason about constraint-based inference, we need more than a procedure for
generating and solving constraints: we require a \emph{formal logic} of
constraints, with a syntax and a declarative semantics that characterizes
satisfiability. This semantics is essential: it validates the design of our
constraint language and provides the foundation for proving the soundness,
completeness, and principality of inference. Without it, the meta-theory of our
approach cannot be stated precisely.

\parcomment{Defining the syntax / semantics}

We now introduce the syntax and semantics of our constraint language.
Building atop the constraint-based inference framework of
\citet*{\BBemlti}, we adopt a constraint language
(\cref{fig:constraints}) that includes both term and type variables. Its
semantics is given by a satisfiability judgment $\semenv \th \c$
(\cref{fig:constraints}). The semantic environment $\semenv$ assigns to each
free type variable $\tv$ a ground type $\gt \in \Ground$ (a type with no
free variables) and to each term variable $\x$ a set of ground
types\footnote {$\gabs$ can be pronounced ``Fraktur S'' or ``ground S''.}
$\gabs \subseteq \Ground$ (the instances of a type scheme bound to
$\x$).
We write $\semenv\where{\tv \is \gt}$ and $\semenv\where{\x \is \gabs}$ for
extensions of $\semenv$, and $\semenv(\t)$ for the ground type obtained by
substitution.

\parcomment{Semantics detailed}

Constraints include basic logical forms: tautological $\ctrue$ (\Rule{True}),
unsatisfiable $\cfalse$, and conjunctive $\ca \cand \cb$ (\Rule{Conj})
constraints. The unification constraint $\cunif \ta \tb$ is satisfied when
$\ta$ and $\tb$ are equal (\Rule{Unif}).
An existential constraint $\cexists \tv \c$ holds if there exists a witness
$\gt$ for $\tv$ satisfying $\c$ (\Rule{Exists}), while a universal constraint
$\cfor \tv \c$ holds if $\c$ is satisfied for every binding of $\tv$
(\Rule{Forall}).
When $\ts$ is a polymorphic type scheme $\tfor \tvs \tp$, we use the notation
$\cleq \ts \t$ as shorthand for the instantiation constraint $\cexists
\tvs \cunif \tp \t$. \OML-specific constraints, such as record label
instantiation $\labfrom \elab \ct \leq \ta \to \tb$
(\cref{sec/overview/omni/match}), are introduced later (\cref{sec/oml/constraints}).

\parcomment{Constraint abstractions}

Polymorphism in the constraint language is expressed through
\emph{generalization} and \emph{instantiation} constraints. In a generalization
constraint (or \emph{let}-constraint) $\clet \x \tv \ca \cb$, the definition of
$\x$ is a constraint abstraction $\cabs \tv \ca$, a function that, when applied
to a type $\t$, returns $\ca \where {\tv \is \t}$. The binding of $\x$ is then
available in $\cb$ and refers to this abstraction. Semantically, the
abstraction $\cabs \tv \ca$ is interpreted as a set of ground types that
satisfies $\ca$, and the generalization constraint requires this set to be
non-empty, \ie there is at least one instantiation of $\tv$ that satisfies
$\ca$.
Instantiations (or applications) $\capp \x \t$ eliminate abstractions by
applying a type $\t$ to the abstraction bound to $\x$. This holds precisely
when $\semenv(\t) \in \semenv(\x)$, \ie $\t$ is one of the satisfiable
instances of $\x$. The $\lambda$-abstraction and application syntax follows
\citet*{Pottier/inferno@icfp2014}. In other presentations, constraint
abstractions $\cabs \tv {\cexists \tvbs \c}$ are written as constrained type
schemes $\all {\tv, \tvbs} C \Rightarrow \tv$, and instantiation constraints
$\capp \x \t$ are written $\x \leq \t$.

\parcomment {Suspended match constraints}

Finally, we introduce \textit{suspended match constraints} $(\cmatch \t \cbrs)$,
which consist of:
\begin{enumerate}

\item
  A matchee $\t$. The constraint remains suspended
  until the \emph{shape} of $\t$ is determined, \ie
  while $\t$ is a type variable.

\item

  A list of branches $\cbrs$ of the form $\cbranch \cpat \c$, where
  $\cpat$ is a shape pattern.\footnote{%
    The match constraints considered in this paper only use a single branch.
    In a richer language, however, multiple branches would be useful; for 
    instance, record projection could be overloaded to operate on
    nominal records as well as tuples (or objects, modules, \etc), requiring
    several branches in the generated constraint. We therefore kept the general
    syntax.%
  } The constraint $\c$ is solved in the extended
  context produced by the matching pattern.

  For example, the wildcard pattern $\cpatwild$ matches any shape,
  and binds nothing.
  To ensure determinism, the set of patterns $\bar \cpat$ must be
  \emph{disjoint}---that is, no shape may be matched by more than one pattern
  in the list.

\end{enumerate}
The formal semantics of suspended match constraints are somewhat involved;
we return to them in the next subsection (\cref{sec/constraints/semantics}).

\begin{wraphbox}{0.2}{0.6}
\begin{mathpar}[inline]
\infer*[right=Exists]
    {\infer*[Right=Unif]
      {\infer*{}{\tint = \tint}}
      {\semenv\where{\tv \is \tint} \th \cunif \tv \tint}}
  {\semenv \th \cexists \tv \cunif \tv \tint}
\end{mathpar}
\end{wraphbox}
Closed constraints are either satisfiable in any semantic environment (\ie
they are tautologies) or unsatisfiable. For example, the satisfiability of
the constraint $\cexists \tv {\cunif \tv \tint}$ is established by the
derivation on the right-hand side.


We write $\ca \centails \cb$ to express that $\ca$ \emph{entails} $\cb$,
meaning every solution $\semenv$ to $\ca$ is also a solution to $\cb$.
We write $\ca \cequiv \cb$ to indicate that $\ca$ and $\cb$ are equivalent,
that is, they have exactly the same set of solutions.

\parcomment {Constraint contexts}

Throughout this paper, we will find it convenient to work with \emph{constraint
contexts}. A constraint context $\C$ is simply a constraint with a \emph{hole},
analogous to term contexts $\E$ introduced in \cref{sec/oml}. We write
$\C\where{\c}$ to denote filling the hole of the context $\C$ with the
constraint $\c$. Hole filling may capture variables, so we frequently require
explicit side conditions when variable capture must be avoided. We write $\bvs
\C$ for the set of variables bound at the hole in $\C$.

\subsection{Suspended constraints}
\label{sec/constraints/semantics}

\parcomment{Defining semantics for suspended constraints is hard}

A central difficulty in our work on suspended constraints was defining a
satisfying semantics. The challenge lies in formalizing what it means for type
information to be \emph{known} without presupposing a \emph{static} solving
order. This is the same issue encountered in our typing rules
(\cref{sec/oml/typing/I}), and we address it in the same way: by introducing a
contextual rule together with a unicity condition.

To define the semantics for suspended constraints, we first introduce
\emph{discharged match constraints}.

\begin{definition}[Discharged match constraint]\label{def/discharged}
  Given a suspended constraint $(\cmatch \t \cbrs)$ and a canonical shape
  $\sh$, we introduce the syntactic sugar $(\cmatched \t \sh \cbrs)$ for the
  \emph{discharged match constraint} that selects the branch in $\cbrs$ that
  matches $\sh$:
\begin{mathpar}
  \cmatched \t \sh {\overline {\cbranch \cpat \c}} \uad\eqdef\uad
  \begin{cases}
    \cexists \tvs \cunif \t \shapp \tvs \cand \theta(\ci)
    & \text{if } \cmatches \cpati \sh \tvs \theta\\
    \cfalse & \text{otherwise}
  \end{cases}
\end{mathpar}

The first conjunct ($\tau = \shapp \tvs$) ensures that $\sh$ is indeed the
canonical shape of $\t$, and the second conjunct is the selected branch
constraint $\ci$ under the appropriate substitution. Since the syntax of
suspended match constraints requires that branch patterns are
non-overlapping, the matching branch $\cbranch \cpati \ci$ is uniquely
determined. It may, however, be undefined if the branches are not exhaustive,
in which case the discharged constraint is $\cfalse$.

\end{definition}

The partial function $(\Cmatches \cpat \sh \tvs)$, introduced in
\cref{fig:constraints}, describes how a pattern $\cpat$ is matched against a
canonical principal shape $\sh$. Before matching, $\sh$ is applied with fresh
shape variables $\tvs$ (of the same arity as $\sh$), so that the result of
matching may refer to them. The match either fails or returns a substitution
$\theta$ mapping pattern variables (\eg record name variables $\ct$ from
\cref{sec/overview/omni}) to shape components (\eg record names $\T$).  These
components may themselves mention the freshly introduced variables $\tvs$.
In~\cref {fig:constraints}, we only introduce the wildcard pattern and its
matching rule; in~\cref {sec/oml/constraints} we extend both the syntax of
shape patterns and the matching function itself to handle the patterns that
arise in \OML.

\paragraph {A natural attempt}

As with typing rules for fragile constructs (\eg \Rule{Rcd-Proj-I-Nat}), suspended
match constraints have a \emph{natural rule}:
\begin{mathpar}
  \infer[Match-Nat]
  {\sh = \shape {\semenv(\t)} \\ \semenv \th \cmatched \t \sh \cbrs}
  {\semenv \th \cmatch \t \cbrs}
\end{mathpar}
This rule states that a suspended constraint holds whenever the corresponding
discharged constraint holds for the canonical shape of $\semenv(\t)$.

\parcomment {The problem}

Although simple and declarative, this semantics is too
permissive. It allows \emph{guessing} the shape of $\t$ rather than requiring
it to be \emph{known}. For example, $\cexists \tv \cmatch \tv {\cbranch \cpatwild {\cunif
\tv \tint}}$ is satisfiable under the natural semantics:
\begin{mathpar}
\def \cmatchex {\cmatch \tv {\cbranch \cpatwild {\cunif \tv \tint}}}
\def \semenvex {\semenv\where{\tv \is \tint}}
    \infer*[Right=Match-Nat]
    {
      \cmatches \cpatwild \tint \eset \eset
      \\
      \infer*[Right=Unif]
        {\tint = \tint}
    {\semenvex \th \cunif \tv \tint}
}{
    \infer*[Right=Exists]
    {\semenvex \th \cmatchex}
  {\semenv \th \cexists \tv \cmatchex}
}
\end{mathpar}
This ``out-of-thin-air'' behavior does not match the intended
meaning of suspended match constraints and raises several problems:
\begin{enumerate*}
  \item a reasonable solver---one that avoids backtracking---cannot
    be complete with respect to this semantics; and

  \item it breaks the existence of principal solutions, just
    as the \emph{natural} typing rules do (\eg \Rule{Rcd-Proj-I-Nat}).

\end{enumerate*}

\paragraph {Contextual semantics}

To rule out guessing, we instead adopt a \emph{contextual} semantics: a
match constraint is satisfiable only if the shape of the type is determined
by the surrounding context. The corresponding rule for suspended
constraints, \Rule {Match-Ctx} (\cref{fig:constraints}), is the only non-syntax-directed rule in our semantics.
\begin{mathpar}
  \infer[Match-Ctx]
    {\Cshape \C \t \sh \\
      \semenv \th \C \where {\cmatched \t \sh \cbrs}
    }
    {\semenv \th \C \where {\cmatch \t \cbrs}}
\end{mathpar}
In this rule, a suspended constraint $(\cmatch \t \cbrs)$ in the context
$\C$ can be discharged, provided the shape $\sh$ is not guessed from $\semenv$,
but recovered from the constraint context $\C$. This \emph{unicity} condition
$\Cshape \C \t \sh$ (defined below) ensures that $\sh$ is uniquely determined
by the context $\C$, capturing precisely what it means for the shape of a type
to be \emph{known}.

\begin{definition}[Erasure]
  \label{def:erasure}
  The erasure $\cerase \c$ of a constraint $\c$ is defined as the constraint
  obtained by replacing suspended match constraints in $\c$ with $\ctrue$.
\end{definition}

\begin{definition}[Simple constraints]
  We say that $\c$ is \emph{simple} if it contains no suspended match
  constraints.
\end{definition}

\begin{definition}[Unicity]
  \label{def:unicity}
  We define the unicity condition $\Cshape \C \t \sh$, which states that $\t$
  has a unique canonical shape $\sh$ within the context $\C$ as:
  \begin{mathpar}[inline]
    \forall \semenv, \gt. \uad
      \semenv \th \cerase {\C\where{\cunif \t \gt}} \implies
          \shape \gt = \sh
  \end{mathpar}.
\end{definition}

\parcomment{Briefly explain unicity (using analogies to typing rules)}

Similarly to the unicity conditions introduced for typing rules in
\cref{sec/oml/typing/I}, unicity
for constraints $\Cshape \C \t \sh$ relies on \emph{erasure} $\cerase
{\C\where{\cunif \t \gt}}$ to restrict attention to previously discharged
match constraints, inducing a causal order between match constraints,
enforced by \Rule {Match-Ctx}.

\parcomment{The interesting cases of unicity}

When $\t$ is not a variable, $\Cshape \square \t \sh$ holds trivially when
$\sh$ is the shape of $\t$. Likewise, when (the erasure of) $\C$ is
unsatisfiable, then $\Cshape \C \tv \sh$ holds vacuously for any $\sh$. The
interesting and nontrivial cases---the ones we illustrate next---arise when
$\t$ is a type variable and $\C$ is satisfiable.

\parcomment {Examples}

\begin{example}
  Recall the problematic example that was satisfiable under the natural
  semantics:
\begin{mathpar}
  \cexists \tv \cmatch \tv {\cbranch \cpatwild {\cunif \tv \tint}}
\end{mathpar}
  Under the contextual semantics, the suspended constraint appears in a
  context $\C$
  with no contextual information. \ie $\C$ is $\cexists \tv \hole$.
  So for any ground type $\gt$, $\cerase{\C\where{\cunif \tv \gt}}$ (\ie
  $\cexists \tv {\cunif \tv \gt}$) is satisfiable, allowing $\gt$ to have an
  arbitrary shape (\eg $\tint$, $\tbool$, \etc). As a result, the uniqueness
  condition $\Cshape \C \tv \sh$ never holds making \Rule{Match-Ctx}
  inapplicable. The constraint is unsatisfiable as intended.

\end{example}

\begin{example}
  Consider the satisfiable constraint:
\begin{mathpar}
\cexists \tv \cunif \tv \tint
  \cand
  \cmatch \tv {\cbranch \cpatwild \ctrue}
\end{mathpar}
  Here, we apply the contextual rule with the context $\C$ equal to $\cexists \tv
  \cunif \tv \tint \cand \hole$. Any solution $\semenv$ of this context
  necessarily satisfies \relax $\cunif \tv \tint$, so we have \relax $\Cshape
  \C \tv \tint$ and the suspended constraint can be discharged.
\end{example}
\begin{example}
  Consider the more intricate example:
\begin{mathpar}
  \cexists {\tva, \tvb}
  {\def \EX
     {\cmatch \tva {\cbranch \cpatwild {\cunif \tvb \tbool}} \and
      \cmatch \tvb {\cbranch \cpatwild \ctrue} \and
      \cunif \tva \tint}
    \False
      {\def \and{\\{}\cand}\Parens
         {\begin{array}{;l}
            \quad \EX
          \end{array}}}
      {\def \and{)\wide\wedge(}\parens{\EX}}
  }
\end{mathpar}
  Suppose we attempt to apply \Rule{Match-Ctx} to the match on $\tvb$ first. We
  want to show $\Cshape \C \tvb \tbool$ for the context $\C$ equal to
  $(\cmatch \tv
  {\cbranch \cpatwild {\cunif \tvb \tbool}}) \cand \hole \cand \cunif \tva
  \tint$. Its erasure $\cerase \C$ is $\ctrue \cand \hole \cand \cunif \tva
  \tint$, which imposes no constraints on $\tvb$. Thus both \relax $\cerase
  {{\C\where{\cunif \tvb \tint}}}$ and \relax $\cerase {\C\where{\cunif \tvb
  \tbool}}$ are both satisfiable: unicity does not hold and \Rule{Match-Ctx}
  cannot be applied.

  By contrast, if we first discharge the match on $\tva$, we consider the
  context $\C$ equal to $\hole \cand (\cmatch \tvb {\cbranch \cpatwild
  \ctrue}) \cand
  \cunif \tva \tint$. Its erasure $\cerase \C$ equal to $\hole \cand \ctrue \cand
  \cunif \tva \tint$ does constraint $\tva$, giving $\Cshape \C \tva \tint$.
  We may therefore discharge the match on $\tva$, rewriting it as
  $(\cmatched \tva \tint {\cbranch \cpatwild {\cunif \tvb \tbool}})$ \ie
  $\tva = \tint \cand \cunif \tvb \tbool$. Substituting back, we
  are left to satisfy the constraint $\C \where {\tva = \tint \cand \cunif \tvb \tbool}$ \ie
\begin{mathpar}[inline]
  \cunif \tva \tint \cand \cunif \tvb \tbool \cand
  (\cmatch \tvb {\cbranch \cpatwild \ctrue}) \cand
  \cunif \tva \tint
\end{mathpar}.
  At this point, unicity for $\tvb$ holds, since the context now includes $\cunif
  \tvb \tbool$. We can therefore apply \Rule{Match-Ctx} to eliminate the final
  match constraint.

  This example demonstrates that suspended constraints must be resolved in
  a dependency-respecting order: attempting to resolve a match constraint too
  early may result in unsatisfiability.
\end{example}

\begin{example}
  Let us consider a constraint with a cyclic dependency between match
  constraints:
\begin{mathpar}
  \cexists {\tva, \tvb}
  {\def \EX
     {\cmatch \tva {\cbranch \cpatwild {\cunif \tvb \tbool}} \and
      \cmatch \tvb {\cbranch \cpatwild {\cunif \tva \tint}}}
    \False
      {\def \and{\\{}\cand}\Parens
         {\begin{array}{;l}
            \quad \EX
          \end{array}}}
      {\def \and{)\wide\wedge(}\parens{\EX}}
  }
\end{mathpar}
  Under the natural semantics this constraint is satisfiable, since one may
  \emph{guess} the assignment $\tva \is \tint, \tvb \is \tbool$,
  making both match constraints succeed. However, our solver and contextual semantics
  reject it.

  Without loss of generality, suppose we attempt to apply \Rule{Match-Ctx} on
  $\tva$ first. We must establish $\Cshape \C \tva \tint$ for the context $\C
  \is \hole \cand \cmatch \tvb {\cbranch \cpatwild {\cunif \tva \tint}}$. But
  the erasure $\cerase \C$ is $\hole \cand \ctrue$, which imposes no constraint on
  $\tva$. Hence unicity fails and \Rule{Match-Ctx} is inapplicable.

\end{example}
\subsection{Constraint generation}
\label{sec/oml/constraints}

\begin{mathparfig}[tp]
  {fig:patterns-oml}
  {Patterns for \OML.}
  \begin{bnfgrammar}
   \entry[Record variables]{\ct}{}\\
   \entry[Scheme variables]{\cscm}{}\\
   \entry[Shape patterns]{\cpat}{
      \ldots
      \and \cpatrcd \ct
      \and \cpatpoly \cscm
    } \\[1ex]
    \entry[Constraints]{\c}{
      \dots
      \and \labfrom \elab \ct \leq \ta \to \tb
      \and \dom \ct = \elabs
      \and \cscm \leq \t
      \and \x \leq \cscm
    }
 \end{bnfgrammar}
  \\
  \def \>{\noalign{\vskip .8ex}}
  \newcommand{\Mrule}[5][]{{#2} \Matches {(#3)} \; #4 &\eqdef& {#5} & #1}
  \begin{tabular}{RCLL}
    \Mrule
      {(\cpatrcd \ct)}
      {\any \tvcs \trcd \T \tvcs} \tvcps
      {[\ct \is \T]}
    \\\>
    \Mrule
      {\cpatpoly \cscm}
      {\any \tvcs \tpoly \ts} \tvcps
      {[\cscm \is \ts \where{\tvcs \is \tvcps}]}
  \end{tabular}
  \\
  \newcommand{\Srule}[3][]{{#2} &\eqdef& {#3} & {#1}}
  \newcommand{\Scases}[3]{%
    \left\{
      \begin{array}{l}
      #1\\
      #2
      \end{array}%
    \right.
    &
    \hspace{-1ex}\begin{array}{l}
      \text{if } #3\\
      \text{otherwise}
    \end{array}%
  }
  \begin{tabular}{RCLL}
    \labfrom \elab \T \leq \ta \to \tb &\eqdef&
      \Scases{\cexists \tvs \cunif \ta  {\trcd \T \tvs} \cand \cunif \tb \t}{\cfalse}{\labenv(\labfrom \elab \T) = \tfor \tvs {\trcd \T \tvs \to \t}}
    \\\>
    \dom{\T} = \elabs &\eqdef&
       \Scases{\ctrue}{\cfalse}{\Dom {\labenv(\T)} = \elabs}
    \\\>
    \Srule
      {(\tfor \tvs \tp) \leq \t}
      {\cexists \tvs \cunif \tp \t}
    \\\>
    \Srule
      {\x \leq (\tfor \tvs \t)}
      {\cfor \tvs \capp \x \t}
  \end{tabular}
\end{mathparfig}

\begin{mathparfig}
  [htpb!]
  {fig/constraint-gen}
  {The constraint generation translation for \OML.}
\newcommand {\Crule}    [2]{#1 &\eqdef& #2 &}
\newcommand {\CruleCond}[3]{\Crule{#1}{#2}#3}
\def \arraystretch{1.1}
\begin{tabular*}{\linewidth}{;;;L=;C;;L~L.}
\Crule
   {\cinfer x \t}
   {\cinst x \t}
\\
\Crule
  {\cinfer {()} \t}
  {\cunif \t \tunit}
\\
\Crule
  {\cinfer {\efun \x \e} \t}
  {\cexists {\tva, \tvb}
    \clet \x \tvp {\cunif \tvp \tv} {\cinfer \e \tvb}
    \cand \cunif \t {\tva \to \tvb}}
\\
\Crule
  {\cinfer {\eapp \ea \eb} \t}
  {\cexists {\tva, \tvb}
    \cinfer \ea {\tvb} \cand \cinfer \eb \tva
    \cand \cunif \tvb {\tva \to \t}}
\\
\Crule
  {\cinfer {\elet \x \ea \eb} \t}
  {\clet \x \tv {\cinfer \ea \tv} {\cinfer \eb \t}}
\\
\Crule
  {\cinfer {\eannot \e \tvs \tp} \t}
  {\cexists \tvs  \cinfer \e \tp \cand \cunif \t \tp}
\\
\Crule
  {\cinfer {\expoly \e \tvs \ts} \t}
  {\cexists {\tvs}
    \cinfer \e \ts
    \cand \cunif \t {\tpoly \ts}}
\\
\Crule
  {\cinfer {\exinst \e \tvs \ts} \t}
  {\cexists {\tvs, \tvb}
    \cinfer \e \tvb
    \cand \cunif \tvb {\tpoly \ts}
    \cand \ts \leq \t}
\\
\Crule
  {\cinfer {\einst \e} \t}
  {\cexists \tv
    \cinfer \e \tv
    \cand \cmatch \tv {\cbranch {\cpatpoly \cscm} \cscm \leq \t}}
\\
\Crule
  {\cinfer {\epoly \e} \t}
  {\clet \x \tv {\cinfer \e \tv}
    {\cmatch \t {\cbranch {\cpatpoly \cscm} {\x \leq \cscm}}}}
\\
\CruleCond
  {\cinfer {\efield \e \elab} \t}
  {\!\begin{cases}
    \cinfer {\exfield \e \T \elab} \t
    \\
    \cexists \tv \cinfer \e \tv \cand
    \cmatch \tv
      {\cbranch {\cpatrcd \ct} {\labfrom \elab \ct \leq \tv \to \t}}
  \end{cases}}
  {\begin{aligned}
    &\text{if\space} \labuni \elab \T
    \\
    &\text{otherwise}
  \end{aligned}}
\\
\Crule
  {\cinfer {\exfield \e \T \elab} \t}
  {\cexists \tv \cinfer \e \tv \cand
   \labfrom \elab \T \leq \tv \to \t}
\\
\CruleCond
  {\cinfer {\erecord {\overline{\elab = \e}}} \t}
  {\!\begin{cases}
    \cinfer {\exrecord \T {\overline{\elab = \e}}} \t
    \\
    \cexists \tvs \cAnd\iton \cinfer \ei \tvi \uad \cand
    \\
    \qquad
        \cmatch \t {
                \cbranch {\cpatrcd \ct}
      {\parens{\dom \ct = \elabs \cand
               \cAnd\iton \labfrom \elabi \ct \leq \t \to \tvi}}}
      \hskip -8em
   \end{cases}}
   {\begin{aligned}
    &\text{if\space} \labsuni \elabs \T \\[.2ex]
    &\text{otherwise}\\
    &\\
  \end{aligned}}
\\
\Crule
  {\cinfer {\exrecord \T {\overline{\elab = \e}}} \t}
  {\cexists \tvs \cAnd\iton \cinfer \ei \tvi
   \cand \dom {\T} = \elabs
   \cand \cAnd\iton \labfrom \elabi \T \leq \t \to \tvi }
\\
\Crule
  {\cinfer {\emagic \es} \t}
  {\cexists \tvs \cAnd\iton \cinfer \ei \tvi}
\vadjust{\medskip}
\\
\Crule
  {\cinfer \e {\tfor \tvs \t}}
  {\cfor \tvs \cinfer \e \t}
\end{tabular*}
\end{mathparfig}

We now present the formal translation from terms $\e$ to constraints $\c$,
such that the resulting constraint is satisfiable if and only if the term is
well typed. The translation is defined as a function $\cinfer \e \t$, where $\e$
is the term to be translated and $\t$ is the expected type of $\e$.
The expected type $\t$ is permitted to contain type variables, which can be
existentially bound in order to perform type inference. The models of
constraint $\cinfer \e \t$ interpret the free variables of
$\t$ such that $\t$ becomes a valid type of $\e$. For example, to infer the
entire type of $\e$ we may pick a fresh type variable $\tv$ for $\t$.

\paragraph{Shape patterns}

Thus far, our formal presentation of shape patterns has remained
abstract, deliberately leaving the syntax and semantics of match constraints
partially unspecified to accommodate a range of language features. We now
concretize this by specifying the shape patterns used in \OML
(see \cref{fig:patterns-oml}), and introducing the corresponding constraints
for the variables they bind.
Shape patterns include:
\begin{enumerate}

  \item Record type patterns $\cpatrcd \ct$ which match record types $\trcd \T \tys$
    and bind the record variable $\ct$ to the name $\T$. 
    The parameters $\tys$ are not bound, since record
    field overloading depends only on the name $\T$.

  \item \Polytype patterns $\cpatpoly \cscm$, which match polytypes $\tpoly \ts$ and
    bind the scheme variable $\cscm$ to $\ts$.
\end{enumerate}

Each new kind of pattern introduces corresponding constraint formers:
$\labfrom \elab \ct \leq \ta \to \tb$ asserts that $\ta \to \tb$ is an
instance of projection type associated with the explicit label $\labfrom
\elab \ct$; $\dom \ct = \elabs$ ensures that the domain of record type
$\ct$ is the set of labels $\elabs$; and $\cscm \leq \t$ checks that $\t$ is
an instance of $\cscm$, while $\x \leq \cscm$ asserts that every instance of
$\cscm$ is an instance of $\x$.

By definition, each of these constraint forms is unsatisfiable, since we do not
extend the satisfiability relation to include them. Their role is purely
syntactic: they only appear within the branches of a suspended match
constraint. Once such a constraint is discharged $(\cmatched \t \sh {\cbranch
\cpats \cs})$ (\cref {def/discharged}), the substitution $\theta$ produced by
the matching pattern $\cpati$---which binds the pattern variables (\eg
$\ct$, $\cscm$, \etc)---is
applied to the corresponding branch $\ci$. Hence, it suffices to define
the semantics of these constraint formers only for their substituted forms (\eg
$\labfrom \elab \T \leq \ta \to \tb$, $\dom {\T} = \elabs$, \etc), as shown in
\cref{fig:patterns-oml}. We define the semantics of these substituted forms by
translation into existing constraints---that is, as syntactic sugar over the
existing core constraint language.

\paragraph{Constraint generation}

Constraint generation $\cinfer \e {\mathop{\t}}$ is defined in
\cref{fig/constraint-gen}. All generated type variables are fresh with respect
to the expected type $\t$, ensuring capture-avoidance. We now review the cases
of the constraint generator, beginning with the cases corresponding to
traditional \ML constructs.

\parcomment{Simple \ML terms}

Unsurprisingly, variables $\x$ generate an instantiation constraint $\capp \x
\t$. The unit term $\eunit$ requires $\t$ to be the unit type $\tunit$. A
function generates a constraint that binds two fresh flexible type variables
for the parameter and return types.  We use a let-constraint to bind the
parameter in the constraint generated for the body of the function. The
let-constraint is monomorphic since $\tvp$ is fully constrained by type
variables defined outside the abstraction's scope and therefore cannot be
generalized. Applications introduce two fresh flexible type variables, one for
the argument type and one for the type of the function, typing each subterm
with these, ensuring $\t$ is the expected return type. Let-bindings generate a
polymorphic let-constraint; $\cabs \tv {\cinfer \e \tv}$ is a principal
constraint abstraction for $\e$: its intended interpretation is the set of all
types that $\e$ admits. Annotations bind their flexible type variables and
enforce the equality of the annotated type $\tp$ and the expected type $\t$.

\parcomment{\Polytypes}

Explicit \polytype boxing asserts that $\e$ has the polymorphic type $\ts$
(using universal quantification) and that the expected type is the \polytype
$\tpoly \ts$. Conversely, explicit unboxing requires that $\t$ be an instance
of $\ts$. These cases introduce no suspended match constrains because the
annotations contain the required type information: the \polytype $\tpoly \ts$.

By contrast, implicit unboxing suspends until the inferred type of $\e$ is
known to be some non-variable type $\tz$. At that point, the suspended match
constraint attempts to match $\tz$ against the pattern $\cpatpoly \cscm$. If
$\tz = \tpoly \ts$, the match succeeds, binding $\cscm$ to $\ts$ and requiring
$\t$ to be an instance of $\cscm$ (\ie $\ts \leq \t$); otherwise, the match
fails and the term is ill-typed.

Implicit boxing behaves dually: we infer the principal type for $\e$ using a
let-constraint and suspend until the expected type of the entire term
resolves to some type $\tz$. If $\tz$ matches the same pattern $\cpatpoly
\cscm$ (\ie $\tz = \tpoly \ts$), we assert that the principal type of $\e$ is
at least as general as $\cscm$, via the constraint $\x \leq \cscm$ (\ie $\x
\leq \ts$); otherwise, the match fails.

\parcomment{Records}

Record constructs are handled in a similar way, but their matching patterns
concern record shapes rather than \polytypes. A record projection generate a
fresh variable $\tv$ for the record type and constrain $\e$ to this type,
suspending until the type of $\e$ ($\tv$) resolves to some type $\tz$. When
$\tz = \trcd \T \tys$, the shape pattern $(\cpatrcd \ct)$ matches successfully,
binding $\ct$ to the record name $\T$; the matching branch then
retrieves the projected label's type from the global label context $\labenv$
and instantiates it to match $\tv \to \t$. If $\tz$ is not a
record type, the match fails and the term is ill-typed.

For record expressions, we generate a fresh variable $\tvi$ for each field
assignment to capture the type of each $\ei$ as $\tvi$. The rest of the
constraint is deferred until the context determines the type of the whole
record type to be $\tz$. Once known, it is matched against the same pattern
$\parens{\cpatrcd \ct}$. If the match succeeds (\ie $\tz = \trcd \T \tys$ and
$\ct = \T$) the labels are instantiated to match the projection types $\t \to
\tvi$, and we additionally check that the domain of $\ct$ is exactly $\elabs$,
ensuring that every label is defined. Otherwise, the constraint is
unsatisfiable as intended.

Explicit records and projections, along with closed-world disambiguated terms,
bypass suspension and directly instantiate the appropriate labels.

\begin{example}
Consider the following example:
\begin{program}[input]
  let ex_4 r = let x = r.x in x + (r : point).y °\ocamlflags 10°
\end{program}
The typing constraint generated for \code{ex_4} contains the following,
where $\tv$ stands for the type of \code{r}:
\begin{mathpar}
  \cexists {\tv}
    \clet x \tvb
      {(\cmatch \tva \dots)}
      {\cinst x \tint \cand \cunif \tv {\ttpoint}}
\end{mathpar}
The suspended constraint can be discharged under our contextual semantics.
We apply the \Rule{Match-Ctx} rule with context
$\C$ equal to
\begin{mathpar}[inline]
  \clet \x \tvb \hole \capp \x \tint \cand
  \cunif \tv \ttpoint
\end{mathpar}.
Although the context includes a let-binding---which in practice
involves let-generalization---we can still deduce $\Cshape \C \tv
  {\ttpoint}$, since the erased context $\cerase
\C$ contains the unification constraint $\cunif \tv \ttpoint$.

This example illustrates that our formulation of suspended constraints
interacts nicely with let-polymorphism. Although the two features are
specified in a modular fashion, they are carefully crafted to work together,
as we further show in our next example.
\end{example}

\begin{example}\label{ex:backprop}
A subtle yet crucial feature of our semantics is its support for
\emph{backpropagation}:
\begin{program}[input,checkocaml]
let ex_1_1 = let getx r = r.x in getx one  °\ocamlflags 10°
\end{program}

As in the previous example, the type of \code{r} cannot be disambiguated from
the let-definition alone. There, the ambiguity was resolved when the type
($\tv$) was unified to a known type ($\ttpoint$) in the let-body. Here, the
situation is more subtle: an \emph{instance} of \code{getx}'s type scheme is
taken, which only satisfies the application (\ocaml!getx one!) if \code{r}
has either a variable type or the record type \ocaml{point}.  However, the
projection \code{r.x} would be ill-typed if \code{r} had a variable type
(unicity would fail), so \code{point} is the unique consistent solution. This
information therefore flows back from the instances in the let-body to the
let-definition, a phenomenon we call \emph{backpropagation}.

The constraint generated when typing
\code{ex_1_1} is:
\begin{mathpar}
\begin{tabular}{L.L}
  \cexists \tv {}
  &\clet {getx} \tvd
     {\cexists {\tvb, \tvc} \Parens {\strut
        \cunif \tvd {\tvb \to \tvc} \cand
	\cmatch \tvb \dots
        }}{}
    \cinst {getx} {(\ttpoint \to \tv)}
\end{tabular}
\end{mathpar}
With the context $\C$ equal to $\clet {getx}
\tvd {\cexists {\tvb, \tvc} \cunif \tvd {\tvb \to \tvc} \cand \hole}
{\capp {getx} {(\ttpoint \to \tv)}}$, we can show that the unicity
predicate $\Cshape \C \tvb \ttpoint$ holds.
For any $\gt$, the erasure $\cerase {\C \where {\cunif \tvb \gt}}$
is
\begin{mathpar}[inline]
\clet {getx}
\tvd {\cexists {\tvb, \tvc} \cunif \tvd {\tvb \to \tvc} \cand \cunif \tvb \gt}
{\capp {getx} {(\ttpoint \to \tv)}}
\end{mathpar}.
Since $getx$ is bound to the constraint abstraction
\relax $\cabs \tvd \exists \tvc.\uad \cunif \delta {(\gt \to \gamma)}$,
the instantiation
\relax $\capp {getx} (\ttpoint \to \tv)$
can only be satisfied when $\gt$ is equal to $\ttpoint$. This proves unicity,
hence the generated constraint for \code{ex_1_1} is satisfiable.
\end{example}

\subsection{Metatheory}
\label{sec/constraints/metatheory}

Constraint generation is sound and complete with respect to the typing judgment:
\begin{theorem}[Constraint generation is sound and complete]
  \label{thm:constraint-gen-is-sound-and-complete}
A closed \OML term $\e$ is typable if and only if
the constraint $\cexists \tv {\cinfer \e \tv}$ is satisfiable.
\end{theorem}

\section{Constraint solving}
\label{sec:solving}

\parcomment{Intro}

We now present a machine for solving constraints in our language. The solver
operates as a rewriting system on constraints $\c \csolve \cp$. Once no
further transitions are applicable, \ie $\c \cnsolve$, the constraint $\c$
is either a solved form---from which we can read off a most general
solution---or unsatisfiable (if the constraint
$\c$ is closed).

\paragraph{Conjunction notation}

Our rewriting rules frequently manipulate collections of constraints. We write
$\bar\c$ for a (possibly empty) sequence of constraints interpreted as their
conjunction $\cAnd_i \ci$. The empty sequence denotes $\ctrue$.

\subsection{Unification}

\begin{figure}[tp]
\begin{mathparsubfig}
  {fig:unification-syntax-and-semantics}
  {Syntax and semantics of unification problems.}
  \begin{minipage}[l]{0.7\textwidth}%
    \begin{bnfgrammar}[\noleftfill]%
      \entry[Unification problems]{\up}{
        \ctrue \and \cfalse \and \upa \cand \upb \and \cexists \tv \up \and \ueq
        \uad\strut
      } \\
      \entry[Multi-equations]{\ueq}{
        \eset \mid \cunif \t \ueq
      } \\
      \entry[Constraints]{\c}{
        \dots \and \ueq
      } \\
      \entry[Unification context]{\Up}{
        \hole
        \and \Up \cand \upb
        \and \upa \cand \Up
        \and \cexists \tv \Up
      }
    \end{bnfgrammar}
    \hfill
  \end{minipage}
  \hfill
  \vcenter{\hbox{
    \infer[Multi-Unif]
      {\forall {\t \in \ueq},\, \semenv(\t) = \gt}
      {\semenv \th \ueq}
  }}
\end{mathparsubfig}

\medskip

\begin{mathparsubfig}
  {fig:unification-algorithm}
  {Unification algorithm as a series of rewriting rules
   $\upa \unif \upb$. All shapes are principal.}
   \rewrite[U-Exists]
      {(\cexists \alpha \upa) \cand \upb }{ \tv \disjoint \upb}
      {\cexists \tv {\upa \cand \upb}}

    \rewrite[U-Cycle]
      {\up }{ \cyclic \up}
      {\cfalse}

    \rewrite[U-True]
      {\up \cand \ctrue}
      {}
      {\up}

    \rewrite[U-False]
      {\Up\where\cfalse }{ \Up \neq \hole}
      {\cfalse}

    \rewrite[U-Merge]
      {\cunif \tv \ueqa \cand \cunif \tv \ueqb}
      {}
      {\cunif \tv {\cunif \ueqa \ueqb}}

    \rewrite[U-Stutter]
      {\cunif \tv {\cunif \tv \ueq}}
      {}
      {\cunif \tv \ueq}

    \rewrite[U-Name]
      {\cunif {\pshapp \parens{\tys, \ti, \typs}} \ueq }
      { \tv \disjoint \tys, \typs, \ueq \\ \ti \notin \TyVars}
      {\cexists \tv {\cunif \tv \ti \cand \cunif {\pshapp \parens{\tys, \tv, \typs}} \ueq}}

    \rewrite[U-Decomp]
      {\cunif {\pshapp \tvs} {\cunif {\pshapp \tvbs} \ueq}}
      {}
      {\cunif {\pshapp \tvs} \ueq \cand \cunif \tvs \tvbs}

    \rewrite[U-Clash]
      {\cunif {\pshapp \tvs} {\cunif {\pshapp[\shp]\tvbs } \ueq }}{
       \sh \neq \shp}
      {\cfalse}

    \rewrite[U-Trivial]
      {\ueq }
      {|\ueq| \leq 1}
      {\ctrue}
\end{mathparsubfig}
\caption{Unification: syntax, semantics, and the rewriting-based solver $\upa \unif \upb$.}
\label{fig:unification}
\end{figure}

Our constraints ultimately reduce to equations between types, which we solve
using first-order unification. Like our solver, we specify unification as a
non-deterministic rewriting relation between \emph{unification problems} $\upa
\unif \upb$, that eventually reduces to a solved form $\hat\up$ or to $\cfalse$.

\parcomment{Unification problems and multi-equations}

Unification problems $\up$
(\cref{fig:unification-syntax-and-semantics}) are a restricted subset
of constraints, extended with \emph{multi-equations}
\citep*{\BBemlti}---a multi-set of types considered
equal. These generalize binary equalities: $\semenv$ satisfies
a multi-equation $\ueq$ if all of its members are mapped to a single
ground type $\gt$ (\Rule{Multi-Unif}). Multi-equations are
considered equal modulo permutation of their members.

The unification rules are listed in \cref{fig:unification-algorithm}. Rewriting
proceeds under an arbitrary context $\Up$, modulo $\alpha$-equivalence and
associativity and commutativity of conjunctions.
Our algorithm is largely standard~\citep*{\BBemlti}, with its main
novelty being the use of \emph{canonical principal shapes} in place of type
constructors. This uniform treatment of monotypes and \polytypes simplifies
unification and improves on the previous treatment of \polytype unification
\citep{\BBpolyml}.

\parcomment{Explanation of the rules}

We briefly summarize the role of each rule. \Rule{U-Exists} lifts existential
quantifiers, enabling applications of \Rule{U-Merge} and \Rule{U-Cycle} since
all multi-equations eventually become part of a single conjunction.
\Rule{U-Merge} combines multi-equations sharing a common variable and
\Rule{U-Stutter} removes duplicate variables. \Rule{U-Decomp} decomposes equal
types with matching shapes into equalities between their subcomponents, while
\Rule{U-Clash} detects shape mismatches that result in failure. \Rule{U-Name}
introduces fresh variables for subcomponents, ensuring unification operates on
\emph{shallow terms}, making sharing of type variables explicit and avoiding
copying types in rules such as \Rule{U-Decomp}. \Rule{U-True} and
\Rule{U-Trivial} eliminate trivial constraints, and \Rule{U-False} propagates
failure.
Finally, \Rule{U-Cycle} implements the \emph{occurs check}, ensuring that a
type variable does not occur in the type it is being unified with. This is a
necessary condition for unification, as it would otherwise lead to infinite
types.
This is formalized by the relation $\tv \prec_\up \tvb$, which indicates that
$\tv$ occurs in a non-variable type $\t$ equated to $\tvb$ in $\up$. Concretely, this holds when 
$\up = \Up\where{\cunif \tvb {\cunif \t \ueq}}$,
$\tv \in \fvs \t$, $\t \notin \TyVars$ and $\tv, \tvb \disjoint \bvs \Up$.
A unification problem $\up$ is said to be cyclic, written $\cyclic \up$, if
$\tv \prec_\up^+ \tv$ for some $\tv$, where $\prec_\up^+$ denotes the 
transitive closure of $\prec_\up$. 


\begin{definition}[Solved form $\hat\up$]
\label{def:solved-form}
  We write $\hat\up$ for constraints in \emph{solved form}, that is,
  constraints of the form $\cexists \tvs {\ueqs}$, where:
\begin{enumerate*}
  \item
    each $\ueqi$ contains at most one non-variable type;
  \item
    each type variable may occur as a member of at most one multi-equation
    $\ueqi$;
  \item the constraint is acyclic.
\end{enumerate*}\relax
\end{definition}

\subsection{Administrative constraints}
\label{sec/solving/admin}

\parcomment {Intro: new administrative constraint formers!}

Section \cref{sec:constraints} introduced the constraint language used for
constraint generation (\cref{sec/oml/constraints}). In this section we extend
the language with a small set of \emph{administrative constraints}: auxiliary
constructs that never appear in generated constraints, but are introduced
internally by the solver to manage generalization and incremental
instantiation.

\parcomment {Point: What for?}

These administrative constraints enable the solver to represent and refine
\emph{partial type schemes} as solving progresses, capturing intermediate
generalization states and tracking how instantiations evolve as suspended
constraints are discharged.

\begin{mathparfig}[tp]
  {fig:constraint-let-regions}
  {Syntax and semantics of region-based $\Let$ and
   incremental instantiation constraints.}
  \begin{bnfgrammar}
    \entry[Instantiation variables]{\inst}{}\\
    \entry[Constraints]{\c}{
       \dots \let \and \wand
       \and \cletr \x \tv \tvs \ca \cb
       \and \cexistsi \inst \x \c
       \and \cpinst \inst \tv \t
    }\\[1ex]
    \entry[Ground regions]{\gr}{
      \greg \gt \semenv \qquad\qquad\! \parens{\gr \in \GroundRegion}
    }\\
    \entrysubseteq[Sets of ground regions]{\gabsr}
      {\GroundRegion}\\
    \entry[Semantic environments]{\semenv}{
      \dots \let \and \wand
      \and \semenv\where{\x \is \gabsr}
      \and \semenv\where{\inst \is \semenvp}
    }
  \end{bnfgrammar}
  \par
  \semenv(\cabsr \tv \tvs \c) \uad\eqdef\uad
    \set {\greg \gt {\semenv\where{\tv \is \gt, \tvs \is \gts}} \in \GroundRegion
         : \semenv\where{\tv \is \gt, \tvs \is \gts} \th \c
         }
  \par
  \infer[LetR]
    {\gabsr = \semenv(\cabsr \tv \tvs \ca) \\
     \gabsr \neq \emptyset \\\\
     \semenv\where{\x \is \gabsr} \th \cb}
    {\semenv \th \cletr \x \tv \tvs \ca \cb}

  \infer[AppR]
    {\greg {\semenv(\t)} \wild \in \semenv(\x)  }
    {\semenv \th \capp \x \t}

  \infer[Exists-Inst]
    {\greg \wild \semenvp \in \semenv(\x) \\\\
     \semenv\where{\inst \is \semenvp} \th \c}
    {\semenv \th \cexistsi \inst \x \c}

  \infer[Incr-Inst]
    {\semenv(\inst)(\tv) = \semenv(\t) }
    {\semenv \th \cpinst \inst \tv \t}
\end{mathparfig}

\paragraph{Regional let-constraints}

\parcomment{Recall the standard solving form for let-constraints}

As is standard in constraint-based formulations of \ML type inference
\citep{\BBemlti}, generalization-related rules operate on constraints
of the form $\clet \x \tv {\exi \tvs \ca} \cb$. The outermost existentially
quantified variables $\tvs$ correspond to the \emph{potentially generalizable}
variables.
\parcomment{What is a region?}%
%
We refer to this existential prefix $\exists \tvs$ as a \emph{region}.
Intuitively, a region delimits the scope of variables that may be generalized
once the abstraction has been fully solved.

\parcomment{Explicit regions}%
To make regions explicit, we introduce regional let-constraints $\cletr \x
\tv \tvs \ca \cb$ (\cref{fig:constraint-let-regions}), where $\tv$ is the
\emph{root} of the region and $\tvs$ are auxiliary existential variables.
The order of $\tvs$ is immaterial; regions are considered equal
up to permutation of these variables.

\parcomment{Semantics}

Satisfiability of regional let-constraints is defined in
\cref{fig:constraint-let-regions}. The semantics of an abstraction with a
region, written $\semenv(\cabsr \tv \tvs \c)$, is a set $\gabsr$ of
\emph{ground regions} that satisfy $\c$. A ground region $\gr$ is a pair $\greg
\gt \semenv$; it satisfies the constraint abstraction $\cabsr \tv \tvs \c$ if
$\semenv$ satisfies $\c$, and $\gt$ is $\semenv(\tv)$. Intuitively, the region
environment $\semenv$ carries a ground type for each inference variable $\tvb$
occurring in $\c$; in particular, a ground instance for $\tv$ and $\tvs$.

The rules \Rule{LetR} and \Rule{AppR} are the \emph{regional counterparts} of
the \Rule{Let} and \Rule{App} rules introduced in \cref{sec:constraints}. They
are semantically equivalent to their non-regional forms, except that each
operates over sets of ground regions ($\gabsr$) rather than sets of ground
types ($\gabs$). The additional environment $\semenv$ carried by these ground
regions $\greg \gt \semenv$ is semantically inert for \Rule{LetR} and
\Rule{AppR} themselves, but will become important later, in the semantics of
\emph{incremental instantiation} constraints (\cref{sec/solver/incr-inst}).

Regional let-constraints strictly generalize ordinary let-constraints, as
captured by the equivalence:
\begin{mathpar}
  \clet \x \tv {\exi \tvs \ca} \cb \Wide\cequiv \cletr \x \tv \tvs \ca \cb
\end{mathpar}

\paragraph{The trouble with lets}


\parcomment{Previous ways of solving abstractions (aka generalization)}


Solving let-constraints is deceptively difficult.
Naively, let-constraints (or \emph{generalization} constraints) could be
solved by copying constraints:
\begin{mathpar}
\hfil
\rewrite
    {\cletr \x \tv \tvs \ca {\C\where{\capp \x \t}}}
    {\tv, \tvs \disjoint \t \\ x \disjoint \bvs \C}
    {\cletr \x \tv \tvs \ca {\C\where{\cexists {\tv, \tvs} \cunif \tv \t \cand \ca}}}
\eqno
{\DefTirName{S-Let-App-Beta}}
\end{mathpar}
This rule, due to \citet*{\BBemlti}, resembles $\beta$-reduction in
some presentations of explicit substitutions. While this rule is sound
when $\ca$ is a \emph{simple} constraint, it becomes unsound for abstractions
containing suspended constraints.

\begin{local}
\def \xgetx{{\ttlab {getx}}}
\def \xdiag{{\ttlab {diag}}}
\def \elabx{{\ttlab x}}
\def \tvp {{\tv_{\ttlab p}}}
\def \tvgp {{\tv_{\ttlab {gp}}}}
\def \tvgetx {{\tva_\xgetx}}
\def \tvgetxret {\tvc}
\def \tgpoint {\ttlab {gpoint}}

\parcomment{Example}
To see why, consider again the example \code{ex_8} from \cref{sec/overview/omni}:%
\begin{program}[input]
type 'a gpoint = { x : 'a; y : 'a }
let diag (n : 'a) : 'a gpoint = { x = n; y = n }
°\halfline°
let ex_8 gp =                                           °\ocamlflags 20° 
  let getx p = p.x in getx (diag 42), (getx gp : float)
\end{program}
As explained earlier, this program is clearly well-typed: the type of \code{p}
is unambiguously determined as \code{'b gpoint} through \emph{backpropagation}
from the first application of \code{getx} (\code{getx (diag 42)}).
The (simplified) generated constraint for \code{ex_8} contains:
\begin{mathpar}
  \begin{array}{rl}
  \cexists {\tvgp, \tvd}
    &\Let \xgetx = \cabs \tvgetx
  {\cexists {\tvp, \tvgetxret}
    {\bigwedge \PARENS {
      \cunif \tvgetx {\tvp \to \tvgetxret} \\
      \cmatch \tvp {\cbranch {(\cpatrcd \ct)}
         {\labfrom {\ttlab \x} \ct \leq \tvp \to \tvgetxret}}
       }}}
  \\[2ex]
    & \In {\capp \xgetx \parens{\trcd \tgpoint \tint \to \tvd}
      \cand \capp \xgetx \parens{\tvgp \to \tfloat}}
  \end{array}
\end{mathpar}
  If we now apply \Rule{S-Let-App-Beta} to both applications of $\xgetx$
  (removing the let-constraint for concision), we obtain:
\begin{mathpar}
  \begin{array}{rl}
    \cexists {\tvgp, \tvd} {} &

  \cexists {\tvp_1, \tvgetxret_1}
\bigwedge \PARENS {
  \cunif {\trcd \tgpoint \tint \to \tvd} {\tvp_1 \to \tvgetxret_1} \\
      \cmatch {\tvp_1} {\cbranch {(\cpatrcd \ct)}
         {\labfrom {\ttlab \x} \ct \leq \tvp_1 \to \tvgetxret_1}}
    } \\

    \cand & \cexists {\tvp_2, \tvgetxret_2}
\bigwedge \PARENS {
  \cunif {\tvgp \to \tfloat} {\tvp_2 \to \tvgetxret_2} \\
      \cmatch {\tvp_2} {\cbranch {(\cpatrcd \ct)}
         {\labfrom {\ttlab \x} \ct \leq \tvp_2 \to \tvgetxret_2}}
    }
  \end{array}
\end{mathpar}
The first conjunct is satisfiable, since $\tvp_1$ is unified with $\trcd
\tgpoint \tint$, thereby discharging the match constraint on $\tvp_1$.  The
second, however, is not: $\tvp_2$ remains underdetermined, leaving its match
constraint unsatisfiable. Thus, although the original constraint was
satisfiable, the application of \Rule{S-Let-App-Beta} makes it
unsatisfiable.  By copying the abstraction, we lose the essential
\emph{sharing} between both instantiations---namely, that both copies of
$\tvp$ ($\tvp_1$ and $\tvp_2$) must have the same shape $\any \tvc \trcd
\tgpoint \tvc$. This loss of sharing is precisely why \Rule{S-Let-App-Beta}
is unsound in the presence of suspended constraints.

\parcomment{What about treating let-constraints monomorphically?}

A tempting alternative, used in
\citet*{\BBoutsidein} and
\citet*{benevs2025simple}, is to treat the let-bindings \emph{monomorphically},
sharing $\tvp$ directly between both applications. However, this is incomplete
with respect to \OML's type system. It forbids local let-polymorphism (commonly 
used in \OCaml) and fails to typecheck \code{ex_8}: the two calls to
\code{getx} require different instantiations of $\tvp$---\code{int gpoint} and
\code{float gpoint}, respectively.

\parcomment{Generalization in the context of constraint solving}

Even setting soundness aside, \Rule{S-Let-App-Beta} is inefficient. Each
application duplicates constraint solving work for the same abstraction.
A more efficient approach is to \emph{solve once and reuse}: first solve the
abstraction once---\eg reducing it to $\cabsr \tv \tvs \ueqs$, where $\tvs$ are
generalizable variables---and then reuse the result at each instantiation site
by only copying the solved constraint $\ueq$. This mirrors the generalization
and instantiation steps of \ML inference algorithms such as $\mathcal{W}$:
$\cabsr \tv \tvs \ueqs$ corresponds to the type scheme $\tfor \tvs
{\sub(\tv)}$, where $\sub$ is the most general unifier of $\ueqs$.
\citet*{\BBemlti} formalize this connection, and the optimized
treatment is naturally expressed as a strategy on top of their
\Rule{S-Let-App-Beta} rule.

\parcomment{The two problems and the solution}

We therefore face two related challenges: \begin{enumerate*}
  \item handle instantiation soundly in the presence of suspended constraints, and
  \item to avoid redundant work by reusing \emph{partial} results across instantiations.
\end{enumerate*}

To address both, we introduce \emph{partial type schemes}, our second novel
mechanism for omnidirectional inference. Partial type schemes are type schemes
that delay commitment to certain quantifications (\eg~$\tvp$ and~$\tvgetxret$).
Such \emph{partially generalized} variables are treated as generalized, but can
be incrementally refined in future as suspended constraints are discharged.

Returning to our running example, we begin with the partial type scheme $\all
{\tvp, \tvgetxret} \tvp \to \tvgetxret$ for \code{getx}, since the suspended
match constraint in its abstraction cannot yet be solved. We continue solving
the body of the let-constraint, tracking every instances of this partial
scheme. As type information flows back---\eg when the shape of $\tvp$ becomes
known via backpropagation from the first application of \code{getx}---the
scheme is refined to $\all {\tvb} \trcd \tgpoint \tvb \to \tvb$. The second
application of \code{getx} then updates its instantiation accordingly, unifying
$\tvp_2$ with $\trcd \tgpoint \tvb_2$ and $\tvgetxret_2$ with $\tvb_2$.

\end{local}

\paragraph{Incremental instantiation}
\label{sec/solver/incr-inst}

\parcomment {Intro partial instantiations}

To support partial type schemes, we also extend the constraint language with
\emph{incremental instantiation constraints}
(\cref{fig:constraint-let-regions}).
We introduce two new constraint formers:
\begin{enumerate}
  \item
    $\cexistsi \inst \x \c$, which binds a fresh instantiation $\inst$ of $\x$'s
    region within $\c$, and
  \item
    $\cpinst \inst \tv \t$, which asserts that the copy of $\tv$ in $\inst$
    equals~$\t$.
\end{enumerate}
The instantiation variable $\inst$ is required to ensure all incremental
instantiations $\cpinst \inst \tv \t$ are solved uniformly.


\parcomment {Conceptual model}
Within the solver, we view incremental instantiations as markers indicating
which parts of the abstraction still need to be copied, and abstractions
themselves are treated as partial type schemes.

\parcomment{Remind why we need this}

This mechanism enables efficient handling of constraint instantiations: solved
parts are reused immediately, while suspended constraints can be solved later,
further refining the abstraction and propagating new equations to all
instantiation sites.

\parcomment{Semantics}

We now turn to the semantics of incremental instantiations
(\cref{fig:constraint-let-regions}). The existential constraint $\cexistsi
\inst \x \c$ is satisfiable (\Rule{Exists-Inst}) if $\c$ is satisfiable with
$\inst$ bound to one of the region environments $\semenvp$ in the
interpretation of $\x$.
An incremental instantiation $\cpinst \inst \tv \t$ is satisfiable
(\Rule{Incr-Inst}) exactly when $\inst$'s instance of $\tv$
(\ie $\semenv(\inst)(\tv)$) is equal to $\t$.

\parcomment {Scoping}

When a regional abstraction $\cletr \x \tv \tvs \ca \cb$ contains an
incremental instantiation constraint $\cpinst \inst \tvb \t$, and $\inst$ is
an instantiation of $\x$, the variable $\tvb$ may range over any free variable
of $\ca$---including $\tv$ and the regional variables $\tvs$. In other
words, the regional let-constraint ${\cletr \x \tv \tvs \ca \cb}$ provides a
non-obvious scoping rule: the variables $\tv$ and $\tvs$ are also bound in
$\cb$, but may occur there only within incremental instantiations of
$\x$. This subtlety justifies introducing an explicit syntax for regional
let-constraints, rather than encoding them as $\clet \x \tv {\exi \tvs \ca}
\cb$.

\subsection{Solving rules}


We now gradually introduce the rules of the constraint solver itself
(\cref {fig:solver-basic,fig:solver-schemes,fig:solver-susp};
for a complete view with all rules placed together, see Appendix~\cref{app/ref/solver}).
These rules define a non-deterministic rewriting
system, operating modulo $\alpha$-equivalence, and the associativity and
commutativity of conjunction. Rewriting takes place under an arbitrary
one-hole constraint context $\C$.
%
A constraint $\c$ is satisfiable if it rewrites to a solved form $\hat\up$
(\cref{def:solved-form}); otherwise it gets stuck---in particular, $\cfalse$ is
considered a stuck constraint. To characterize all stuck constraints, we 
introduce \emph{normal forms}.
Every solved form is a normal form, but not every normal form is solved. 

\begin{definition}[Normal forms]
  \label{def/normal-forms}
  A constraint $\c$ is in \emph{normal form}, written $\hat\c$, if no 
  rewriting rules apply, \ie $\c \cnsolve$. Similarly, a constraint context 
  $\C$ is in \emph{normal form}, written $\hat\C$, if $\C\where\ctrue$ is in
  normal form. 
\end{definition}

\paragraph{Basic rules}

\begin{mathparfig}[tp]
  {fig:solver-basic}
  {Basic rewriting rules $\ca \csolve \cb$.}
  \rewrite[S-Unif]
    {\upa}
    {\upa \unif \upb}
    {\upb}

  \rewrite[S-False]
    {\C\where\cfalse}
    {\C \neq \hole}
    {\cfalse}

  \rewrite[S-Let]
    {\clet \x \tv \ca \cb}
    {}
    {\cletr \x \tv \eset \ca \cb}

  \rewrite[S-Exists-Conj]
    {(\cexists \alpha \ca) \cand \cb}
    {\tv \disjoint \cb}
    {\cexists \tv {\ca \cand \cb}}

  \rewrite[S-Let-ExistsLeft]
    {\cletr \x \tv \tvs {\cexists \tvb \ca} \cb}
    {\tvb \disjoint \tv, \tvs, \cb}
    {\cletr \x \tv {\tvs, \tvb} \ca \cb}

  \rewrite[S-Let-ExistsRight]
    {\cletr \x \tv \tvs \ca {\cexists \tvb \cb}}
    {\tvb \disjoint \tv, \tvs, \ca}
    {\cexists \tvb {\cletr \x \tv \tvs \ca \cb}}

  \rewrite[S-Let-ConjLeft]
    {\cletr \x \tv \tvs {\ca \cand \cb} \cc}
    {\ca \disjoint \tv, \tvs}
    {\ca \cand \cletr \x \tv \tvs \cb \cc}

  \rewrite[S-Let-ConjRight]
    {\cletr \x \tv \tvs \ca (\cb \cand \cc)}
    {\x \disjoint \cb}
    {\cb \cand \cletr \x \tv \tvs \ca \cc}
\end{mathparfig}

\parcomment{Unification}


\cref{fig:solver-basic} contains a selected set of basic solving rules, drawn
from the standard repertoire of small-step constraint solvers for \ML type inference.
\parcomment{Explanation of rules}
\Rule{S-Unif} invokes the unification algorithm on the
current unification problem. The unification algorithm itself is treated as a
black box by the solver, so the system could be extended with any
equational theory of types implemented by the unification algorithm.
\Rule{S-Let} rewrites let-constraints into regional form.
\Rule{S-Exists-Conj} lifts existentials across conjunctions;
\Rule{S-Let-ExistsLeft} and \Rule{S-Let-ExistsRight} lift existentials across
let-binders; \Rule{S-Let-ConjLeft}, \Rule{S-Let-ConjRight} hoist constraints
out of let-binders when they are independent of the local variables.
Collectively, these lifting rules normalize the structure of each region into a
block of existentially bound variables, whose body consists of a conjunction of
solved multi-equations followed by a residual constraint---typically an
instantiation, let-binding, or suspended constraint.

\parcomment{\OML constraints do not need dedicated rules}

\OML-specific constraints, such as the label and \polytype instantiation
constraints ($\labfrom \elab \ct \leq \ta \to \tb$, $\cleq \cscm \t$,
\etc), require no special treatment in our solver. Once their pattern variables
are substituted---after solving a match constraint---they are desugared into
constraints already handled by the solver.

\begin{mathparfig}[tp]
  {fig:solver-schemes}
  {Solving rules for let-constraints and instantiations.}
\rewrite[S-Inst-Copy]
  {\cletr \x \tv \tvs {\c} \C\where{\cpapp \x \tvb \tvc \inst} \\}
   {\acyclic {\c} \\
    \x \disjoint \bvs \C \\\\
    \c = \cp \cand \cunif \tvb {\cunif {\shapp \tvbs} \ueq}\\
    \tvb \in \reg \tv \tvs \\
    \tvbs' \disjoint \tvb, \tvc, \tvbs}
  {{\begin{array}{l}\cletr \x \tv \tvs {\c}
      \C\where{\cexists {\tvbs'}
      \cunif \tvc {\shapp \tvbs'} \cand \overline{\cpapp \x {\tvb} {\tvb'} \inst}}\end{array}}}

\rewrite[S-Inst-Unif]
  {\cpinst \inst \tvb \tvca \cand \cpinst \inst \tvb \tvcb}
  {}
  {\cpinst \inst \tvb \tvca \cand \cunif \tvca \tvcb}

\rewrite[S-Inst-Poly]
  {\cletr \x \tv {\tvs} {\ueqs \cand \c}
      {\C\where{\cpapp \x \tvb \tvc \inst}}}
  {\cfor \tvb \cexists {\tvbs} {\ueqs} \cequiv \ctrue \\\\
   \tvb, \tvbs \subseteq \tv, \tvs \\
   \tvb, \tvbs \disjoint \c \\
   \inst(\tvb) \disjoint \insts \C \\
   \x \disjoint \bvs \C}
  {\cletr \x \tv {\tvs} {\ueqs \cand \c} {\C\where\ctrue}}

\rewrite[S-Inst-Mono]
  {\cletr \x \tv \tvs \c {\C\where{\cpapp \x \tvb \tvc \inst}} \\}
  {\tvb \notin \reg \tv \tvs \\ \x, \tvb \disjoint \bvs \C}
  {\cletr \x \tv \tvs \c {\C\where{\cunif \tvb \tvc}}}
\\
\rewrite[S-Let-AppR]
  {\cletr \x \tv \tvs \c {\C\where{\capp \x \t}} \\}
  {\tvc \disjoint \t \\
   \x \disjoint \bvs \C}
  {{\begin{array}{l}\cletr \x \tv \tvs \c
     {\C\where{\cexistsi {\tvc, \inst} \x
              {\cpinst \inst \tv \tvc \cand \cunif \tvc \t}}}\end{array}}}

\rewrite[S-Let-Solve]
  {\cletr \x \tv \tvs \ueqs \c\\}
  {\cexists {\tv, \tvs} \ueqs \cequiv \ctrue \\
   \x \disjoint \c}
  {\c}

\rewrite[S-Compress]
  {\cletr \x \tv {\tvs, \tvb} {\ca \cand \cunif \tvb {\cunif \tvc \ueq}} {\cb}}
  {\tvb \neq \tvc}
  {\cletr \x \tv {\tvs}
     {\ca\where{\tvb \is \tvc}
      \cand \cunif \tvc {\ueq\where{\tvb \is \tvc}}}
     {\cb\where{\x(\tvb) \is \tvc}}}

\rewrite[S-Exists-Lower]
  {\cletr \x \tv {\tvas, \tvbs} \ca \cb\\}
  {\th \cdetermines {\cexists {\tv, \tvas} \ca} \tvbs}
  {\cexists \tvbs \cletr \x \tv \tvas \ca \cb}

\inferrule
  [Det-Dom]
  {\tvc \disjoint \tvbs, \tvas \\ \tvs \subseteq \fvs \ueq}
  {\th \cdetermines {\cexists \tvbs \c \cand \cunif \tvc \ueq} \tvs}

\inferrule
  [Det-Esc]
  {\fvs \t \disjoint \tvs, \tvbs}
  {\th \cdetermines {\cexists \tvbs \c \cand \cunif \tvs {\cunif \t \ueq}} \tvs}
\end{mathparfig}

\paragraph{Incremental instantiations}
\parcomment{Solving incremental instantiations}

\cref{fig:solver-schemes} describes the solving rules in charge of incremental
instantiation.
Incremental instantiation constraints are reduced using the following rules:
\begin{enumerate}

\item
    \Rule{S-Inst-Copy} copies the shape of a type to the
    instantiation site, 
    introducing fresh variables for each subcomponents and marking them with
    corresponding instantiation constraints.
    We write $\cpapp \x \tvb \t \inst$ as a shorthand for $\cpinst \inst \tvb \t$
    when $\inst$ is bound with $\exists \inst^\x$ in the context. To ensure
    termination, the abstraction must contain acyclic types.

  \item
    \Rule{S-Inst-Unif} unifies two instantiations if they both
    refer to the same source variable $\tvb$ at the same instantiation site
    $\inst$. 
    
\end{enumerate}
There are three cases in which an instantiation constraint is eliminated:
\begin{enumerate}
  \item
   \sloppy
    A base type (\ie a nullary shape) is copied and no further
    instantiations are needed (\Rule{S-Inst-Copy}).
    
  \item
    The copied variable $\tvb$ is polymorphic, and thus the instantiation
    constraint imposes no restriction (\Rule{S-Inst-Poly}), provided no
    other instantiations of $\tvb$ exist for the same instantiation $\inst$ (if
    not, then apply \Rule{S-Inst-Unif}).

  \item
    The copy is monomorphic and in scope, so we unify it directly
    (\Rule{S-Inst-Mono}).
\end{enumerate}


\parcomment{S-Let-AppR}

\Rule{S-Let-AppR} rewrites an instantiation constraint $\capp \x \t$ introducing an
incremental instantiation constraint $\cpinst \inst \tv \tvc$. Here, $\inst$ is
a fresh instantiation of $\x$, $\tv$ is the \emph{root} of $\x$'s region, and
$\tvc$ is a fresh alias for $\t$. We introduce $\tvc$ explicitly, since our
rewriting rules for incremental instantiations generally assume that the copied
type is a variable rather than an arbitrary type.


\paragraph{Let-constraints}

\cref{fig:solver-schemes} also presents the rules governing generalization.
\parcomment{Cleaning up partial instantiations and let-constraints}
\Rule{S-Let-Solve} removes a let-constraint when the bound term
variable is unused and the abstraction is satisfiable. \Rule{S-Compress}
determines that a regional variable $\tvb$ is an an alias for $\tvc$. We
replace every free occurrence of $\tvb$ with $\tvc$---\emph{including} the
domains of any incremental instantiation constraints, written as the
substitution $\where{\x(\tvb) \is \tvc}$.

Conceptually, \Rule{S-Compress} acts a variable-level analogue of
\Rule{S-Inst-Copy}: both rules copy solved constraints $\ueq$ from an
abstraction to its instantiation constraints. While \Rule{S-Inst-Copy}
propagates the head shape of a multi-equation, \Rule{S-Compress} propagates
equalities between variables within a multi-equation, thereby enabling
subsequent applications of \Rule{S-Inst-Unif}.

\parcomment{Lowering}

\Rule{S-Exists-Lower} implements the non-trivial case of lowering
existentials across let-binders. It identifies a subset of variables in
the region of a let-constraint that are unified with variables from
outside the region. Such variables are considered monomorphic and thus
cannot be generalized; they can instead be safely lowered to the outer
scope.

\parcomment {Determines}

This is the case when the types of $\tvbs$ are \emph{determined} in a unique
way. In short, $\c$ determines $\tvbs$ if and only if the solutions for
$\tvbs$ are uniquely fixed by the solutions to other variables in $\c$.
\begin{definition}
  $\cdetermines \c \tvbs$ if and only if every ground assignments $\semenv$
  and $\semenvp$ that satisfy (the erasure of) $\c$ and coincide outside of
  $\tvbs$ coincide on $\tvbs$ as well.
  \belowdisplayskip 0em
  \begin{mathpar}
    \cdetermines \c \tvb \uad\eqdef\uad \all {\semenv, \semenvp} \uad
      \semenv \th \cerase \c
      \wedge \semenvp \th \cerase \c
      \wedge \semenv =_{\setminus \tvbs} \semenvp
      \implies
      \semenv = \semenvp
  \end{mathpar}
\end{definition}
\parcomment {How the determines relation corresponds to ML}
Conceptually, this corresponds to the negation of the generalization
condition in \ML: a type variable \emph{cannot} be generalized if it appears
in the typing context. In the constraint setting, it \emph{cannot} be
generalized if it depends on variables from outside the region. For
instance, $\cexists \tvb \cunif \tv {\tvb \to \tvc}$ determines $\tvc$,
as $\tv$ is free and therefore constrains the solution of $\tvc$.

\parcomment{How to decide the relation}

To decide when $\cdetermines \c \tvs$, we introduce the judgment $\th
\cdetermines \c \tvs$, which syntactically proves that $\tvs$ are determined
in $\c$.
If $\c$ is of the form $\cexists \tvbs \cp$ where $\tvbs \disjoint \tvs$, then
we search for a multi-equation $\ueq$ in $\cp$ of the form:
\begin{enumerate*}
  \item[(\Rule{Det-Dom})]
    $\cunif \tvc \ueq'$ where $\tvc \disjoint \tvs, \tvbs$ and
    $\tvs \subseteq \fvs {\ueq'}$, or
  \item[(\Rule{Det-Esc})]
    $\cunif \tvs {\cunif \t \ueq'}$ where
    $\fvs \t \disjoint \tvs, \tvbs$.
\end{enumerate*}
This syntactic relation coincides with the semantic definition of
determinacy whenever $\c$ is in solved form. Otherwise, it is a
sound approximation of the semantic definition.

\parcomment{Why we lower?}

Lowering such variables improves solver efficiency. It avoids unnecessary
duplication of work that would otherwise occur via \Rule{S-Inst-Copy}. Because
these variables are determined by monomorphic ones, lowering allows the solver
to apply \Rule{S-Inst-Mono} directly at instantiation sites, rather than
duplicating equivalent constraints that ultimately express the same semantic
fact---that the variable is monomorphic and shared across instances.

\paragraph{Suspended match constraints}

\begin{mathparfig}[tp]
  {fig:solver-susp}
  {Rewriting rules for suspended match constraints.}
\rewrite[S-Match-Ctx]
  {\C\where{\cmatch \t \cbrs}}
  {\th \Cshape \C \t \sh}
  {\C\where{\cmatched \t {\sh} \cbrs}}

\infer[Uni-Var]
  {\color{gray}\tv \disjoint \bvs \Cb}
  {\th \Cshape {\Ca\where{\cunif \tv
     {\cunif \t \ueq} \cand \Cb\where{\hole}}} \tv {\shape \t}}

\infer[Uni-Type]
  {{\color{gray}\t \notin \TyVars}}
  {\th \Cshape \C \t {\shape \t}}

\infer[Uni-BackProp]
  {\th \Cshape{\parens{\cletr \x \tv \tvs {\Ca\where{\ctrue}}
  {\Cb\where{\cpapp \x \tvp \tvc \inst \cand \hole}}}} \tvc \sh \\
   \color{gray}\tvp \in \tv, \tvs \\
   \color{gray}\x \disjoint \bvs \Cb \\
   \color{gray}\tvp \disjoint \bvs \Ca}
  {\th \Cshape
     {\parens{\cletr \x \tv \tvs {\Ca\where{\hole}}
                  {\Cb\where{\cpapp \x \tvp \tvc \inst}}}}
     \tvp \sh}
\end{mathparfig}

\cref{fig:solver-susp} describes the rules for discharging suspended
constraints.
\parcomment{S-Match-Ctx}
%
\break \Rule{S-Match-Ctx} solves suspended
match constraints. It would not be effective to allow rewriting whenever the
unicity condition $\Cshape \C \t \sh$ holds, because it is not a-priori
feasible to check or decide this semantic condition which quantifies over all
solutions. Instead we introduce a restricted, decidable approximation via
 three syntactic ``unicity rules'' that form the judgment $\th \Cshape \C \t \sh$.

The unicity rule \Rule{Uni-Type} applies when $\t$ is a non-variable type $\t$,
in which case the shape is simply $\shape \t$. \Rule{Uni-Var} applies when the
scrutinee is a variable $\tv$ and the context establishes that $\tv$ is equal
to some non-variable type $\t$ by exhibiting an equality $\cunif \tv {\cunif \t
\ueq}$ and $\t$ is a non-variable type. In this case, the shape of $\tv$ is
$\shape \t$.

Finally, \Rule{Uni-BackProp} expresses
\emph{backpropagation}, previously illustrated in \cref{ex:backprop}. In
particular, the shape of a regional variable can sometimes be determined
from its instantiations. If an abstraction contains a regional variable
$\tvp$, and the constraint context includes an incremental instantiation
$\cpapp \x \tvp \tvc \inst$ such that the instance $\tvc$ of $\tvp$ has
the unique shape $\sh$, then $\tvp$ must also have shape $\sh$, as any
other shape would render the instantiation unsatisfiable. This rule is
well-founded because the regional depth of the hole strictly decreases
in the premise.

For a normalized context $\hat\C$ (\cref{def/normal-forms}), this syntactic
condition exactly characterizes the semantic definition of unicity: $\Cshape \C
\tau \sh$ holds if and only if $\th \Cshape \C \tau \sh$ is derivable
(Appendix~\cref{lem:unicity-completeness}).

\begin{example}
\begin{local}
\def \xgetx{{\ttlab {getx}}}
\def \xdiag{{\ttlab {diag}}}
\def \elabx{{\ttlab x}}
\def \tvp {{\tv_{\ttlab p}}}
\def \tvpc {{\tvc_{\ttlab p}}}
\def \tvpcp {{\tvc_{\ttlab p}'}}
\def \tvgp {{\tv_{\ttlab {gp}}}}
\def \tvgetx {{\tva_\xgetx}}
\def \tvcgetx {{\tvc_\xgetx}}
\def \tvcpgetx {{\tvc_\xgetx'}}
\def \tvgetxret {\tvc}
\def \tvcgetxret {\tvca}
\def \tvcpgetxret {\tvcb}
\def \tgpoint {\ttlab {gpoint}}
\def \qcommaq {\quad,\quad}
\def \commaq {,\quad}

Let us consider the (simplified) constraint generated for \ocaml{ex_8}, discussed in \cref{sec/solving/admin}. 
For readability, we begin by factoring out the following invariant two-hole constraint context $\C$:
\begin{mathpar}
  \C\Where{\hole_1, \hole_2} \quad \eqdef \quad   \begin{array}{rl}
  \cexists {\tvgp, \tvd}
    &\Let \xgetx = \cabs \tvgetx
  {\cexists {\tvp, \tvgetxret}
    { \cunif \tvgetx {\tvp \to \tvgetxret} \cand \hole_1 }}
  \\[1ex]
    & \In {\hole_2
      \cand \capp \xgetx \parens{\tvgp \to \tfloat}}
  \end{array}
\end{mathpar}
The initial constraint can then be written as follows:
\begin{mathpar}
  \begin{array}{cl}
  \cexists {\tvgp, \tvd}
    &\Let \xgetx = \cabs \tvgetx
  {\cexists {\tvp, \tvgetxret}
    {\bigwedge \PARENS {
      \cunif \tvgetx {\tvp \to \tvgetxret} \\
      \cmatch \tvp {\cbranch {(\cpatrcd \ct)}
         {\labfrom {\ttlab \x} \ct \leq \tvp \to \tvgetxret}}
       }}}
  \\[2ex]
    & \In {\capp \xgetx \parens{\trcd \tgpoint \tint \to \tvd}
      \cand \capp \xgetx \parens{\tvgp \to \tfloat}} \\[3ex]
    \equiv &\C\Where{
        \cmatch \tvp {\cbranch {(\cpatrcd \ct)}
         {\labfrom {\ttlab \x} \ct \leq \tvp \to \tvgetxret}},
        \capp \xgetx \parens{\trcd \tgpoint \tint \to \tvd}
       }
  \end{array}
\end{mathpar}
We now reduce the constraint step by step using the rules 
outlined above.
\begin{local}
\def \arraystretch{1.2}
\small
\begin{mathpar}
\begin{tabular*}{\linewidth}{;;L;;L~R}
  &\C\Where{
        \cmatch \tvp {\cbranch {(\cpatrcd \ct)}
         {\labfrom {\ttlab \x} \ct \leq \tvp \to \tvgetxret}}\qcommaq
        \capp \xgetx \parens{\trcd \tgpoint \tint \to \tvd}
      }\\[3ex]
\csolve 
  &\C\Where{
        \ldots \commaq
        \cexistsi {\tvcgetx, \inst} \xgetx \cpinst \inst \tvgetx \tvcgetx \cand \cunif {\tvcgetx} {\trcd \tgpoint \tint \to \tvd}
       }
    & \Rule{S-Let-AppR}\\[3ex]
\csolve
  &\C\Where{
      \ldots \commaq
      \cexistsi {\tvcgetx, \tvpc, \tvcgetxret, \inst} \xgetx 
        {\bigwedge \PARENS 
          { \cunif \tvcgetx {\tvpc \to \tvcgetxret} \cand \cpinst \inst \tvp \tvpc \cand \cpinst \inst \tvc \tvcgetxret \\ 
          \cunif {\tvcgetx} {\trcd \tgpoint \tint \to \tvd} }}
  }
  & \Rule{S-Inst-Copy}\\[3ex]
\csolve^+
  &\C\Where{
      \ldots \commaq
      \cexistsi {\tvcgetx, \tvpc, \tvcgetxret, \inst} \xgetx 
        {\bigwedge \PARENS 
          { \cpinst \inst \tvp \tvpc \cand \cpinst \inst \tvc \tvcgetxret \cand  \cunif {\tvcgetx} {\tvpc \to \tvcgetxret}\\ 
            \cunif \tvcgetxret \tvd \cand \cunif \tvpc {\trcd \tgpoint \tint}  }}
  }
  & \\[3ex]
\csolve
  &\C\Where{
    \cmatched \tvp {\any \eta {\trcd \tgpoint \eta}} {\cbranch {(\cpatrcd \ct)}
         {\labfrom {\ttlab \x} \ct \leq \tvp \to \tvgetxret}}\qcommaq
      \ldots
  }
  & \Rule{S-Match-Ctx}\\[2ex]
\csolve^+
  &\C\Where{
    {\labfrom {\ttlab \x} \tgpoint \leq \tvp \to \tvgetxret} \qcommaq
    \ldots
  }
  & \\[2ex]
\csolve^+
  &\C\Where{
    \cunif \tvp {\trcd \tgpoint \tvgetxret} \qcommaq
    \ldots
  }\\[2ex]
\csolve
    &\C\Where{
        \ldots \commaq
        \cexistsi {\tvcgetx, \tvpc, \tvcgetxret, \tvcgetxret', \inst} \xgetx 
          {\bigwedge \PARENS 
            {\cpinst \inst \tvgetxret \tvcgetxret \cand \cpinst \inst \tvgetxret {\tvcgetxret'} \cand \cunif {\tvcgetx} {\tvpc \to \tvcgetxret}\\ 
              \cunif \tvcgetxret \tvd \cand \cunif \tvpc {\trcd \tgpoint \tint} \\ 
            \cunif {\tvpc} {\trcd \tgpoint {\tvcgetxret'}}}}
    }
    & \Rule{S-Inst-Copy}
\end{tabular*}
\end{mathpar}
\begin{mathpar}
\begin{tabular*}{\linewidth}{;;L;;L~R}
\csolve 
    &\C\Where{
        \ldots \commaq
        \cexistsi {\tvcgetx, \tvpc, \tvcgetxret, \tvcgetxret', \inst} \xgetx 
          {\bigwedge \PARENS 
            {\cpinst \inst \tvgetxret \tvcgetxret \cand  \cunif {\tvcgetx} {\tvpc \to \tvcgetxret}\\ 
              \cunif \tvcgetxret \tvd \cand \cunif \tvpc {\trcd \tgpoint \tint} \\ 
            \cunif {\tvpc} {\trcd \tgpoint {\tvcgetxret'}} \cand \cunif {\tvcgetxret} {\tvcgetxret'}}}
    }
    & \Rule{S-Inst-Unif} \\[5ex]
\csolve^+
    &\C\Where{
        \ldots \commaq
        \cexistsi {\tvcgetx, \tvpc, \tvcgetxret, \tvcgetxret', \inst} \xgetx 
          {\bigwedge \PARENS 
            {\cpinst \inst \tvgetxret \tvcgetxret  \\ 
            \cunif {\tvcgetx} {\tvpc \to \tvcgetxret} \cand \cunif \tvpc {\trcd \tgpoint \tvcgetxret} \\ 
            \cunif {\tvcgetxret} {\cunif \tvd {\cunif {\tvcgetxret'} \tint}}}}
    }
    & \\[5ex]
\csolve
    &\C\Where{
        \ldots \commaq
        \cexistsi {\tvcgetx, \tvpc, \tvcgetxret, \tvcgetxret', \inst} \xgetx 
          {\bigwedge \PARENS 
            {\cunif {\tvcgetx} {\tvpc \to \tvcgetxret} \cand \cunif \tvpc {\trcd \tgpoint \tvcgetxret} \\ 
            \cunif {\tvcgetxret} {\cunif \tvd {\cunif {\tvcgetxret'} \tint}}}}
    }
    & \Rule{S-Inst-Poly}\\[3ex]
\end{tabular*}
\end{mathpar}
\end{local}

We begin by applying \Rule{S-Let-AppR} to the application of $\xgetx$,
introducing a fresh instantiation $\inst$. Applying \Rule{S-Inst-Copy} then
copies the arrow shape of $\tvgetx$ and introduces two incremental
instantiation constraints on the parameter type $\tvp$ and return type $\tvgetxret$. 
Using \Rule{U-Merge} and \Rule{U-Decomp}, we combine the two equations on
$\tvcgetx$, from which the solver deduces that $\tvpc$ must be the record type
$\trcd \tgpoint \tint$. 
From this equality on $\tvpc$, it follows that $\tvpc$ has the shape $\any \eta
{\trcd \tgpoint \eta}$. Together with the instantiation $\cpinst \inst \tvp
\tvpc$, we can apply \Rule{S-Match-Ctx} using \Rule{Uni-BackProp}. The match
constraint is then discharged, yielding a label instantiation constraint
forcing $\tvp$ to be $\trcd \tgpoint \tvgetxret$.
We next apply \Rule{S-Inst-Copy} to this equality on $\tvp$, introducing a
duplicate instantiation constraint on $\tvgetxret$, which is eliminated using
\Rule{S-Inst-Unif}. 
Finally, applying \Rule{U-Merge} and \Rule{U-Decomp} simplifies the remaining
equalities on $\tvcgetxret$ and $\tvpc$, and \Rule{S-Inst-Poly} removes the
last instantiation constraint on $\tvgetxret$, since it is
polymorphic.

\end{local}
\end{example}

\subsection{Metatheory}

In this section, we establish the correctness of our solver. Correctness
follows from three standard metatheoretic properties: \emph{progress},
\emph{preservation}, and \emph{termination}.  Together, they ensure that every
satisfiable (term-variable-closed) constraint eventually reduces to an
equivalent solved form.


\begin{definition}
A constraint $\c$ is term-variable-closed if all its term variables $\x$ are
bound \ie $\fvs \c \subseteq \TyVars$.
\end{definition}

\begin{lemma}[Scope preservation]
If $\ca \csolve \cb$, then $\fvs \ca \supseteq \fvs \cb$.
\end{lemma}

\begin{theorem}[Closed Progress]
  If a term-variable-closed constraint $\c$ cannot take a step $\c \csolve \cp$,
  then either:
  \begin{enumerate}
    \item $\c$ is solved.
    \item $\c$ is $\cfalse$.
    \item for every match constraint ${\cmatch \tv \cbrs}$ in $\c$, \ie
      there exists $\hat\C$ in normal form (\cref{def/normal-forms}) such
      that  $\c = \hat\C\where{\cmatch \tv \cbrs}$, then 
      $\Cshape {\hat\C} \tv \sh$ does not hold for any $\sh$.
  \end{enumerate}
\end{theorem}

\vskip -\lastskip
\begin{restatable}[Termination]{theorem}{terminationBIS}
  \label{thm:termination}
  The constraint solver terminates on all inputs.
\end{restatable}
\vskip -\lastskip
\begin{restatable}[Preservation]{theorem}{preservationBIS}
  \label{thm:preservation}
  If $\ca \csolve \cb$, then $\ca \cequiv \cb$.
\end{restatable}

\begin{corollary}[Correctness]
  \label{corollary:correctness}
  For the term-variable-closed constraint $\c$, $\c$ is satisfiable if and
  only if\/ $\c \csolve^* \hat\up$ and $\hat\up$ is a solved form equivalent
  to $\c$.
\end{corollary}

\section{Implementation}
\label{sec:implementation}

We have a working prototype implementing the \OML language with suspended
match constraints and partial type schemes, in which we have reproduced the
various type-system features and examples presented in this work. The
implementation  closely follows the constraint-based presentation described 
in the previous section.
It is public and open-source, available at
\url{https://github.com/johnyob/omniml}. Its implementation is inspired by
previous work such as \Inferno~\citep {Pottier/inferno@icfp2014,
Pottier/inferno@opam}. It uses state-of-the-art ML type inference
implementation techniques for efficiency, such as a Tarjan's union-find data
structure for unification \citep*{journals/jacm/Tarjan75} and \emph{ranks} (or
\emph{levels}) for efficient generalization \citep*{Remy/mleth}. Let us discuss
a few salient points.

\paragraph{Unification and scheduling}

Each unsolved unification variable maintains a \emph{wait list} of suspended
constraints that are blocked until the variable is unified with a concrete
type. When such a unification occurs, the wait list is flushed: the suspended
constraints are scheduled on the global constraint scheduler, which is
responsible for eventually solving them.

\paragraph{From a stack to a tree}



Some efficient implementations of \ML type inference, starting with
\citet {Remy/mleth}, represent the solver state as a linear
\emph{stack} of inference regions, from the outermost variable scope
to the current region.
Each let-binding and, more generally, each generalization location,
introduces a new scope, hence a new region.
Since the solver visits the source terms, hence introduces local regions in
a depth-first manner, these variable scopes may be represented by an integer
\emph{rank} or \emph{level}\footnote{Ranks or levels in fact correspond to
  De Bruijn levels.} attached to each inference variable or type node.
Unification maintains these levels to their minimum value when merging
types, which amounts to compute the least common ancestor of the region
types should belong to.
The use of integer levels to represent regions, which is quite efficient and
simple to implement, does not suffice for partial generalization.
If generalization at some region contain a variable appearing in a suspended
match constraint, the region must be kept alive while we continue inference
in other regions. Hence, later parts of the constraint may introduce a new
let-region at the same level that is unrelated to the suspended one---neither its
ancestor nor its descendant---breaking the linear assumption that allowed
the representation of variable scopes by integer levels.

We must instead use a \emph{tree} of nested let-regions to represent scopes,
similar to the binding tree of $G$-nodes in
\MLF~\cite{Remy-Yakobowski/mlf-type-inference@icfp08}. Under this scheme,
levels no longer uniquely determine a variable's region.  Instead, we interpret
a level relative to a path in the region tree from the root. When two variables
are unified, they must always lie on some shared path---by scoping
invariants---so computing their minimum level (along this path) still suffices
to determine the least common ancestor: we keep the efficient integer
comparisons.

\paragraph{Partial generalization}

Generalization is the process of determining which type variables are
polymorphic and which are monomorphic (\ie implementing
\Rule{S-Exists-Lower}).  \emph{Partial generalization} arises when a region
cannot be fully generalized due to suspended constraints that may 
update its variables. To manage this, we classify type variables into four
categories:
\begin{itemize}
  \let \Item \item
  \renewcommand \item [1][]{\Item[(\textbf{#1})]}
  \item[I] Variables are yet to be generalized. \\
    \emph{Introduced by instantiations or source types in constraints.}

  \item[G] Variables that are generalized. \\
    \emph{Not accessible from any instance type. Definitely polymorphic.}

  \item[PG] Variables that are partially generalizable. \\
    \emph{Generalizable variables mentioned by suspended match constraint or partial
    instantiations. Maybe polymorphic, maybe monomorphic.}

  \item[PI] Variables that were previously partially generalized
    but have since been updated.  \\
    \emph{Awaiting re-generalization. Introduced by the unification of partially
    generalized types.}
\end{itemize}

During solving, we track a set of \emph{guards} for each type variable. Each
guard indicates that the variable is captured by a suspended constraint, and may
therefore be updated when that constraint is discharged. At generalization
time, these guards conservatively approximate whether a variable may be
updated in the future. Variables with at least one guard are partially
generalized (\textbf{PG}); those without guards are fully generalized
(\textbf{G}).

When an instance is taken from a partially generalized variable $\tv$
(\textbf{PG}), we retain a forward reference from $\tv$ to the instance $\tvb$.
This enables $\tv$ to notify $\tvb$ if it is later updated, at which
point~$\tv$ becomes a partial instance (\textbf{PI}) and propagates its
updates to $\tvb$ (and to its other instances).
This mirrors, in reverse, the way our formalized solver uses incremental
instantiation constraints to track copies. In addition, the instance $\tvb$
remains guarded by $\tv$ until the latter is either lowered or fully
generalized.

Once a suspended match constraint is solved, it removes the guards it
introduced.  This may enable previously partially generalized variables $\tv$
(\textbf{PG}) to become fully generalized (\textbf{G}). Conversely, if a
partially generalized variable $\tv$ is lowered (\eg by \Rule{S-Exists-Lower}),
becoming an instance (\textbf{PI} or \textbf{I}), it is unified with all its
instances $\tvbs$ obtained via its forward references.

\paragraph{Lazy generalization}

Repeatedly generalizing a region after every update is expensive.  Instead we
generalize on demand. We mark regions as ``stale'' when they may require
re-generalization. When an instance is taken, we re-generalize the stale
descendants of the region in the region tree.

\parcomment{Lazy generalization is an optimization for a common case, it is
not optimal} 

Although this technique prevents premature re-generalization before
instantiation, it only optimizes for a common case. In some situations,
re-generalization before instantiation can still be more costly than
instantiating first, since the instantiation may discharge suspended
constraints that further refine the region's partial type scheme.

\section{Discussion}
\label{sec/discussion}

We now discuss practical considerations and an alternative
semi-unification-based implementation of incremental instantiation.

\paragraph{Error reporting}

Recent work on error reporting for constraint-based inference
\citep*{journals/pacmpl/BhanukaPBB23} shows that the structure of constraints
can be leveraged to produce precise, localized diagnostics. 
Our setting fits well with this line of work: suspended constraints, in
particular, do not fundamentally complicate error reporting. 
In fact, the causal ordering among suspended constraints gives a natural
account of information flow, which can aid in tracing the source of type
errors. 
At present, our prototype supports only basic error reporting, attaching source
locations to constraints. That said, we see no obstacle to adopting the
approach of \citet*{journals/pacmpl/BhanukaPBB23}.

\paragraph{Integration into \OCaml} 

Integrating suspended constraints into \OCaml's type checker raises two
questions: their compatibility with \OCaml's existing type system, and the
feasibility of implementing them within \OCaml's current typechecker. 

From a type-theoretic perspective, we are confident that suspended constraints
and incremental instantiation are compatible with the various features of the
\OCaml type system. Moreover, several other features of \OCaml could benefit
from suspended constraints. In some cases, such as GADTs, this perspective
could open up a new research topic of its own. Nevertheless, reformulating
every fragile feature to use omnidirectionality is not required for initial
adoption in \OCaml: existing features could remain unchanged, with suspended
constraints introduced selectively where it provides clear benefits. 

For the implementation perspective, the situation is more delicate. The current
\OCaml has grown organically over more than two decades. The underlying
inference algorithm uses a single-pass architecture implemented with global
mutable state and few abstraction boundaries. As a result, the typechecker is
large, intricate and increasingly difficult to maintain. Our current expert
opinion is that extending the existing implementation to support suspended
constraints is feasible.%
\footnote{%
  Compared to the implementation described in \cref{sec:implementation}, \OCaml
  already maintains a tree of inference regions. One could introduce the
  additional data structures required for \emph{partial and lazy
  generalization}, along with the scheduling mechanisms needed to implement
  suspended constraints. 
}
However, such an extension would combine the existing mutable state and
in-place updates with a more complex control flow, raising concerns about code
complexity and long-term maintainability. These concerns could ultimately argue
against integrating suspended constraints into the current typechecker. 

At the same time, there are ongoing efforts towards a clean reimplementation of
the \OCaml typechecker. Such designs would provide a more natural setting for
suspended constraints and could significantly simplify their integration.

\paragraph{Incremental instantiation via semi-unification}

\parcomment{What is semi-unification? And it's relation to \ML}

Semi-unification \citep*{Henglein/phd,Henglein/tiasu@toplas93}, introduced in
the late 1980's, is a generalization of unification that solves a collection
 of \emph{instantiation inequalities} $\bigwedge_{i \in I} \ti \leq \tip$,
meaning that there exist a unifier $\theta$ and a collection of
substitutions $\bar\rho$ such that $\rho_i(\theta(\ti)) = \theta(\tip)$.
Unfortunately, semi-unification was soon proved undecidable~\citep*{%
Kfoury-Tiuryn-Urzyczyn/ic1993} and almost abandoned for type inference
purposes---even though \citeauthor{Henglein/phd} observed that his
semi-algorithm appeared to terminate on \emph{most} \ML type inference
problems involving polymorphic recursion.%
\footnote{
  Although the undecidability of semi-unification already implied the existence
  of such examples, the specific class of examples of non-terminating cases was
  only characterized later by
  \citet*{Figueiredo-Camarao/semiunif-ptm@unpub2004}. 
  These examples use extremely complex patterns, suggesting that they are
  unlikely to 
  appear in practice, unless crafted deliberately.
}
Moreover, in the absence of polymorphic recursion, semi-unification constraints
arising from \ML type inference are acyclic and therefore guaranteed to
terminate.

\parcomment{Semi-unification is omnidirectionally-ready}

For our purposes, semi-unification offers an alternative foundation for
incremental instantiation. Because semi-unification constraints can be
solved \emph{incrementally}, in any order, making them a natural fit for
scaling omnidirectional type inference to \ML-style polymorphism. Although it
offered little advantage for traditional \ML inference,
\citeauthor{Henglein/phd}'s algorithm was \emph{omnidirectionally-ready},
merely waiting for fragile constructs to give it new life!

\parcomment{What we did}

Motivated by this observation, we explored a second prototype that implements
let-generalization and incremental instantiation via semi-unification. To
enable a more direct encoding of type inference problems and improve solving
efficiency, we enriched semi-unification constraints with a notion of
\emph{scope}, forming the same tree-like structure as regions in our
\emph{generalization tree} (\cref{sec:implementation}). This extends the
formalization of semi-unification with \emph{unknowns} proposed by~\citet
{Lushman-Cormack/subti@tr2008}, and allows for more efficient handling of
\emph{monomorphization}---that is, instantiations that collapse into
unification constraints (akin to \Rule{S-Inst-Mono} in our constraint solver).

\parcomment{The solutions coincide}

Although this prototype began from a different perspective, the implementation
problems and their solutions turned out to be strikingly similar to those in
our first prototype (described in \cref{sec:implementation}). This close
correspondence naturally raises the question of whether partial type schemes
could be enriched to cope with polymorphic recursion.

\section{Related work}
\label{sec:related-work}

\paragraph{Suspended constraints}

Suspending constraints that cannot be solved yet is not a novel idea: it is a
standard approach for implementing unification in dependently typed systems. This goes
back to Huet's algorithm for higher-order unification~\citep*{huet-unif} and
pattern unification~\citep*{Miller/pattern-unif@iclp91} where flexible-flexible
pairs are delayed until at least one side becomes rigid.
Our contribution lies in combining constraint suspension with \ML-style
implicit polymorphism---largely absent from dependently typed systems---and in
formulating a declarative constraint semantics.

Conditional constraints \citep*{journals/njc/Pottier00} also delay resolution,
waiting until the top-level constructor of a type is known. They provide an
\code{if-then-else}-like primitive, but differ crucially from our suspended
constraints: in Pottier's system, an unresolved conditional constraint is
considered satisfiable, whereas in ours, an unresolved suspended constraint is
not. This difference forces our semantics to track what is \emph{known} in a
context. Consequently, unresolved conditional constraints may enter a
generalized type scheme as a form of qualified types, while our suspended
constraints cannot. These semantic differences lead the two approaches to
address very different user-facing type system features.

\OutsideIn~\citep*{conf/icfp/SchrijversJSV09} is a type system for GADTs that
introduces \emph{delayed implications} of the form $\where{\tvs}(\all \tvbs \ca
\Rightarrow \cb)$. Constraint solving for delayed implications proceeds in two
steps; solving simple constraints first and then solving delayed implications.
The deferral ensures that inference for GADT match branches occurs when more is
known about the scrutinee and expected return type from the context.
To ensure principality, \OutsideIn enforces an algorithmic restriction: the
variables $\tvs$ must already be instantiated to concrete type constructors
before they may be unified by the implication's conclusion $\cb$. This ensures
information only flows from the outside into the implication's conclusion.
Notably, they do give a declarative specification for this restriction, using
an elegant but mysterious quantification on all possible ways to type the
context outside the GADT clauses. Using our new perspective on \emph{known}
type information, we can say that their semantics enforces that only
\emph{known} information from outside GADT clauses can be used inside.
Later work on \OutsideIn, namely
\OutsideInX~\citep*{\BBoutsidein},
argues that delayed implication constraints make local let-generalization all
but unmanageable, both in theory and implementation. Their proposed fix is to
abandon local let-generalization altogether. By contrast, our work shows that
the difficult interactions between let-generalization and suspended constraints
can be resolved. Furthermore, \OutsideInX forgoes a declarative specification
complete with respect to its inference algorithm, on the grounds that such a
specification would be ``as complicated and hard to understand as the
[inference] algorithm''. We believe that our \emph{omnidirectional recipe}
could provide a declarative specification: one capable of being principal and
complete for GADTs, and we would be interested in studying this application.

\paragraph{Directionality}
\label{sec/related-work/directionality}

The fixed directionality of bidirectional type systems is a known limitation.
Prior work has sought to mitigate this within the bidirectional setting by
altering the order in which information flows. For instance,
\citet*{conf/esop/XieO18}'s \emph{let arguments go first} typechecks
applications $\eapp \ea \eb$ by checking the argument $\eb$ before the function
$\ea$, allowing information to flow from argument to function (\cf
\Rule{Chk-App}).

\parcomment{Contextual typing}

More recent work on \emph{contextual typing} \citep{journals/pacmpl/XueO24, journals/pacmpl/XueCJO26}
goes further by deferring the commitment between \Rule{Syn-App}
and \Rule{Chk-App} by parameterizing typing $\G \th_m \e : \t$ 
with contextual information $m$.
This enables typing multiple arguments from right to left when sufficient
contextual information is available, and thus successfully typechecks
programs such as~\code{ex_6}
from~\cref{sec/overview/directional}. Nevertheless, it still enforces a
fixed order of propagation and, consequently, some well-typed programs are
rejected as ill-typed (\eg \code{ex_6_2, ex_6_3}).\footnote{\let \code
\Code%
Readers wishing to verify
that \code{ex_6_2} and \code{ex_6_3} are rejected under contextual typing must
inline \code{app} and \code{rev_app}. }

%
%
%
%
%

\parcomment{Comparisons}

Omnidirectionality removes this restriction, allowing information to flow
without committing to a fixed propagation order. The contextual information $m$
(or \emph{mask}) of contextual typing can be viewed as a \emph{lossy}
abstraction of the surrounding context $\E$ in our system: omnidirectionality
preserves $\E$ faithfully and exploits the information that the
mask $m$ forgets, propagating more type information and accepting more
well-typed programs. 

However, contextual typing can continue to guide inference even when no
contextual information is available (\ie $m = \Blacksquare$). In particular,
\emph{contextual subtyping} may default to reflexivity:
\begin{mathpar}
  \infer[DS-Refl]
    { }
    {\G \th_\Blacksquare A \leq A}
\end{mathpar}
Within inference, this amounts to resolving subtyping with unification. This
mechanism is analogous to the idea of \emph{default
rules}~(\cref{sec/default-rules/poly}) in our setting.

Because contextual typing still fixes a propagation order, this default can be forced too 
early. The following program, in the contextual System-$\F$ calculus of \citet{journals/pacmpl/XueCJO26} (\LCTI),
illustrates the issue:
\begin{program}[input]
let choose : 'a. 'a -> 'a -> 'a = fun x y -> x
let ex_1_3 = (choose id : (($$'a. 'a -> 'a) -> 'a. 'a -> 'a) -> _) °\Ocamlcomment{\ocamlFlag {\LCTI}1}°
\end{program}
This program is rejected because \ocaml{id} must be inferred ($m =
\Blacksquare$) before the outer annotation can guide the instantiation of
\ocaml{choose}. With no contextual information available, this instantiation of
\ocaml{id} must proceed via \Rule{DS-Refl}. As a result, \ocaml{choose} is
instantiated with \ocaml{'a. 'a -> 'a}, yielding a type incomparable with the
annotation. 

Finally, comparisons between omnidirectionality and contextual typing are not
quite apples-to-apples. Contextual typing targets intersection types and
first-class polymorphism, whereas omnidirectionality has focused on static
overloading and semi-explicit first-class polymorphism. Since important
extensions of omnidirectionality, including default rules and first-class
polymorphism, remain future work, direct expressiveness comparisons are
premature. 

\paragraph{Overloading}

Qualified types~\citep*{jones-qualified-types}, best known for their use in
\Haskell's type-classes, are related to our suspended match constraints: both
represent constraints on types or type variables that are delayed. At
generalization time, constraints on generalizable variables are retained in the
type scheme, yielding a \emph{constrained type scheme} $\tfor \tvs {\c
\Rightarrow \t}$. This is much simpler to implement than our partial type
schemes, but it provides a different behavior: each instance may resolve $\c$
differently (as the constraint is copied on instantiation).
Qualified types are excellent choice when this is the desired behavior,
typically for \emph{dynamic overloading} \citep{conf/popl/WadlerB89}. But they
are insufficient when we require a unique resolution of the constraint across
all instances---as in \emph{static overloading}.
Static overloading can be emulated using qualified types by applying
implementation-defined tweaks to specific built-in type-classes, enforcing that
resolution is global and failing when it is under-determined. However, such
mechanisms lack a clear declarative semantics: to our knowledge, capturing
failure due to missing information requires our unicity conditions.

\citet{Leijen-Ye/prefix@pldi2025} recently proposed a bidirectional account of
generalized static overloading within \ML. However, their approach is limited by
its reliance on fixed directionality (\cref{sec/overview/directional}).
Variational typechecking \citep{Chen-Erwig-Walkingshaw/variational@toplas2014}
was originally developed for reasoning about well-typed CPP \code{#ifdef}-style
macros, introducing \emph{choices}  $a \angles {\ea, \eb}$, where $a$ is a
\emph{dimension} with \emph{alternatives} $\ea$ and $\eb$. Once dimensions are
fixed, we are able to project a well-typed non-variational program.
\citet{benevs2025simple} apply this machinery to recast static overloading as
variational typing, with a resolution algorithm that uniquely selects the
dimensions.
However, their system removes \emph{local let-generalization} and requires an
exponential-time resolution procedure---an unavoidable consequence of the
NP-hardness of \emph{general} static overloading \citep*
{Chargueraud-Bodin-Dunfield-Riboulet/jfla2025}.

Partial type schemes provide an alternative that preserves \ML's local
let-generalization while suspended constraints offer a tractable account of
static overloading. By enforcing resolution using \emph{known} type information
(captured by our novel unicity condition) rather than \emph{guessed}
information, our approach remains tractable. Our experience suggests that this
is a ``goldilocks'' solution: expressive enough for most applications, yet
tractable, and (crucially) compatible with \ML's \emph{local
let-generalization}.
However, we have not yet extended our approach to capture \emph{generalized
static overloading}; developing a full omnidirectional account of that setting
is left to future work.

\paragraph{First-class polymorphism}


\Polytypes provide a form of \emph{semi-explicit} first-class polymorphism and
are of particular interest due to their role in \OCaml's inference of
polymorphic methods and polymorphic function parameters. 
Fully implicit first-class polymorphism is more expressive than \polytypes and
other forms of boxed polymorphism. In the latter, the programmer must
explicitly manage the boxing and unboxing of terms (and types) for
generalization and instantiation. However, making these operations implicits
shifts the burden to inference, which must determine when to introduce and
eliminate these polytypes. It remains unclear whether such inference can be
achieved using omnidirectionality \emph{without} default rules
(\cref{sec/default-rules/poly})---a mechanism also used in bidirectional
systems (\cref{sec/related-work/directionality}).

Many systems for implicit first-class polymorphism exist; here we highlight a
few in the context of \ML.
\MLF~\citep*{LeBotlan-Remy/recasting-mlf} is an extension of \ML that supports
first-class polymorphism that goes beyond the power of System $\mathsf{F}$,
while retaining type inference with principal types. It is a generalization of
\polytypes, relying on \geninst-directionality, but it remains unclear how to
effectively scale \MLF to the rest of \OCaml's features. Nevertheless, it would
be interesting to explore whether the directional sensitivity of \MLF could be
removed by omnidirectionality.

\FreezeML~\citep*{10.1145/3385412.3386003} is an impredicative type inference
system in which polymorphic variables can be \emph{frozen}, written $\lceil x
\rceil$, and only allowing instantiation on ordinary (unfrozen) variables $x$.
Unlike \polytypes, \FreezeML permits higher-rank types directly in the syntax
of types, though these can be encoded back into a minor extension of
\polytypes.\footnote{%
Each let-binding $\elet \x \ea \eb$ introduces a boxed polytype $\elet \x
{\epoly \ea} {\eb}$, with each variable occurrence of $\x$ implicitly
unboxing it $\einst \x$. Freezing disables this implicit unboxing.
By default, the implicit boxing operator $\epoly \ea$ infers a polytype 
$\tpoly \ts$, where $\ts$ is the most general type of $\ea$ when it 
is not already known from the surrounding context.}  
The essential difference is that generalization of higher-rank types is
implicit, inferring the most-general type, whereas \polytypes require an
explicit type annotation $\expoly \e \tvs \ts$. 

\QuickLook \citep{journals/pacmpl/SerranoHJV20}, \Haskell's latest approach at
first-class polymorphism, uses bidirectional propagation to take a ``quick
look'' at the spine of an application to guide the instantiation of higher-rank
functions. This permits \emph{argument reordering}, allowing later arguments to
influence the type of earlier ones within the same application spine. While
this mitigates the order-dependence of fixed directionality
(\cref{sec/overview/directional}), it still enforces a global direction of type
propagation between surrounding constructs (\eg let-bindings). In this sense,
\QuickLook is \emph{locally omnidirectional} but \emph{globally directional}:
they reduce order sensitivity within applications, but do not provide a
declarative account of type flow between arbitrary constructs, as required for
omnidirectionality. More recently, \Frost \citep{journals/corr/abs-2607-16061}
follows a similar line but extends it to handle $\eta$-expansion and introduces
a freezing operator, akin to that of \FreezeML, that helps propagate type
information when the default order of inference is insufficient.

\paragraph{Type-based disambiguation in \SML}

In \SML, tuples are treated as structural records with numeric labels, and
projections such as \code{#1} or \code{#2} are type-directed. This relies on
row typing: if $\e$ has the type $\tsrecord {j = \tj; \varrho}$, where
$\varrho$ is a row describing the remaining tuple fields, then $\eproj \e j$
has type $\tj$. The same mechanism also underlies record projections.

However, \SML does not support row-polymorphic definitions: restricting to
monomorphic rows allows a simple yet efficient compilation strategy for records
and tuples. To enforce this, \SML adds a prose side-condition to the typing
rules requiring that each row variable be fully determined by the ``program
context''~\citep*[Section 4.11]{sml-definition}. Our unicity conditions
provides a precise, declarative specification for this informal restriction,
filling a gap in the original specification.

\citet*[page 36]{rossberg-hamlet} remarks that this restriction interacts
poorly with \texttt{let}-polymorphism and therefore not supported in practice:
\begin{quotation}\em
  Under item 1 the Definition states that “the program context” must
  determine the exact type of flexible records, but it does not
  specify any bounds on the size of this context. Unlimited context is
  clearly infeasible since it is incompatible with let polymorphism:
  at the point of generalization the structure of a type must be
  determined precisely enough to know what we have to quantify
  over. We thus restrict the context for resolving flexible records to
  the innermost surrounding value declaration, as most other SML
  systems seem to do as well. This is in par with our treatment of
  overloading (see 5.8).
\end{quotation}
We have solved this difficult interaction between disambiguation and
let-polymorphism. In particular, we implemented tuples in our
prototype and formalize their meta-theory in Appendix
\cref{app:full-reference}; the system behaves as expected.

We also remark that \PolySML~\citep*{sml-polyml} goes further than other \SML
implementations on this front, supporting examples that even rely on
backpropagation:
\begin{program}[input,deletekeywords={true,false}]
let fun fst r = #1 r in (fst (1, 2), fst (true, false)) end; °\polysmlflags 00°
\end{program}

\paragraph{Record field overloading in \GHC}

\Haskell 98 derives selector functions for each record field, an approach that
precludes declaring two record types with the same field names within a single
module. \GHC 6.8.1 (2007) introduced the \texttt{DisambiguateRecordField}
extension for type-directed disambiguation of record labels,%
\footnote{
  \url{https://ghc.gitlab.haskell.org/ghc/doc/users_guide/exts/disambiguate_record_fields.html}
}
and \GHC 8.0.1 (2016) added \texttt{DuplicateRecordFields}, allowing distinct
record types with overlapping field names in the same module.

However, the \GHC developers found that type-based disambiguation using
bidirectional typechecking is not sufficiently predictable in practice
for users~\citep*{RFC-overloaded-record-fields}, and that the reliance on
projection functions, rather than a dedicated projection syntax, further
complicated the implementation by making all function calls potentially
ambiguous. As a result, \GHC is gradually moving away from type-based
disambiguation towards an approach based on qualified types constraints of the
form \code{HasField "x" a b} (\texttt{OverloadedRecordDot}, 2017).%
\footnote{
  \url{https://ghc.gitlab.haskell.org/ghc/doc/users_guide/exts/hasfield.html}
}

For instance, the projection \code{ex_1} from \cref{sec/overview/overloading}
would be written as:
\begin{program}[input]
ex_1 :: HasField "x" a b => a -> b °\ghcflags 01°
ex_1 r = getField @"x" r
\end{program}
It is worth emphasizing that \OML rejects this program. Since \OML implements
\emph{static} overloading, it requires a unique resolution for each overloaded
occurrence; as a result, programs such as \code{ex_1}, whose projection is
ambiguous, are ill-typed. In contrast, a type-class-based solution---\ie one
using qualified types---implements \emph{dynamic} rather than \emph{static}
overloading, allowing each instantiation to resolve the field independently.
This distinction has practical consequences: dynamic overloading is able to
naturally support field-polymorphic programs such as \code{ex_1}, but it does
not guarantee efficient code generation (though \GHC is often able to optimize
type-class programs effectively).


Beyond efficiency, the \texttt{HasField} approach is also limited in
expressiveness. It works well for records with `simple' field types (\ie field
types with \emph{uniform} typing rules for projections), but it does not scale
to field types where the typing rule itself depends on type-directed
disambiguation---such as fields containing polymorphic types or GADT
existentials. In contrast, suspended match constraints naturally accommodate
such non-uniformity in typing rules.

\section{Conclusions}
\label{sec:future-work}

We presented a constraint-based framework for omnidirectional type inference,
scaled to \ML with \emph{local let-generalization}. Central to our approach is
a new declarative account of when a type is \emph{known} from the context,
rather than \emph{guessed}. Our constraint solver is omnidirectional:
constraints may be solved in \emph{any way}, enabled by partial type schemes.
Through two instantiations of our \emph{omnidirectional recipe}, we obtained a
sound, complete, and \emph{principal} type inference algorithm---in short,
principality held \emph{anyway}, precisely because of omnidirectionality.

\paragraph{Future work}

We aim to extend our framework to support more advanced features. One direction
is generalized \emph{static overloading}; another is \emph{implicit first-class 
polymorphism}. We also plan to investigate \emph{default rules}---a mechanism
where ambiguity is resolved by falling back on a default, non-principal choice
\eg \OCaml selects the most recent matching record type in scope for ambiguous
field names.

\section*{Acknowledgments}

\paragraph{Making of}

The unsatisfactory situation of type-directed disambiguation in \OCaml has
been on our mind for years; a first concrete step to which this work
can be traced is a specific
example\footnote{\url{https://github.com/ocaml/ocaml/issues/7388}} of
inconvenient order of propagation (between pattern and expression in
a let-binding) presented by Andreas Rossberg as part of general
feedback on the use of \OCaml in the \texttt{wasm} reference
interpreter \citep*{rossberg-wasm}. Gabriel Scherer proposed delaying
the resolution of ambiguous constraints and advised Olivier Martinot
as an intern on this topic in Summer
2020~\citep*{frozen-constraints-2021}. Together, they
identified the implementation and specification difficulties; their
implementation would only handle simple cases where no incremental
generalization was necessary (in particular, it did not support
backpropagation), and they did not have a precise declarative
specification. Gabriel Scherer kept working on this infrequently until
2024, identified the need for incremental generalization and
instantiation, but its implementation and declarative specifications
remained incomplete.

Alistair O'Brien started working on suspended constraints in Fall
2024, coming from work on constraint-based presentations of GADT
checking---suspended constraints could delay the checking of GADT
patterns until the equalities to be introduced are known to be between
rigid types. He proposed an alternative declarative semantics,
and implemented a working prototype of incremental instantiation, 
described in the present paper.

Didier R\'emy started working on suspended constraints in Winter
2024--2025, coming from questions of principal type inference in
the presence of modular implicits, in collaboration with Samuel
Vivien. Didier R\'emy implemented a second prototype for incremental 
instantiation based on semi-unification, which is fairly different from
the prototype described in the present paper, and brought the idea of
handling \polytypes as suspended constraints.

All co-authors contributed to the writing and metatheory, with
Alistair O'Brien as the main contributor at every step.

\paragraph{Thanks}

We thank the anonymous reviewers for their thoughtful feedback and questions,
and Simon Peyton-Jones, whose insightful remarks, we hope, helped us make this
work more approachable. We are also grateful to Adam Gundry for clarifying
aspects of record-field disambiguation in \GHC, and to Jeremy Yallop for his
sage guidance throughout this project. Finally, we thank Brendan Coll, Daniel
Gooding, Micheal Lee, and Fran{\latexc{c}}ois Pottier for their helpful
comments on earlier drafts. 
%
This work was partially funded by the \OCaml Software Foundation.

\begin{local}
\let \t \latext \let \c \latexc

\bibliography{suspended}
\end{local}

\newpage
\FloatBarrier 
\appendix

\section*{\Large Organization of appendices}
\label{app:outline}

\paragraph{Reference appendix}
\cref{app:full-reference} gives a full reference for all
  definitions, grammars and figures in the paper, including all cases
  (even those omitted from the main paper for reasons of space).

\paragraph{Proof appendices} These appendices contain proofs for the
formal claims in the article. They are typically written tersely.
\begin{itemize}
\item \cref{app:proofs-constraints} proves properties of the
  constraint language and its semantics. The main result is
  canonicalization, which morally establishes that uses of the
  contextual rule \Rule{Match-Ctx} can be ``permuted down'' in the
  proof until they are all at the bottom of the derivation, followed
  by a proof on a simple constraint.
\item \cref{app:proofs-solving} proves the correctness of the
  constraint solver with respect to the semantics.
\item \cref{app/oml/proofs} proves the properties about the \OML
  type system, in particular the correctness of constraint generation.
\end{itemize}

\clearpage
\section{Full technical reference}
\label{appendix:figures}
\label{app:full-reference}

This section repeats all the technical definitions mentioned in the paper,
including the cases, rules, and definitions that were omitted from the main
paper to save space. It can serve as a useful cheatsheet to understand a
definition in full, or when studying the meta-theory of the system.

\begin{mathpar}
\begin{bnfgrammar}[\noleftfill\negrightfill]
\entryset[Type variables]{\tva, \tvb, \tvc}{\TyVars}{}
\\
\entry[Types]{\t}{
\tv \and
\tunit \and
\ta \to \tb \and
\Pi \iton \ti \and
\trcd \T \tys \and
\tpoly \ts
}\\
\entry[Type schemes]{\ts}{
\t \and
\all \tv \ts
}\\[1ex]
\entry[Ground types]{\gt}{}{}\\
\entry[Ground region]{\gr}{\greg \gt \semenv}\\
\entrysubseteq[Sets of ground types]{\gabs}{\Ground}{}{}\\
\entrysubseteq[Sets of ground regions]{\gabsr}{\GroundRegion}{}\\
\entry[Constraints]{\c}{
\ctrue
\and  \cfalse
\and  \ca \cand \cb
\and  \cexists \tv \c
\and 	\cfor \tv \c
\and  \cunif \ta \tb
\nextline
\and  \clet \x \tv \ca \cb
\and  \capp \x \t
\nextline
\and  \cmatch \t \cbrs
\nextline
\and \ueq
\and \cletr \x \tv \tvs \ca \cb
\and \cexistsi \inst \x \c
\and \cpinst \inst \tv \t
\nextline
\and \labfrom \elab \ct \leq \ta \to \tb
\and \dom \ct = \elabs
\and \cscm \leq \t \mid \x \leq \cscm
}\\[1ex]
\entry[Branches]{\cbr}{\cbranch \cpat \c} \\
\entry[Shape patterns]{\cpat}{
\cpatwild \and \cpatprod \tv j \and \cpatrcd \ct \and \cpatpoly \cscm
} \\[1ex]
\entry[Semantic environment]{\semenv}{
\eset \and \semenv\where{\tv := \gt}
\and \semenv\where{\x := \gabs}
\and \semenv\where{\x := \gabsr}
\and \semenv\where{\inst := \semenvp}
}\\[1ex]
\entry[Unification problems]{\up}{
\ctrue \and \cfalse \and \upa \cand \upb \and \cexists \tv \up \and \ueq
} \\
\entry[Multi-equations]{\ueq}{
\eset \mid \cunif \t \ueq
} \\[1ex]
\entry[Constraint contexts]{\C}{
\hole
\and \C \cand \c
\and \c \cand \C
\and \cexists \tv \C
\and \cfor \tv \C
\nextline
\and \clet \x \tv \C \c
\and \clet \x \tv \c \C
\nextline
\and \cletr \x \tv \tvs \C \c
\and \cletr \x \tv \tvs \c \C
\and \cexistsi \inst \x \C
} \\
\entry[Shapes] {\Sh} {
\any \tvcs \t
}
\\
\entry[\llap{Canonical} principal shapes] {\sh} {} {}\\
\entry[Terms]{\e}{
x \and
() \and
\efun x e \and
\eapp \ea \eb \and
\elet x \ea \eb \and
\eannot \e \tvs \t \andcr
\erecord {\overline{\elab = \e} } \and
\efield e \elab \and
\exrecord \T {\overline{\elab = \e}} \and
\exfield \e \T \elab
\andcr
(\ea, \ldots, \en) \and
\efield e j \and
\exfield e n j \andcr
\epoly e \and
\expoly e \tvs \ts \and
\einst e \and
\exinst e \tvs \ts
\nextline
\and \emagic \es
}\\
\entry[Term contexts]{\E}{
  \hole
  \and \eapp \E \e
  \and \eapp \e \E
  \and \efun \x \E
  \and \elet \x \E \e
  \and \elet \x \e \E
  \and \eannot \E \tvs \t
  \andcr \erecord {\elaba = \ea\; \ldots\; \elabi = \E\; \ldots\; \elab_n = \en}
  \and \efield \E \elab
  \andcr \exrecord \T {\elaba = \ea\; \ldots\; \elabi = \E\; \ldots\; \elab_n = \en}
  \and \exfield \E \T \elab
  \andcr (\ea, \ldots, \E, \ldots, \en)
  \and \eproj \E j
  \and \exproj \E j n
  \andcr \epoly \E
  \and \expoly \E \tvs \ts
  \and \einst \E
  \and \exinst \E \tvs \ts
  \andcr
  \emagic {\ea, \ldots, \E, \ldots, \en}
}\\
\entry[Typing contexts]{\G}{
   \eset \and
   \G, x : \ts
}\\
\entry[Label environment]{\labenv}{
  \eset \and \labenv, \elab : \tfor \tvs {\trcd \T \tvs \to \t}
}
\end{bnfgrammar}
\end{mathpar}

\begin{judgboxmathpar}
  {\semenv \th \c}
  {Under the environment $\semenv$, the constraint $\c$ is satisfiable.}
  \infer[True]
    { }
    {\semenv \th \ctrue}

  \infer[Conj]
    {\semenv \th \ca \\
     \semenv \th \cb}
    {\semenv \th \ca \cand \cb}

  \infer[Exists]
    {\semenv\where{\tv \is \gt} \th \c}
    {\semenv \th \cexists \tv \c}

  \infer[Forall]
    {\forall \gt, ~ \semenv\where{\tv \is \gt} \th \c}
    {\semenv \th \tfor \tv \c}

  \infer[Unif]
    {\semenv(\ta) = \semenv(\tb)}
    {\semenv \th \cunif \ta \tb}

  \infer[Let]
    {\semenv\where{\x \is \gabs} \th \cb \\
     \gabs = \semenv(\cabs \tv \ca) \\
     \gabs \neq \eset}
    {\semenv \th \clet \x \tv \ca \cb}

  \infer[App]
    {\semenv(\t) \in \semenv(\x)}
    {\semenv \th \capp \x \t}

  \infer[Match-Ctx]
    {\Cshape \C \t \sh \\
      \semenv \th \C\where{\cmatched \t \sh \cbrs}
    }
    {\semenv \th \C\where{\cmatch \t \cbrs}}

  \infer[Multi-Unif]
    {\forall \t \in \ueq,\, \semenv(\t) = \gt}
    {\semenv \th \ueq}
\\
  \infer[LetR]
    {\semenv\where{\x \is \gabsr} \th \cb\\
     \gabsr = \semenv(\cabsr \tv \tvs \ca) \\
     \gabsr \neq \emptyset}
    {\semenv \th \cletr \x \tv \tvs \ca \cb}

  \infer[AppR]
    {\greg {\semenv(\t)} \wild \in \semenv(\x)  }
    {\semenv \th \capp \x \t}

  \infer[Exists-Inst]
    {\semenv\where{\inst \is \semenvp} \th \c\\
     \greg \wild \semenvp \in \semenv(\x)}
    {\semenv \th \cexistsi \inst \x \c}

  \infer[Incr-Inst]
    {\semenv(\inst)(\tv) = \semenv(\t) }
    {\semenv \th \cpinst \inst \tv \t}
\\
  \newcommand{\Srule}[3][]{{#2} &\eqdef& {#3} & {#1}}
  \newcommand{\Scases}[3]{%
    \left\{
      \begin{array}{l}
      #1\\[0.4ex]
      #2
      \end{array}%
    \right.
    &
    \hspace{-1ex}\begin{array}{l}
      \text{if } #3\\[0.4ex]
      \text{otherwise}
    \end{array}%
  }
  \begin{tabular}{RCLL}
    \labfrom \elab \T \leq \ta \to \tb &\eqdef&
      \Scases{\cexists \tvs \cunif \ta  {\trcd \T \tvs} \cand \cunif \tb \t}{\cfalse}{\labenv(\labfrom \elab \T) = \tfor \tvs {\trcd \T \tvs \to \t}}
    \\[1ex]
    \dom{\T} = \elabs &\eqdef& \Scases{\ctrue}{\cfalse}{\Dom {\labenv(\T)} = \elabs} \\[1ex]
    \Srule
      {(\tfor \tvs \tp) \leq \t}
      {\cexists \tvs \cunif \tp \t}
    \\[1ex]
    \Srule
      {\x \leq (\tfor \tvs \t)}
      {\cfor \tvs \capp \x \t}
    \\[1ex]
    \cmatched \t \sh {\cbranch \cpat \cs} &\eqdef&
    \Scases
      {\cexists \tvs \cunif \t \shapp \tvs \cand \theta(\ci)}
      {\cfalse}
      {\cmatches \cpati \sh \tvs \theta}
  \end{tabular}
\\
\semenv(\cabs \tv \c) \Wide\eqdef \
  \set {\gt \in \Ground : \semenv\where{\tv \is \gt} \th \c}
\\
  \semenv(\cabsr \tv \tvs \c) \uad\eqdef\uad \set{\greg \gt {\semenv\where{\tv \is \gt, \tvs \is \gts}} \in \GroundRegion :
    \semenv\where{\tv \is \gt, \tvs \is \gts} \th \c}
\\
\Cshape \C \t \sh \Wide\eqdef \
  \forall \semenv, \gt. \uad
      \semenv \th \cerase {\C\where{\cunif \t \gt}} \implies \shape \gt = \sh
\end{judgboxmathpar}

\begin{judgboxmathpar}
  {\Sh \preceq \Shp}
  {The shape $\Shp$ is an instance of $\Sh$. Alternatively, $\Shp$ is more general than $\Sh$.}
  \infer[Inst-Shape]
    { }
    {\any {\tvcs_1} \t \preceq
     \any {\tvcs_2} \t \where {\tvcs_1 \is \tys_1}}
\end{judgboxmathpar}

We write $\bot$ for the trivial shape $\any \tvc \tvc$. $\Shapes$ denotes the set of shapes and $\Shapesz$ is the set of non-trivial shapes.

\begin{definition}
A non-trivial shape $\Sh \in \Shapesz$ is the principal shape of the type
$\t$ iff:
\begin{enumerate}
  \item
    $\exists \typs,\ \t = \shapp[\Sh] \typs$
  \item
    $\forall \Shp \in \Shapesz, \forall \typs,\ \t = \shapp[\Shp] \typs
    \implies \Sh \preceq \Shp$
\end{enumerate}

A principal shape $\any \tvcs \t$ is \emph{canonical} if the sequence of its
free variables $\tvcs$ appear in the order in which the variables occur in
$\t$. $\shape \t$ is the canonical principal shape of $\t$.
\end{definition}

\begin{judgboxmathpar}
  {\cmatches \cpat \sh \tvcs \theta}
  {The pattern $\cpat$ matches the shape $\sh$ with fresh components $\tvcs$ binding\\pattern variables in $\theta$.}
  \\
  \newcommand{\Mrule}[5][]{{#2} \Matches {(#3)} \; #4 &\eqdef& {#5} & #1}
  \begin{tabular}{RCLL}
    \Mrule[\text{if } n \geq j]
      {\cpatprod \tv j}
      {\any \tvcs \Pi\iton \tvcs} \tvcps
      {[\tv \is \tvc_j']}
    \\[1ex]
    \Mrule
      {\cpatrcd \ct}
      {\any \tvcs \trcd \T \tvcs} \tvcps
      {[\ct \is \T]}
    \\[1ex]
    \Mrule
      {\cpatpoly \cscm}
      {\any \tvcs \tpoly \ts} \tvcps
      {[\cscm \is \ts \where{\tvcs \is \tvcps}]}
  \end{tabular}
\end{judgboxmathpar}

\begin{judgboxmathpar}
  {\c \simple}
  {The constraint $\c$ is simple.}
  \label{fig:simple}
  \inferrule[Simple-True]
    { }
    {\ctrue \simple}

  \inferrule[Simple-False]
    { }
    {\cfalse \simple}

  \inferrule[Simple-Conj]
    {\ca \simple \\ \cb \simple}
    {\ca \cand \cb \simple}

  \inferrule[Simple-Exists]
    {\c \simple}
    {\cexists \tv \c \simple}

  \inferrule[Simple-Forall]
    {\c \simple}
    {\cfor \tv \c \simple}

  \inferrule[Simple-Unif]
    { }
    {\cunif \ta \tb \simple}

  \inferrule[Simple-Let]
    {\ca \simple \\ \cb \simple}
    {\clet \x \tv \ca \cb \simple}

  \inferrule[Simple-App]
    { }
    {\capp \x \t \simple}

  \inferrule[Simple-Multi-unif]
    { }
    {\ueq \simple}

  \inferrule[Simple-LetR]
    {\ca \simple \\ \cb \simple}
    {\cletr \x \tv \tvs \ca \cb \simple}

  \inferrule[Simple-Exists-Inst]
    {\c \simple}
    {\cexistsi \inst \x \c \simple}

  \inferrule[Simple-Incr-Inst]
    { }
    {\cpinst \inst \tv \t \simple}
\end{judgboxmathpar}

We write $\semenv \thsimple \c$ for the satisfiability of a simple constraint
$\c$. \\

\begin{judgboxmathpar}
  {\C \simple}
  {The constraint context $\C$ is simple.}
  \label{fig:simple-context}

  \inferrule[Simple-Ctx-Hole]
    { }
    {\hole \simple}

  \inferrule[Simple-Ctx-Conj-Left]
    {\C \simple \\ \c \simple}
    {\C \cand \c \simple}

  \inferrule[Simple-Ctx-Conj-Right]
    {\C \simple \\ \c simple}
    {\c \cand \C \simple}

  \inferrule[Simple-Ctx-Exists]
    {\C \simple}
    {\cexists \tv \C \simple}

  \inferrule[Simple-Ctx-Forall]
    {\C \simple}
    {\cfor \tv \C \simple}

  \inferrule[Simple-Ctx-Let-Abs]
    {\C \simple \\ \c \simple}
    {\clet \x \tv \C \c \simple}

  \inferrule[Simple-Ctx-Let-In]
    {\c \simple \\ \C \simple}
    {\clet \x \tv \c \C \simple}

  \inferrule[Simple-Ctx-Exists-Inst]
    {\C \simple}
    {\cexistsi \inst \x \C \simple}
\end{judgboxmathpar}

\begin{judgboxmathpar}
  {\cerase \c}
  {The erasure of $\c$.}
  \label{fig:erasure}
\newcommand{\Erule}[2]{\cerase {#1} &\eqdef& {#2}}
\begin{tabular}{RCL}
  \Erule{\ctrue}{\ctrue} \\
  \Erule{\cfalse}{\cfalse} \\
  \Erule{\ca \cand \cb}{\cerase \ca \cand \cerase \cb} \\
  \Erule{\cexists \tv \c}{\cexists \tv \cerase \c} \\
  \Erule{\cfor \tv \c}{\cfor \tv \cerase \c} \\
  \Erule{\cunif \ta \tb}{\cunif \ta \tb} \\
  \Erule{\clet \x \tv \ca \cb}{\clet \x \tv {\cerase \ca} {\cerase \cb}} \\
  \Erule{\capp \x \t}{\capp \x \t} \\
    \Erule{\cmatch \t {\cbranch {\bar \cpat} {\bar \c}}}{\ctrue} \\
  \Erule{\ueq}{\ueq} \\
  \Erule{\cletr \x \tv \tvs \ca \cb}{\cletr \x \tv \tvs {\cerase \ca} {\cerase \cb}} \\
  \Erule{\cexistsi \inst \x \c}{\cexistsi \inst \x \cerase \c}\\
  \Erule{\cpinst \inst \tv \t}{\cpinst \inst \tv \t}
\end{tabular}
\end{judgboxmathpar}

\begin{judgboxmathpar}
  {\semenv \Th \c}
  {Under the semantic environment $\semenv$,
   the constraint $\c$ is canonically satisfiable.}
  \label{fig:canonical-sem}
  \inferrule[Can-Simple]
    {\semenv \thsimple \c}
    {\semenv \Th \c}

  \inferrule[Can-Match-Ctx]
    {\Cshape \C \t \sh \\ \semenv \Th \C\where{\cmatched \t \sh \cbrs}}
    {\semenv \Th \C\where{\cmatch \t \cbrs}}
\end{judgboxmathpar}

\begin{judgboxmathpar}
  {{\labfrom \elab \T} \leq \tya \to \tyb}
  {The label $\elab$ of the record $\T$ has the field type $\tyb$ and record type $\tya$.}
  \inferrule[Lab-Inst]
    {\labenv(\labfrom \elab \T) = \tfor \tvs {\trcd \T \tvs} \to \t }
    {{\labfrom \elab \T} \leq \trcd \T \tys \to \t\where{\tvs \is \tys} }
\end{judgboxmathpar}

\judgbox
  {\labuni \elab \T}
  {The label $\elab$ infers the unique record name $\T$.}

\begin{judgboxmathpar}
  {\labsuni \elabs \T}
  {The \emph{closed} set of labels $\elabs$ infer the unique record name $\T$.}
  \infer[Lab-Uni]
    {\elab \in \Dom {\labenv(\T)} \\\\
     \forall \Tp, \uad \elab \in \Dom {\labenv(\Tp)} \implies {\T} = \Tp}
    {\labuni \elab \T}

  \infer[Labs-Uni]
    {\elabs = \Dom {\labenv(\T)}  \\\\
     \forall \Tp, \uad \Dom{\labenv(\Tp)} = \elabs \implies {\T} = \Tp}
    {\labsuni \elabs \T}
\end{judgboxmathpar}

\begin{judgboxmathpar}
  {\G \th \e : \ts}
  {Under the typing context $\G$, the term $\e$ is assigned the type $\ts $}
  \inferrule[Var]
    {x : \sigma \in \G}
    {\G \th x : \sigma}

  \inferrule[Fun]
    {\G, x : \ta \th e : \tb }
    {\G \th \efun x e : \ta \to \tb}

  \inferrule[App]
    {\G \th \ea : \ta \to \tb \\
     \G \th \eb : \ta}
    {\G \th \eapp \ea \eb : \tb}

  \inferrule[Unit]
    { }
    {\G \th () : 1}

  \inferrule[Annot]
    {\G \th e : \t\where {\tvs \is \tys}}
    {\G \th (e : \exi \tvs \t) : \t\where {\tvs \is \tys}}

  \inferrule[Gen]
    {\G \th e : \sigma \\ \tv \disjoint \G}
    {\G \th e : \tfor \tv \sigma}

  \inferrule[Inst]
    {\G \th e : \tfor \tv \ts}
    {\G \th e : \ts \where{\tv \is \t}}

  \inferrule[Let]
    {\G \th \ea : \sigma \\
     \G, x : \sigma \th \eb : \t}
    {\G \th \elet x \ea \eb : \t}

  \inferrule[Tuple]
    {\parens{\G \th \ei : \ti}\iton}
    {\G \th (\ea, \ldots, \en) : \Pi\iton \ti}
\\
  \inferrule[Proj-X]
    {\G \th \e : \Pi\iton \ti \\
     1 \leq j \leq n}
    {\G \th \exproj \e j n : \tj}

  \inferrule[Proj-I]
    {\eshape \E \e {\any \tvcs \Pi\iton \tvcs} \\
     \G \th \E\where{\exproj \e j n} : \t}
    {\G \th \E\where{\eproj \e j} : \t}

  \inferrule [Poly-X]
    {\G \th \e : \ts\where {\tvs \is \tys}}
    {\G \th \expoly \e \tvs \ts : \tpoly {\ts \where {\tvs \is \tys}}}

  \inferrule [Poly-I]
    {\Eshape \E \e {{\any \tvcs \tpoly \ts}} \\
     \G \th \E \where{\expoly \e \tvcs \ts} : \t}
    {\G \th \E \where{\epoly \e} : \t}

  \inferrule [Use-X]
    {\G \th \e : \tpoly \ts \where {\tvs \is \tys}}
    {\G \th \exinst e \tvs \ts : \ts \where {\tvs \is \tys}}

  \inferrule [Use-I]
    {\eshape \E  \e {\any \tvcs \tpoly \ts} \\
     \G \th \E\where{\exinst \e \tvcs \ts} : \t}
    {\G \th \E\where{\einst \e} : \t}
\\
  \inferrule[Rcd-X]
    {\elabs = \Dom {\labenv(\T)} \\
     \parens{{\labfrom \elabi \T} \leq \t \to \tyi}\iton \\
     \parens{\G \th \ei : \tyi}\iton }
    {\G \th \exrecord \T {\elaba = \ea; \ldots; \elab_n = \en} : \t }

  \inferrule[Rcd-Closed]
    {\labsuni \elabs \T \\ \G \th \exrecord \T {\overline{\elab = \e}} : \t}
    {\G \th \erecord {\overline{\elab = \e}} : \t}

  \inferrule[Rcd-I]
    {\Eshape \E \es {\any \tvcs \trcd \T \tvcs} \\
     \G \th \E\where{\exrecord \T {\overline{\elab = \e}}} : \t }
    {\G \th \E\where{\erecord {\overline{\elab = \e}}} : \t}

  \inferrule[Rcd-Proj-X]
    {{\labfrom \elab \T} \leq \tya \to \tyb \\ \G \th \e : \tya }
    {\G \th \exfield \e \T \elab : \tyb}

  \inferrule[Rcd-Proj-Closed]
    {\labuni \elab \T \\ \G \th \exfield \e \T \elab : \t }
    {\G \th \efield \e \elab : \t}

  \inferrule[Rcd-Proj-I]
    {\eshape \E \e {\any \tvcs \trcd \T \tvcs} \\
     \G \th \E\where{\exfield \e \T \elab} : \t}
    {\G \th \E\where{\efield \e \elab} : \t}

  \inferrule[Hole]
    {\parens{\G \th \ei : \ti}\iton}
    {\G \th \emagic {\bar \e} : \tp}

\def \Eqdef {&\eqdef&}
{\begin{tabular}{RCL}
\eshape \E \e \sh \Eqdef
  \forall \G, \t, \gt, \uad
  \G \th \eerase {\E \where {\emagic {\eannot \e {} \gt }}} : \t
      \wide\implies \shape \gt = \sh
\\[1ex]
\Eshape \E {\bar\e} \sh \Eqdef
  \forall \G, \t, \gt, \uad
      \G \th \eerase {\E\where{\eannotmagic {\bar \e} {} \gt}} : \t
      \wide\implies \shape \gt = \sh
\\[1ex]
\end{tabular}}
\end{judgboxmathpar}

\judgbox
  {\cinfer {\G \th \e} \t}
  {$\cinfer {\G \th \e} \t$ is satisfiable iff
   $\e$ has the expected \emph{known} type $\t$ under \emph{known} context $\G$}

\smallskip
\judgbox
  {\cinfer \e \ts}
  {$\cinfer \e \ts$ is satisfiable iff
   $\e$ has the expected \emph{known} type scheme $\ts$.}

\smallskip
\begin{judgboxmathpar}
  {\cinfer \e \t}
  {$\cinfer \e \t$ is satisfiable iff
   $\e$ has the expected \emph{known} type $\t$.}
\newcommand {\Crule}[2]{#1 &\eqdef& #2}
\def \arraystretch{1.2}
\begin{tabular}{.L;CL.}
\Crule
   {\cinfer x \t}
   {\cinst x \t}
\\
\Crule
  {\cinfer {()} \t}
  {\cunif \t \tunit}
\\
\Crule
  {\cinfer {\efun \x \e} \t}
  {\cexists {\tva, \tvb}
    \clet \x \tvp {\cunif \tvp \tv} {\cinfer \e \tvb}
    \cand \cunif \t {\tva \to \tvb}}
\\
\Crule
  {\cinfer {\eapp \ea \eb} \t}
  {\cexists {\tva, \tvb}
    \cinfer \ea {\tvb} \cand \cinfer \eb \tva
    \cand \cunif \tvb {\tva \to \t}}
\\
\Crule
  {\cinfer {\elet \x \ea \eb} \t}
  {\clet \x \tv {\cinfer \ea \tv} {\cinfer \eb \t}}
\\
\Crule
  {\cinfer {\eannot \e \tvs \tp} \t}
  {\cexists \tvs  \cinfer \e \tp \cand \cunif \t \tp}
\\
\Crule
  {\cinfer {\etuple {\ea, \ldots, \en}} \t}
  {\cexists \tvs \cunif \t {\tProd \tvs}
    \cand \cAnd \iton \cinfer \ei {\tv_i}}
\\
\Crule
  {\cinfer {\exproj \e j n} \t}
  {\cexists {\tv, \tvs}
    \cinfer \e \tv
    \cand \cunif \tv {\tProd \tvs}
    \cand \cunif \t {\tvj}}
\\
\Crule
  {\cinfer {\eproj \e j} \t}
  {\cexists \tv \cinfer \e \tv
    \cand \cmatch \tv {\cbranch {\cpatprod \tvb j} {\cunif \t \tvb}}}
\\
\Crule
  {\cinfer {\expoly \e \tvs \ts} \t}
  {\cexists {\tvs}
    \cinfer \e \ts
    \cand \cunif \t {\tpoly \ts}}
\\
\Crule
  {\cinfer {\exinst \e \tvs \ts} \t}
  {\cexists {\tvs, \tvb}
    \cinfer \e \tvb
    \cand \cunif \tvb {\tpoly \ts}
    \cand \ts \leq \t}
\\
\Crule
  {\cinfer {\einst \e} \t}
  {\cexists \tv
    \cinfer \e \tv
    \cand \cmatch \tv {\cbranch {\cpatpoly \cscm} \cscm \leq \t}}
\\
\Crule
  {\cinfer {\epoly \e} \t}
  {\clet \x \tv {\cinfer \e \tv}
    {\cmatch \t {\cbranch {\cpatpoly \cscm} {\x \leq \cscm}}}}
\\
\Crule
  {\cinfer {\efield \e \elab} \t}
  {\!\begin{cases}
    \cinfer {\exfield \e \T \elab} \t
    & \text{if } \labuni \elab \T
    \\
    \cexists \tv \cinfer \e \tv \cand
    \cmatch \tv
      {\cbranch {\cpatrcd \ct} {\labfrom \elab \ct \leq \tv \to \t}}
    \hskip 6ex
    & \text{otherwise}
  \end{cases}}
\\
\Crule
  {\cinfer {\exfield \e \T \elab} \t}
  {\cexists \tv \cinfer \e \tv \cand
   \labfrom \elab \T \leq \tv \to \t}
\\
\Crule
  {\cinfer {\erecord {\overline{\elab = \e}}} \t}
  {\!\begin{cases}
    \cinfer {\exrecord \T {\overline{\elab = \e}}} \t
    & \text{if } \labsuni \elabs \T \\
    \cexists \tvs \cAnd\iton \cinfer \ei \tvi & \text{otherwise} \\
    \uad\cand\uad \cmatch \t {\cbranch {\cpatrcd \ct}
      {\parens{\dom \ct = \elabs \cand
               \cAnd\iton \labfrom \elabi \ct \leq \t \to \tvi}}}
    \hskip -2em
   \end{cases}}
\\
\Crule
  {\cinfer {\exrecord \T {\overline{\elab = \e}}} \t}
  {\cexists \tvs \cAnd\iton \cinfer \ei \tvi \cand \dom {\T} = \elabs \cand \cAnd\iton \labfrom \elabi \T \leq \t \to \tvi }
\\
\Crule
  {\cinfer {\emagic \es} \t}
  {\cexists \tvs \cAnd\iton \cinfer \ei \tvi}
\\[1em]
\Crule
  {\cinfer \e {\tfor \tvs \t}}
  {\cfor \tvs \cinfer \e \t}
\\[1em]
\Crule
  {\cinfer {\eset \th \e} \t}
  {\cinfer \e \t}
\\
\Crule
  {\cinfer {\x : \ts, \G \th \e} \t}
  {\clet \x \tv {\ts \leq \tv} {\cinfer {\G \th \e} \t}}
\end{tabular}
\end{judgboxmathpar}
\emph{Note: When $\t$ is $\tv$ it is considered an \emph{unknown} expected type.}

\begin{judgboxmathpar}
  {\e \simple}
  {The term $\e$ is simple.}
  \inferrule[Simple-Var]
    { }
    {\x \simple}

  \inferrule[Simple-Fun]
    {\e \simple}
    {\efun \x \e \simple}

  \inferrule[Simple-App]
    {\ea \simple \\ \eb \simple}
    {\eapp \ea \eb \simple}

  \inferrule[Simple-Unit]
    { }
    {\eunit \simple}

  \inferrule[Simple-Let]
    {\ea \simple \\ \eb \simple}
    {\elet \x \ea \eb \simple}

  \inferrule[Simple-Annot]
    {\e \simple}
    {\eannot \e \tvs \t \simple}

  \inferrule[Simple-Tuple]
    {\parens {\ei \simple}\iton}
    {\etuple {\ea, \ldots, \en} \simple}

  \inferrule[Simple-Proj-X]
    {\e \simple}
    {\exproj \e j n \simple}

  \inferrule[Simple-Poly-X]
    {\e \simple}
    {\expoly \e \tvs \ts \simple}

  \inferrule[Simple-Use-X]
    {\e \simple}
    {\exinst \e \tvs \ts \simple}

  \inferrule[Simple-Rcd-X]
    {\parens {\ei \simple} \iton}
    {\exrecord \T {\elaba = \ea\; \ldots\; \elab_n = \en}}

  \inferrule[Simple-Rcd-Closed]
    {\parens {\ei \simple} \iton \\ \labsuni \elabs \T}
    {\erecord {\elaba = \ea\; \ldots\; \elab_n = \en}}

  \inferrule[Simple-Rcd-Proj-X]
    {\e \simple}
    {{\exfield \e \T \elab} \simple}

  \inferrule[Simple-Rcd-Proj-Closed]
    {\e \simple \\ \labuni \elab \T}
    {\efield \e \elab \simple}

  \inferrule[Simple-Hole]
    {\parens{\ei \simple}\iton}
    {(\emagic \es) \simple}
\end{judgboxmathpar}

\label {app/ref/eerase}
\begin{judgboxmathpar}
  {\eerase \e}
  {The erasure of $\e$.}
\newcommand{\Erule}[2]{\eerase {#1} &\eqdef& {#2}}
  \begin{tabular}{RCL}
  \Erule{\x}{\x} \\
  \Erule{\efun \x \e}{\efun \x \eerase \e} \\
  \Erule{\eapp \ea \eb}{\eapp {\eerase \ea} {\eerase \eb}} \\
  \Erule{\eunit}{\eunit} \\
  \Erule{\elet \x \ea \eb}{\elet \x {\eerase \ea} {\eerase \eb}} \\
  \Erule{\eannot \e \tvs \t}{\eannot {\eerase \e} \tvs \t} \\
  \Erule{\etuple {\ea, \ldots, \en}}{\etuple {\eerase \ea, \ldots, \eerase \en}} \\
  \Erule{\eproj \e j}{\emagic {\eerase \e}} \\
  \Erule{\exproj \e j n}{\exproj {\eerase \e} j n} \\
  \Erule{\expoly \e \tvs \ts}{\expoly {\eerase \e} \tvs \ts} \\
  \Erule{\epoly \e}{\emagic {\eerase \e}}\\
  \Erule{\einst \e}{\emagic {\eerase \e}}\\
  \Erule{\exinst \e \tvs \ts}{\exinst {\eerase \e} \tvs \ts}\\
    \Erule{\erecord {\elaba = \ea; \ldots; \elab_n = \en}}{\begin{cases}
      \erecord {\elaba = \eerase \ea; \ldots; \elab_n = \eerase \en} &\text{if } \labsuni \elabs \T \\
      \emagic {\eerase \ea, \ldots, \eerase \en} & \text{otherwise}
    \end{cases}}\\
  \Erule{\exrecord \T {\elaba = \ea; \ldots; \elab_n = \en}}{\exrecord \T {\elaba = \eerase \ea; \ldots; \elab_n = \eerase \en}}\\
    \Erule{\efield \e \elab}{\begin{cases}
      \efield {\eerase \e} \elab & \text{if } \labuni \elab \T \\
      \emagic {\eerase \e} & \text{otherwise}
    \end{cases}}\\
  \Erule{\exfield \e \T \elab}{\exfield {\eerase \e} \T \elab}\\
  \Erule{\emagic \es}{\emagic {\parens {\eerase \ei} \iton}}\\
\end{tabular}
\end{judgboxmathpar}

\begin{judgboxmathpar}
  {\G \thsimplesd \e : \t}
  {Under the typing context $\G$, the simple term $\e$ has the type $\t$.}
\\
  \inferrule[Var-SD]
    {x : \tfor \tvs \t \in \G}
    {\G \thsimplesd x : \t\where{\tvs \is \tys}}

  \inferrule[Let-SD]
    {\G \thsimplesd \ea : \ta\\
     \tvs \disjoint \fvs \G \\
     \G, x : \tfor \tvs \ta \thsimplesd \eb : \tb}
    {\G \thsimplesd \elet x \ea \eb : \tb}

  \inferrule [Poly-X]
    {\G \thsimplesd \e : \t\where {\tvs \is \tys} \\ \tvbs \disjoint \G}
    {\G \thsimplesd \expoly \e \tvs {\tfor \tvbs \t} : \tpoly {\ts \where {\tvs \is \tys}}}

  \inferrule [Use-X]
    {\G \th \e : \tpoly {\tfor \tvbs \t} \where {\tvs \is \tys}}
    {\G \th \exinst e \tvs {\tfor \tvbs \t} : \t \where {\tvs \is \tys, \tvbs \is \typs}}
\end{judgboxmathpar}

\begin{judgboxmathpar}
  {\Th \e : \t}
  {The term $\e$ canonically has the type $\t$.}
  \inferrule[Can-Base]
    {\eset \thsimplesd \e : \t}
    {\Th \e : \t}

  \inferrule[Can-Proj-I]
    {\eshape \E \e {\any \tvcs \Pi\iton \tvcs} \\
     \Th \E\where{\exproj \e j n} : \t}
    {\Th \E\where{\eproj \e j} : \t}

  \inferrule [Can-Poly-I]
    {\Eshape \E \e {{\any \tvcs \tpoly \ts}} \\
     \Th \E \where{\expoly \e \tvcs \ts} : \t}
    {\Th \E \where{\epoly \e} : \t}

  \inferrule [Can-Use-I]
    {\eshape \E  \e {\any \tvcs \tpoly \ts} \\
     \Th \E\where{\exinst \e \tvcs \ts} : \t}
    {\Th \E\where{\einst \e} : \t}

  \inferrule[Can-Rcd-I]
    {\Eshape \E \es {\any \tvcs \trcd \T \tvcs} \\
     \Th \E\where{\exrecord \T {\elaba = \ea; \ldots; \elab_n = \en}} : \t}
    {\Th \E\where{\erecord {\elaba = \ea; \ldots; \elab_n = \en}} : \t}

  \inferrule[Can-Rcd-Proj-I]
    {\Eshape \E \e {\any \tvcs \trcd \T \tvcs} \\
     \Th \E\where{\exfield \e \T \elab} : \t}
    {\Th \E\where{\efield \e \elab} : \t}
\end{judgboxmathpar}

\label {app/rules/unif}
\begin{judgboxmathpar}
  {\up \unif \upp}
  {The unifier rewrites $\up$ to $\upp$.}
   \rewrite[U-Exists]
      {(\cexists \alpha \upa) \cand \upb }{ \tv \disjoint \upb}
      {\cexists \tv {\upa \cand \upb}}

    \rewrite[U-Cycle]
      {\up }{ \cyclic \up}
      {\cfalse}

    \rewrite[U-True]
      {\up \cand \ctrue}
      {}
      {\up}

    \rewrite[U-False]
      {\Up\where{\cfalse}}
      { \Up \neq \hole}
      {\cfalse}

    \rewrite[U-Merge]
      {\cunif \tv \ueqa \cand \cunif \tv \ueqb}
      {}
      {\cunif \tv {\cunif \ueqa \ueqb}}

    \rewrite[U-Stutter]
      {\cunif \tv {\cunif \tv \ueq}}
      {}
      {\cunif \tv \ueq}

    \rewrite[U-Name]
      {\cunif {\pshapp \parens{\tys, \ti, \typs}} \ueq }
      {\tv \disjoint \tys, \typs, \ueq \\ \ti \notin \TyVars}
      {\cexists \tv
        {\cunif \tv \ti \cand
         \cunif {\pshapp \parens{\tys, \tv, \typs}} \ueq}}

    \rewrite[U-Decomp]
      {\cunif {\pshapp \tvs} {\cunif {\pshapp \tvbs} \ueq}}
      {}
      {\cunif {\pshapp \tvs} \ueq \cand \cunif \tvs \tvbs}

    \rewrite[U-Clash]
      {\cunif {\pshapp \tvs} {\cunif {\pshapp[\shp]\tvbs } \ueq }}{
       \sh \neq \shp}
      {\cfalse}

    \rewrite[U-Trivial]
      {\ueq}
      {|\ueq| \leq 1}
      {\ctrue}
\end{judgboxmathpar}

\judgbox
  {\tva \prec_\up \tvb}
  {$\tva$ occurs in $\tvb$ in the constraint $\up$.}
\smallskip
\begin{judgboxmathpar}
  {\cyclic \up}
  {$\up$ contains a cyclic type.}

  \inferrule[Occurs]
    {\tva \in \fvs \t \\
     \t \notin \TyVars \\
     \tva, \tvb \disjoint \bvs \Up}
    {\tva \prec_{\Up\where{\cunif {\tvb~} {~\cunif {\t~} ~\ueq}}} \tvb}

  \inferrule[Cycle]
    {\tva \prec_\up^+ \tva}
    {\cyclic \up}
\end{judgboxmathpar}

\begin{definition}[Solved form $\hat\up$]
  We write $\hat\up$ for constraints in \emph{solved form}, that is,
  constraints of the form $\cexists \tvs \ueqs$, where:
\begin{enumerate*}
  \item
    each $\ueqi$ contains at most one non-variable type;
  \item
    each type variable may occur as a member of at most one multi-equation
    $\ueqi$;
  \item the constraint is acyclic.
\end{enumerate*}\relax
\end{definition}

\label {app/ref/solver}
\begin{judgboxmathpar}
  {\c \csolve \cp}
  {The constraint solver rewrites $\c$ to $\cp$.}

  \rewrite[S-Unif]
    {\upa}
    {\upa \unif \upb}
    {\upb}

  \rewrite[S-True]
    {C \cand \ctrue}
    {}
    {C}

  \rewrite[S-False]
    {\C\where\cfalse}
    {\C \neq \hole}
    {\cfalse}

  \rewrite[S-Let]
    {\clet \x \tv \ca \cb}
    {}
    {\cletr \x \tv \eset \ca \cb}

  \rewrite[S-Exists-Conj]
    {(\cexists \alpha \ca) \cand \cb }{
     \tv \disjoint \cb}
    {\cexists \tv {\ca \cand \cb}}

  \rewrite[S-Let-ExistsLeft]
    {\cletr \x \tv \tvs {\cexists \tvb \ca} \cb }{
     \tvb \disjoint \tv, \tvs, \cb}
    {\cletr \x \tv {\tvs, \tvb} \ca \cb}

  \rewrite[S-Let-ExistsRight]
    {\cletr \x \tv \tvs \ca {\cexists \tvb \cb} }{
     \tvb \disjoint \tv, \tvs, \ca}
    {\cexists \tvb {\clet \x \tvs \ca \cb}}

  \rewrite[S-Let-ConjLeft]
    {\cletr \x \tv \tvs {\ca \cand \cb} \cc }{
     \ca \disjoint \tv, \tvs}
    {\ca \cand \cletr \x \tv \tvs \cb \cc}

  \rewrite[S-Let-ConjRight]
    {\cletr \x \tv \tvs \ca (\cb \cand \cc) }{
     \x \disjoint \cc}
    {\cc \cand \Clet \x \tv \ca \cb}

  \rewrite[S-Match-Ctx]
    {\C\where{\cmatch \t \cbrs}}
    {\th \Cshape \C \t \sh}
    {\C\where{\cmatched \t {\sh} \cbrs}}

   \rewrite[S-Inst-Name]
    {\cpinst \inst \tv \t}
    {\t \notin \TyVars}
    {\cexists \tvc \cunif \tvc \t \cand \cpinst \inst \tv \tvc}

  \rewrite[S-Let-AppR]
    {\cletr \x \tv \tvs \c {\C\where{\capp \x \t}}}
    {\tvc \disjoint \t \\ \x \disjoint \bvs \C}
    {\cletr \x \tv \tvs \c {\C\where{\cexistsi
      {\tvc, \inst} \x { \cpinst \inst \tv \tvc \cand \cunif \tvc \t}}}}

  \rewrite[S-Inst-Copy]
    {\cletr \x \tv \tvs {\c} \C\where{\cpapp \x \tvb \tvc \inst}}
    {\c = \cp \cand \cunif \tvb {\cunif {\shapp \tvbs} \ueq}\\\\
     \tvb \in \reg \tv \tvs \\
     \neg \cyclic {\c} \\
     \tvbs' \disjoint \tvb, \tvc, \tvbs \\
     \x \disjoint \bvs \C}
    {\cletr \x \tv \tvs {\c}
      \C\where{\cexists {\tvbs'}
         {\cunif \tvc {\shapp \tvbs'} \cand \overline{\cpapp \x {\tvb} {\tvb'} \inst}}}}

  \rewrite[S-Inst-Unify]
    {\cpinst \inst \tvb \tvca \cand \cpinst \inst \tvb \tvcb}
    {}
    {\cpinst \inst \tvb \tvca \cand \cunif \tvca \tvcb}

  \rewrite[S-Inst-Poly]
    {\cletr \x \tv {\tvs} {\ueqs \cand \c}
        {\C\where{\cpapp \x \tvb \tvc \inst}}}
    {\cfor \tvb \cexists {\tvbs} {\ueqs} \cequiv \ctrue \\\\
     \tvb, \tvbs \subseteq \tv, \tvs \\
     \tvb, \tvbs \disjoint \c \\
     \inst(\tvb) \disjoint \insts \C \\
     \x \disjoint \bvs \C}
    {\cletr \x \tv {\tvs} {\ueqs \cand \c} {\C\where\ctrue}}

  \rewrite[S-Inst-Mono]
    {\cletr \x \tv \tvs \c {\C\where{\cpapp \x \tvb \tvc \inst}}}
    {\\\tvb \notin \reg \tv \tvs \\ \x, \tvb \disjoint \bvs \C}
    {\cletr \x \tv \tvs \c {\C\where{\cunif \tvb \tvc}}}

  \rewrite[S-Compress]
    {\cletr \x \tv {\tvs, \tvb}
       {\ca \cand \cunif \tvb {\cunif \tvc \ueq}} {\cb}}
    {\tvb \neq \tvc}
    {\cletr \x \tv {\tvs}
       {\ca\where{\tvb \is \tvc} \cand \cunif \tvc {\ueq\where{\tvb \is \tvc}}}
       {\cb\where{\x(\tvb) \is \tvc}}}

  \rewrite[S-Gc]
    {\cletr \x \tv {\tvs, \tvb} {\ca \cand \cunif \tvb \ueq} \cb}
    {\tvb \disjoint \ca, \ueq, \cb}
    {\cletr \x \tv {\tvs} {\ca \cand \ueq} \cb}

  \rewrite[S-Exists-Lower]
    {\cletr \x \tv {\tvas, \tvbs} \ca \cb}
    {\th \cdetermines {\cexists {\tv, \tvas} \ca} \tvbs}
    {\cexists \tvbs \cletr \x \tv \tvas \ca \cb}

  \rewrite[S-Exists-Exists-Inst]
    {\cexistsi \inst \x \cexists \tv \c}
    {}
    {\cexists \tv \cexistsi \inst \x \c}

  \rewrite[S-Exists-Inst-Conj]
    {\cexistsi \inst \x \ca \cand \cb}
    {\inst \disjoint \ca}
    {\ca \cand \cexistsi \inst \x \cb}

  \rewrite[S-Exists-Inst-Let]
    {\cletr \x \tv \tvs \ca {\cexistsi \inst \xp \cb}}
    {\x \neq \xp}
    {\cexistsi \inst \xp \cletr \x \tv \tvs \ca \cb}

  \rewrite[S-Exists-Inst-Solve]
    {\cexistsi \inst \x \c}
    {\inst \disjoint \c}
    {\c}

  \rewrite[S-All-Conj]
    {\cfor \tvs {\cexists \tvbs {\ca \cand \cb}}}
    {\tvs, \tvbs \disjoint \ca}
    {\ca \cand \cfor \tvs {\cexists \tvbs \cb}}

  \rewrite[S-Exists-All]
    {\cfor \tvs {\cexists {\tvbs, \tvcs} \c}}
    {\th \cdetermines {\cexists {\tvs, \tvbs} \c} \tvcs}
    {\cexists \tvcs \cfor \tvs {\cexists \tvbs \c}}

  \rewrite[S-All-Escape]
    {\cfor {\tvs, \tv} {\cexists \tvbs {\c \cand \ueqs}}}
    {\tv \prec_{\ueqs}^* \tvc \\
     \tvc \disjoint \tv, \tvbs \\
     \tv \disjoint \tvbs}
    {\cfalse}

  \rewrite[S-All-Rigid]
    {\cfor {\tvs, \tv}
      {\cexists \tvbs \c \cand \cunif \tv {\cunif \t \ueq}}}
    {\t \notin \TyVars \\ \tv \disjoint \tvbs}
    {\cfalse}

  \rewrite[S-All-Solve]
    {\cfor \tvs \cexists \tvbs \ueqs}
    {\cexists \tvbs \ueqs \cequiv \ctrue}
    {\ctrue}
\end{judgboxmathpar}

\begin{judgboxmathpar}
  {\th \Cshape \C \t \sh}
  {Under $\C$, the type $\t$ has the provably unique canonical shape $\sh$.}
  \infer[S-Uni-Var]
    {\color{gray}\tv \disjoint \bvs \Cb}
    {\th \Cshape {\Ca\where{\cunif \tv {\cunif \t \ueq} \cand \Cb\where{-}}} \tv {~\shape \t}}

  \infer[S-Uni-Type]
    {{\color{gray}\t \notin \TyVars}}
    {\th \Cshape \C \t {~\shape \t}}

  \infer[S-Uni-BackProp]
    {\th \Cshape{\cletr \x \tv \tvs {\Ca\where{\ctrue}} {\Cb\where{\cpapp \x \tvp \tvc \inst \cand -}}} \tvc \sh \\
     \color{gray}\tvp \in \tv, \tvs \\
     \color{gray}\x \disjoint \bvs \Cb \\
     \color{gray}\tvp \disjoint \bvs \Ca}
    {\th \Cshape{\cletr \x \tv \tvs {\Ca\where{-}} {\Cb\where{\cpapp \x \tvp \tvc \inst}}} \tv \sh}

\end{judgboxmathpar}

\begin{definition}
  $\cdetermines \c \tvbs$ if and only if every ground assignments
  $\semenv$ and $\semenvp$ that satisfy (the erasure of) $\c$ and coincide outside of $\tvb$
  coincide on $\tvbs$ as well.
  \begin{mathpar}
    \cdetermines \c \tvb \uad\eqdef\uad \all {\semenv, \semenvp} \uad
      \semenv \th \cerase \c
      \wedge \semenvp \th \cerase \c
      \wedge \semenv =_{\setminus \tvbs} \semenvp
      \implies
      \semenv = \semenvp
  \end{mathpar}
\end{definition}

\begin{judgboxmathpar}
  {\th \cdetermines \c \tvs}
  {$\c$ provably determines $\tvs$.}

  \inferrule
    [S-Det-Dom]
    {\tvc \disjoint \tvbs, \tvas \\ \tvs \subseteq \fvs \ueq}
    {\th \cdetermines {\cexists \tvbs \c \cand \cunif \tvc \ueq} \tvs}

  \inferrule
    [S-Det-Esc]
    {\fvs \t \disjoint \tvs, \tvbs}
    {\th \cdetermines {\cexists \tvbs \c \cand \cunif \tvs {\cunif \t \ueq}} \tvs}
\end{judgboxmathpar}

\begin{judgboxmathpar}
  {\insts \c}
  {The set of instantiations in $\c$.}
  \newcommand{\Srule}[2]{#1 &\eqdef& #2}
  \begin{tabular}{RCL}
    \Srule{\insts \ctrue}{\eset}\\
    \Srule{\insts \cfalse}{\eset}\\
    \Srule{\insts {\ca \cand \cb}}{\insts \ca \cup \insts \cb}\\
    \Srule{\insts {\cexists \tv \c}}{\insts \c}\\
    \Srule{\insts {\cfor \tv \c}}{\insts \c}\\
    \Srule{\insts {\cunif \t \tp}}{\eset}\\
    \Srule{\insts {\clet \x \tv \ca \cb}}{\insts \ca \cup \insts \cb}\\
    \Srule{\insts {\capp \x \t}}{\eset}\\
    \Srule{\insts {\ueq}}{\eset}\\
    \Srule{\insts {\cletr \x \tv \tvs \ca \cb}}{\insts \ca \cup \insts \cb}\\
    \Srule{\insts {\cexistsi \inst \x \c}}{\insts \c}\\
    \Srule{\insts {\cpinst \inst \tv \t}}{\set {\inst(\tv)}}
  \end{tabular}
\end{judgboxmathpar}

\begin{definition}[Normal forms]
  \label{def/normal-forms-appendix}
  A constraint $\c$ is in \emph{normal form}, written $\hat\c$, if no 
  rewriting rules apply, \ie $\c \cnsolve$. Similarly, a constraint context 
  $\C$ is in \emph{normal form}, written $\hat\C$, if $\C\where\ctrue$ is in normal form.
\end{definition}

\begin{definition}[Measure]
  For the relation $\semenv \th \c$, the following measure enables a useful
  induction principle:
    \begin{mathpar}
    \cmeasure \c \uad\eqdef\uad \angles{\cnmatches \c, \csize \c}
  \end{mathpar}
  where $\angles \ldots$ denotes a pair with lexicographic ordering, and:
  \begin{enumerate}

    \item $\cnmatches \c$ is the number of $\cmatch \t \cbrs$ constraints in
      $\c$.

    \item the last component $\csize \c$ is a structural measure of constraints \ie a
      conjunction $\ca \cand \cb$ is larger than the two conjuncts $\ca,
      \cb$.

  \end{enumerate}
\end{definition}

\judgbox
  {\ctxinfer \E \tp \t}
  {$\ctxinfer \E \tp \t$ is a satisfiable context iff the context $\E$ has \\ the expected type $\t$ given the hole has the type $\tp$.}

{
\newcommand {\Crule}[2]{#1 &\eqdef& #2}
\def \arraystretch{1.4}
\small
\begin{longtable}{.L;C;L.}
  \\
  \Crule{\ctxinfer \hole \t \t}{\hole}\\
  \Crule{\ctxinfer {\parens{\eapp \E \e}} \tp \t}{\cexists {\tva, \tvb} \cunif \tva {\tvb \to \t} \cand \ctxinfer \E \tp \tva \cand \cinfer \e \tvb}\\
  \Crule{\ctxinfer {\parens{\eapp \e \E}} \tp \t}{\cexists {\tva, \tvb} \cunif \tva {\tvb \to \t} \cand \cinfer \e \tva \cand \ctxinfer \E \tp \tvb}\\
  \Crule{\ctxinfer {\parens{\elet \x \E \e}} \tp \t}{\clet \x \tv {\ctxinfer \E \tp \tv} {\cinfer \e \t}}\\
  \Crule{\ctxinfer {\parens{\elet \x \e \E}} \tp \t}{\clet \x \tv {\cinfer \e \tv} {\ctxinfer \E \tp \t}}\\
  \Crule{\ctxinfer {{\eannot \E \tvs \tpp}} \tp \t}{\cexists \tvs \cunif \t \tpp \cand \ctxinfer \E \tp \t}\\
  \Crule{\ctxinfer {\etuple {\ea, \ldots, \E_j, \ldots, \en}} \tp \t}{\cexists \tvs \cunif \t {\tProd \tvi} \cand \cAnd_{i \neq j} \cinfer \ei \tvi \cand \ctxinfer {\E_j} \tp \tvj} \\
  \Crule{\ctxinfer {\parens{\exproj \E j n}} \tp \t}{\cexists {\tv, \tvs} \ctxinfer \E \tp \tv \cand \cunif \tv {\tProd \tvi} \cand \cunif \t \tvj}\\
  \Crule{\ctxinfer {\parens{\eproj \E j}} \tp \t}{\cexists \tv \ctxinfer \E \tp \tv \cand \cmatch \tv {\cbranch {\cpatprod \tvc j} {\cunif \t \tvc}}}\\
  \Crule{\ctxinfer {\expoly \E \tvs \ts} \tp \t}{\cexists \tvs \ctxinfer \E \tp \ts \cand \cunif \t {\tpoly \ts}}\\
  \Crule{\ctxinfer {\exinst \E \tvs \ts} \tp \t}{\cexists {\tvs, \tvb} \ctxinfer \E \tp \tvb \cand \cunif \tvb {\tpoly \ts} \cand \cleq \ts \t}\\
  \Crule{\ctxinfer {\epoly \E} \tp \t}{\clet \x \tv {\ctxinfer \E \tp \tv} {\\&&\cmatch \t {\cbranch {\cpatpoly \cscm} {\cleq \x \cscm}}}}\\
  \Crule{\ctxinfer {\einst \E} \tp \t}{\cexists \tv {\ctxinfer \E \tp \tv \cand \cmatch \tv {\cbranch {\cpatpoly \cscm} {\cleq \cscm \t}}}}\\
\Crule
  {\ctxinfer {\parens{\efield \E \elab}} \tp \t}
  {\!\begin{cases}
    \ctxinfer {\parens{\exfield \E \T \elab}} \tp \t
    & \text{if } \labuni \elab \T
    \\
      \cexists \tv \ctxinfer \E \tp \tv
    & \text{otherwise}
    \\
    \cand ~ \cmatch \tv {\cbranch {\cpatrcd \ct} {\labfrom \elab \ct \leq \tv \to \t}}
  \end{cases}}
\\
\Crule
  {\ctxinfer {\parens{\exfield \E \T \elab}} \tp \t}
  {\cexists \tv \ctxinfer \E \tp \tv \cand \labfrom \elab \T \leq \tv \to \t}
\\
\Crule
  {\ctxinfer {\erecord {\elaba = \ea; \ldots; \elab_j = \E_j; \ldots; \elab_n = \en}} \tp \t}
  {\!\begin{cases}
    \ctxinfer {\exrecord \T { \ldots; \elab_j = \E_j; \ldots}} \tp \t
    & \text{if } \labsuni \elabs \T
    \\
    \cexists \tvs \cAnd_{i \neq j} \cinfer \ei \tvi \cand \ctxinfer {\E_j} \tp \tvj
    & \text{otherwise} \\
    \uad\cand~\cmatch \t {\cbranch {\cpatrcd \ct}
      {\\ \hspace{3em}
       \dom {\ct} = \elabs \cand \cAnd\iton \labfrom \elabi \ct \leq \t \to \tvi}}
  \end{cases} }
\\
\Crule
  {\ctxinfer {\exrecord \T {\elaba = \ea; \ldots; \elab_j = \E_j; \ldots; \elab_n = \en}} \tp \t}
  {\cexists \tvs \cAnd_{i \neq j} \cinfer \ei \tvi \cand \ctxinfer {\E_j} \tp \tvj \cand \dom {\T} = \elabs \\
  &&\cand \cAnd\iton \labfrom \elabi \T \leq \t \to \tvi }
\\
\Crule
  {\ctxinfer {\parens{\emagic {\ea, \ldots, \E_j, \ldots, \en}}} \tp \t}
  {\cexists \tvs \cAnd_{i \neq j} \cinfer \ei \tvi \cand \ctxinfer {\E_j} \tp \tvj}
\\\\
\Crule
  {\ctxinfer \E \tp {\tfor \tvs \t}}
  {\cfor \tvs \ctxinfer \E \tp \t}
\\\\
\end{longtable}

}

\section{Properties of the constraint language}
\label{app:proofs-constraints}

This appendix establishes key properties of the constraint language. The first
is the principality of shapes \cref{thm:principal-shapes}: any non-variable type
$\t$ admits a non-trivial principal shape $\sh$.

The second is the canonicalization of satisfiability derivations $\semenv \th
\c$, which enables a simple induction principal for reasoning about unicity.
This canonical form for derivations is a crucial tool in our proof of
soundness and completeness in \cref{app/oml/proofs}.

\subsection{Principality of shapes}

\principalShapesBIS*
\begin{proof}
  Let us assume $\t$ is a non-variable type.

  \begin{proofcases}
    \proofcase{$\t$ is a type constructor $\tconstr \tys$}

    $\tconstr$ is a top-level type constructor of arity $n$, which in
    our setting may be the nullary $\tunit$, the binary arrow,
    the $n$-ary product, or a nominal record type. In all
    these cases, the shape of $\t$ is $\any \tvcs \tconstr \tvcs$
    where $\tvcs$ is a sequence of $n$ distinct type variables. This
    is clearly principal.

    \proofcase{$\t$ is a polytype $\tpoly {\tfor \tvs \t}$}

    We may assume \Wlog that each variable of $\tvs$ occurs free in
    $\t$.
    Let $(\pi_i)\iton$ be the sequence of shortest paths in $\t$ that cannot be
    extended to reach a (polymorphic) variable in $\tvas$, in lexicographic
    order and $\tvcs$ be a sequence $(\tvci)\iton$ of distinct variables that do
    not appear in~$\t$.
    Let $\tyz$ be $\t \where {\pi_i \is \tvci}\iton$, \ie the term $\t$ where each
    path $\pi_i$ has been substituted by the variable $\tvci$.  Let $\Sh$ be the
    shape $\any \tvcs {\tpoly {\all \tvs \tyz}}$.
    We claim that $\Sh$ is actually the principal shape of $\tpoly {\all \tvs
    \t}$.

    \medskip
    \locallabelreset

    By construction, $\t$ is equal to $\shapp[\Sh] \tys$~\llabel 1.
    where $\tys$ is the sequence composed of $\ti$ equal to $\t/\pi_i$
    for $i$ ranging from $1$ to $n$.
    Indeed, by
    definition, $\shapp[\Sh] \tys$ is equal to $(\t\where {\pi_i \is \tvci}\iton)
    \where {\tvci \is \ti}$ which is obviously equal to $\t$.
    The remaining of the proof checks that $\Sh$ is minimal~\llabel 2, that is,
    we assume that $\Sh'$ is another shape such that $\tpoly {\all\tvs\t}$ is
    equal to $\shapp [\Shp] \typs$ for some $\typs$~\llabel H and show that $\Sh
    \preceq \Shp$~\llabel C.

    \medskip

    It follows from~\lref H that
      $\Shp$ must be a polytype shape, \ie of the form $\any \tvcps {\tpoly
      {\all \tvbs \typ}}$ and
      $\tpoly {\all \tvs \t}$ is equal to $\tpoly {\all\tvbs \tp} \where {\tvcps
      \is \typs}$~\llabel{P}.
    \relax
    We may assume \Wlog that $\tvbs$ and $\tvcps$ are disjoint, that
    $\tvcps$ does not contain useless variables, \ie
    that they all appear in $\tp$ and that they actually appear in lexicographic
    order.
    \relax
    Now that never term contains useless variables, \lref P implies that the
    sequences $\tvas$ and $\tvbs$ can be put in one-to-one correspondences.
    Besides, since they all ordered in the order of appearance in terms, they
    the correspondence respects the ordering. Hence, the substitution $\where
    {\tvbs \is \tvas}$ is a renaming. Therefore, we can assume \Wlog that
    $\tvbs$ is $\tvas$,
    \relax
    That is, \lref P becomes that $\tpoly {\all \tvs \t}$ is equal to $\tpoly
    {\all \tvs \typ \where {\tvcps \is \typs}}$, which given that variables
    $\tvs$ appear in the same order in both terms, implies that $\t$ is
    equal to $\typ \where {\tvcps \is
    \typs}$~\llabel T.

    \relax

    \medskip

    Since $\typs$ does not contain any variable in $\tvs$, every path $\pi_i$
    is a path in $\typ$. Thus, we may write $\typ$ as
    \relax $\typ \where {\pi_i \is \tyi''}\iton$ where $\tyi''$ is $\typ/\pi_i$.
    This is also equal to
    \relax $(\typ \where {\pi_i \is \tvci}\iton) \where {\tvci \is \tyi''}\iton$,
    that is $\tyz\where {\tvci \is \tyi''}\iton$.
    In summary, we have $\typ$ is equal to
    \relax $\tyz \where {\tvci \is \tyi''}\iton$,
    which implies that
    \relax  $\tpoly {\all \tvs \typ}$ is equal to
    \relax  $\tpoly {\all \tvs {\tyz \where {\tvci \is \tyi''}\iton}}$, \ie
    \relax  $\tpoly {\all \tvs \tyz} \where {\tvci \is \tyi''}\iton$~\llabel E.
    By \Rule {Inst-Shape}, we have
    \begin{mathpar}[inline]
    \any \tvcs  \tpoly {\all \tvs \tyz} \preceq
    \any \tvcps\tpoly {\all \tvs \tyz} \where {\tvci \is \tyi''}\iton,
    \end{mathpar}
    which, given~\lref E, is exactly~\lref C.

  \end{proofcases}
\end{proof}

\subsection{Canonicalization of satisfiability}

They key result in this section is that our semantic derivations $\semenv \th
\c$ can always be rewritten to only apply the rule \Rule{Match-Ctx} at the very
bottom of the derivation, rather than in the middle of derivations. This
corresponds to explicitating the unique shapes of all suspended constraints (in
some order that respects the dependency between suspended constraints), and
then continuing with a syntax-directed proof of a fully-discharged constraint.

We did not impose this ordering in our definition of the semantics to make it
more flexible and more declarative, but the inversion principle that it
provides will be helpful when reasoning about the solver in
\cref{app:proofs-solving}.

We define in \cref{fig:simple} a formal judgment $\c \simple$ that says that
$\c$ does not contain any suspended match constraint, and extend it trivially
to constraint contexts: $\C \simple$. In particular, the erasure $\cerase \c$
of a constraint (\cref{def:erasure}) is always simple. We then introduce in
\cref{fig:canonical-sem} a ``canonical'' semantic judgment $\semenv \Th \c$
that enforces the structure we mentioned: its derivation starts by discharging
suspended constraints, until eventually we reach a simple constraint $\c$.
Below we prove that any semantic derivation $\semenv \th \c$ can be turned into
a canonical semantic derivation $\semenv \Th \c$.

We can think of this result as controlling the amount of non-syntax-directness in our rules: we need some of it, but it suffices to have it only at the outside, and it contains a more standard derivation that is easy to reason about.

\paragraph{Inversion} When $\c$ is simple, a derivation of $\semenv \th \c$
does not use the contextual rule (it is a derivation in $\semenv \thsimple
\c$), so it enjoys the usual inversion principle on syntax-directed judgments;
for example, if $\semenv \thsimple {\ca \cand \cb}$ then by inversion $\semenv
\thsimple \ca$ and $\semenv \thsimple \cb$, etc.

\paragraph{Congruence} Congruence does not hold in general in our system due to
the contextual rule. For example, $\ca \eqdef (\cmatch \tva {\cbranch \wild
\ctrue})$ is unsatisfiable so we have $\ca \cequiv \cfalse$, but for $\C \eqdef
(\cexists \tva \cunif \tva \tint \cand \hole)$ we have $\C \where \ca \cequiv
\ctrue$ and $\C \where \cfalse \equiv \cfalse$. It holds simply for simple
constraints.

\begin{lemma}[Simple congruence]
  \label{lem:cong-simple}
  Given simple constraints $\ca, \cb$ and simple context $\C$.
  If \\$\ca \centails \cb$, then $\C\where{\ca} \centails \C\where{\cb}$.

  \begin{proof}
    Induction on the derivation of $\C \simple$.
  \end{proof}
\end{lemma}

\paragraph{Composability}

The composability result below is an important test of our definition of the
unicity condition $\Cshape \C \t \sh$, which is in part engineered for this
lemma to be simple to prove. In the past we used a definition of unicity
that also required $\C \where \ctrue$ to be satisfiable, which broke the
composability property.
\begin{lemma}[Composability of unicity]
  \label{lem:compose-unicity}
  If $\Cshape \Ca \t \sh$, then $\Cshape {\Cb\where\Ca} \t \sh$.
  \begin{proof}
    Induction on the structure of $\Cb$.
    \begin{proofcases}
      \proofcase{$\hole$}
        immediate.
      \proofcase{$\Cc \cand \c$}

	\begin{llproof}
	  \shapePf{\Ca}{\t}{\sh}{Premise}
	  \shapePf{\Cc\where\Ca}{\t}{\sh}{By \ih}
	  \ForallPf{\semenv, \gt}{}{Definition of $\Cshape {\parens{\Cc\where\Ca \cand \c}} \t \sh$}
	  \vdashPf{\semenv}{\cerase {\Cc\where\Ca\where{\cunif \t \gt}} \cand \cerase \c}{$\implies$I}
	  \vdashPf{\semenv}{\cerase {\Cc\where\Ca\where{\cunif \t \gt}}}{Simple inversion}
	  \eqPf{\shape \gt}{\sh}{$\implies$E on $\Cshape {\Cc\where\Ca} \t \sh$}
\Hand 	  \shapePf{\parens{\Cc\where\Ca \cand \c}}{\t}{\sh}{Above}
	\end{llproof}

      \proofcase{$\c \cand \Cc$}

	\begin{llproof}
	  Similar to the $\Cc \cand \c$ case.
	\end{llproof}

      \proofcase{$\cexists \tv \Cc$}

	\begin{llproof}
	  \shapePf{\Ca}{\t}{\sh}{Premise}
	  \shapePf{\Cc\where\Ca}{\t}{\sh}{By \ih}
	  \ForallPf{\semenv, \gt}{}{Definition of $\Cshape {\parens{\cexists \tv \Cc\where\Ca}} \t \sh$}
	  \vdashPf{\semenv}{\cexists \tv \cerase {\Cc\where\Ca\where{\cunif \t \gt}}}{$\implies$I}
	  \vdashPf{\semenv\where{\tv \is \gtp}}{\cerase {\Cc\where\Ca\where{\cunif \t \gt}}}{Simple inversion}
	  \eqPf{\shape \gt}{\sh}{$\implies$E on $\Cshape {\Cc\where\Ca} \t \sh$}
\Hand 	  \shapePf{\parens{\cexists \tv \Cc\where\Ca}}{\t}{\sh}{Above}
	\end{llproof}

      \proofcase{$\cfor \tv \Cc$}

	\begin{llproof}
	  Similar to $\cexists \tv \Cc$ case.
	\end{llproof}

      \proofcase{$\cexistsi \inst \x \Cc$}

	\begin{llproof}
	  Similar to $\cexists \tv \Cc$ case.
	\end{llproof}

      \proofcase{$\clet \x \tv \Cc \c$}

	\begin{llproof}
	  \shapePf{\Ca}{\t}{\sh}{Premise}
	  \shapePf{\Cc\where\Ca}{\t}{\sh}{By \ih}
	  \ForallPf{\semenv, \gt}{}{Definition of $\Cshape {\parens {\Let \x \ldots}} \t \sh$}
	  \vdashPf{\semenv}{\clet \x \tv {\cerase {\Cc\where\Ca\where{\cunif \t \gt}}} {\cerase \c}}{$\implies$I}
	  \vdashPf{\semenv}{\cexists \tv \cerase {\Cc\where\Ca\where{\cunif \t \gt}}}{Simple inversion}
	  \vdashPf{\semenv\where{\tv \is \gtp}}{\cerase {\Cc\where\Ca\where{\cunif \t \gt}}}{Simple inversion}
	  \eqPf{\shape \gt}{\sh}{$\implies$E on $\Cshape {\Cc\where\Ca} \t \sh$}
\Hand 	  \shapePf{\parens{\clet \x \tv {\Cc\where\Ca} \c}}{\t}{\sh}{Above}
	\end{llproof}

      \proofcase{$\clet \x \tv \c \Cc$}

	\begin{llproof}
	  Similar to $\clet \x \tv \Cc \c$ case.
	\end{llproof}

      \proofcase{$\cletr \x \tv \tvs \Cc \c$}

	\begin{llproof}
	  Similar to $\clet \x \tv \Cc \c$ case.
	\end{llproof}

      \proofcase{$\cletr \x \tv \tvs \c \Cc$}

	\begin{llproof}
	  Similar to $\clet \x \tv \c \Cc$ case.
	\end{llproof}
    \end{proofcases}
  \end{proof}
\end{lemma}

\begin{lemma}[Inversion of unicity]
  \label{lem:unicity-inversion}~
  \begin{enumerate}[(\roman*)]
    \item If $\Cshape {\parens{\cexists \tv \C}} \t \sh$, then $\Cshape \C \t \sh$.
    \item If $\Cshape {\parens{\cfor \tv \C}} \t \sh$, then $\Cshape \C \t \sh$.
  \end{enumerate}
  \begin{proof}
    The definition of $\Cshape \C \t \sh$ uses simple semantics on the
    erasure $\cerase \C$, so these results are easily shown by simple inversion.
  \end{proof}
\end{lemma}

\begin{lemma}[Decanonicalization]
  \label{lem:decanonicalization}
  If $\semenv \Th \c$, then $\semenv \th \c$.
  \begin{proof}
    Induction on the given derivation $\semenv \Th \c$
  \end{proof}
\end{lemma}

\begin{theorem}[Canonicalization]
  \label{thm:canonicalization}
  If $\semenv \th \c$, then $\semenv \Th \c$.
  \begin{proof}
  We proceed by induction on $\semenv \th \c$ with the measure $\cmeasure \c$.
  \begin{proofcases}
    \proofcasederivation
      {True}
      { }
      {\semenv \th \ctrue}

      \begin{llproof}
\Hand   \VdashPf{\semenv}{\ctrue}{immediate by \Rule{Can-Base}}
      \end{llproof}

    \proofcasederivation
      {Unif}
      {\semenv(\ta) = \semenv(\tb)}
      {\semenv \th \cunif \ta \tb}

      \begin{llproof}
	Similar to the \Rule{True} case.
      \end{llproof}
    \proofcasederivation
      {Conj}
      {\semenv \th \ca \\ \semenv \th \cb}
      {\semenv \th \ca \cand \cb}

      \begin{llproof}
	\vdashPf{\semenv}{\ca} {Premise}
	\vdashPf{\semenv}{\cb} {Premise}
	\VdashPf{\semenv}{\ca} {By \ih}
	\VdashPf{\semenv}{\cb} {By \ih}
	\decolumnizePf
	\casesPf{\semenv \Th \ca, \semenv \Th \cb}
      \end{llproof}

      \begin{proofcases}
	\proofcasederivationdouble
	  {Can-Base}
	  {\semenv \th \ca \\ \ca \simple}
	  {\semenv \Th \ca}
	  {Can-Base}
	  {\semenv \th \cb \\ \cb \simple}
	  {\semenv \Th \cb}

        \begin{llproof}
\Hand     \VdashPf{\semenv}{\ca \cand \cb}{immediate by \Rule{Can-Base}}
        \end{llproof}

	\proofcasederivationdouble
	  {Can-Match-Ctx}
	  {\Cshape \C \t \sh \\ \semenv \Th \C\where{\cmatched \t \sh \cbrs}}
	  {\semenv \Th \underbrace{\C\where{\cmatch \t \cbrs}}_\ca}
	  {}
	  {}
	  {\semenv \Th \cb}

	  \begin{llproof}

	    \VdashPf{\semenv}{\C\where{\cmatched \t \sh \cbrs}} {Premise}
	    \vdashPf{\semenv}{\C\where{\cmatched \t \sh \cbrs}} {\cref{lem:decanonicalization}}
	\decolumnizePf
	    \vdashPf{\semenv}{\C\where{\cmatched \t \sh \cbrs} \cand \cb}{By \Rule{Conj}}
 	    \VdashPf{\semenv}{\C\where{\cmatched \t \sh \cbrs} \cand \cb}{By \ih}
	    \Pf{}{}{\Cshape \C \tv \sh}{Premise}
	    \Pf{}{}{\Cshape {\parens {\C \cand \cb}} \tv \sh}{\cref{lem:compose-unicity}}
\Hand 	    \VdashPf{\semenv}{\C\where{\cmatch \t \cbrs}}{By \Rule{Can-Match-Ctx}}
	  \end{llproof}

	\proofcasederivationdouble
	  {}
	  {}
	  {\semenv \Th \ca}
	  {Can-Match-Ctx}
	  {\Cshape \C \t \sh \\ \semenv \Th \C\where{\cmatched \t \sh \cbrs}}
	  {\semenv \Th \underbrace{\C\where{\cmatch \t \cbrs}}_\cb}

	  \begin{llproof}
	    \Pf{}{}{}{Symmetric to the above case.}
	  \end{llproof}
      \end{proofcases}

      \proofcasederivation
	{Exists}
	{\semenv\where{\tv \is \gt} \th \c}
	{\semenv \th \cexists \tv \c}

	\begin{llproof}
	  \vdashPf{\semenv\where{\tv \is \gt}}{\c}{Premise}
	  \VdashPf{\semenv\where{\tv \is \gt}}{\c}{By \ih}
	  \casesPf{\semenv\where{\tv \is \gt} \Th \c}
	\end{llproof}

	\begin{proofcases}

	    \proofcasederivation
	      {Can-Base}
	      {\semenv\where{\tv \is \gt} \th \c \\ \c \simple}
	      {\semenv\where{\tv \is \gt} \Th \c}

	      \begin{llproof}
\Hand 		\VdashPf{\semenv}{\cexists \tv \c}{Immediate by \Rule{Can-Base}}
	      \end{llproof}

	      \proofcasederivation
		{Can-Match-Ctx}
		{\Cshape \C \t \sh \\ \semenv\where{\tv \is \gt} \Th \C\where{\cmatched \t \sh \cbrs}}
		{\semenv \Th \underbrace{\C\where{\cmatch \t \cbrs}}_\c}

		\begin{llproof}
		  \VdashPf{\semenv\where{\tv \is \gt}}{\C\where{\cmatched \t \sh \cbrs}}{Premise}
		  \vdashPf{\semenv\where{\tv \is \gt}}{\C\where{\cmatched \t \sh \cbrs}}{\cref{lem:decanonicalization}}
	    \decolumnizePf
		  \vdashPf{\semenv}{\cexists \tv \C\where{\cmatched \t \sh \cbrs}}{By \Rule{Exists}}
		  \VdashPf{\semenv}{\cexists \tv \C\where{\cmatched \t \sh \cbrs}}{By \ih}
		  \Pf{}{}{\Cshape \C \t \sh}{Premise}
		  \Pf{}{}{\Cshape {\parens {\cexists \tv \C}} \t \sh}{\cref{lem:compose-unicity}}
\Hand             \VdashPf{\semenv}{\cexists \tv \C\where{\cmatch \t \cbrs}}{By \Rule{Can-Match-Ctx}}
		\end{llproof}
	\end{proofcases}

	\proofcasederivation
	  {Forall}
	  {\forall \gt,~ \semenv\where{\tv \is \gt} \th \c}
	  {\semenv \th \cfor \tv \c}

	  \begin{llproof}
	    Similar to the \Rule{Exists} case.
	  \end{llproof}

	\proofcasederivation
	  {Let}
    {\semenv\where{\x \is \gabs} \th \cb\\
     \gabs = \semenv(\cabs \tv \ca) \\
     \gabs \neq \eset
     }
    {\semenv \th \clet \x \tv \ca \cb}

    \begin{llproof}
      \VdashPf{\semenv\where{\x \is \gabs}}{\cb}{By \ih}
      \casesPf{\semenv\where{\x \is \gabs} \Th \cb}
    \end{llproof}

          \begin{proofcases}
            \proofcasederivation
              {Can-Base}
              {\semenv\where{\x \is \gabs} \th \cb \\ \cb \simple}
              {\semenv\where{\x \is \gabs} \Th \cb}

              \begin{llproof}
\Hand                \VdashPf{\semenv}{\clet \x \tv \ca \cb}{Immediate by \Rule{Can-Base}}
              \end{llproof}

	    \proofcasederivation
		{Can-Match-Ctx}
		{\Cshape \C \t \sh \\ \semenv\where{\x \is \gabs} \Th \C\where{\cmatched \t \sh \cbrs}}
		{\semenv\where{\x \is \gabs} \Th \underbrace{\C\where{\cmatch \t \cbrs}}_\cb}

		\begin{llproof}
		  \shapePf{\C}{\t}{\sh}{Premise}
		  \shapePf{\parens{\clet \x \tv \ca \C}}{\t}{\sh}{\cref{lem:compose-unicity}}
		  \VdashPf{\semenv\where{\x \is \gabs}}{\C\where{\cmatched \t \sh \cbrs}}{Premise}
		  \vdashPf{\semenv\where{\x \is \gabs}}{\C\where{\cmatched \t \sh \cbrs}}{\cref{lem:decanonicalization}}
		  \vdashPf{\semenv}{\clet \x \tv \ca {\C\where{\cmatched \t \sh \cbrs}}}{By \Rule{Let}}
		  \VdashPf{\semenv}{\clet \x \tv \ca {\C\where{\cmatched \t \sh \cbrs}}}{By \ih}
\Hand 		  \VdashPf{\semenv}{\clet \x \tv \ca {\C\where{\cmatch \t \sh}}}{By \Rule{Can-Match-Ctx}}
		\end{llproof}
	  \end{proofcases}

      \proofcasederivation
	{App}
	{\semenv(\t) \in \semenv(\x)}
	{\semenv \th \capp \x \t}

      \begin{llproof}
	Similar to the \Rule{True} case.
      \end{llproof}

      \proofcasederivation
        {Match-Ctx}
        {\Cshape \C \t \sh \\ \semenv \th \C\where{\cmatched \t \sh \cbrs}}
        {\semenv \th \C\where{\cmatch \t \cbrs}}

      \begin{llproof}
        \shapePf{\C}{\t}{\sh}{Premise}
        \VdashPf{\semenv}{\C\where{\cmatched \t \sh \cbrs}}{Premise}
        \vdashPf{\semenv}{\C\where{\cmatched \t \sh \cbrs}}{By \ih}
\Hand        \vdashPf{\semenv}{\C\where{\cmatch \t \cbrs}}{By \Rule{Can-Match-Ctx}}
      \end{llproof}

      \proofcasederivation
	{LetR}
        {\semenv\where{\x \is \gabsr} \th \cb \\
         \gabsr = \semenv(\cabsr \tv \tvs \ca) \\
         \gabsr \neq \emptyset
         }
        {\semenv \th \cletr \x \tv \tvs \ca \cb}

      \begin{llproof}
	Similar to the \Rule{Let} case.
      \end{llproof}

      \proofcasederivation
	{AppR}
        {\greg {\semenv(\t)} \wild \in \semenv(\x)  }
        {\semenv \th \capp \x \t}

      \begin{llproof}
	Similar to the \Rule{App} case.
      \end{llproof}

    \proofcasederivation
      {Exists-Inst}
      {\greg \wild \semenvp \in \semenv(\x) \\\\
       \semenv\where{\inst \is \semenvp} \th \c}
      {\semenv \th \cexistsi \inst \x \c}

      \begin{llproof}
	Similar to the \Rule{Exists} case.
      \end{llproof}

    \proofcasederivation
      {Multi-Unif}
      {\forall \t \in \ueq,~ \semenv(\t) = \gt}
      {\semenv \th \ueq}

      \begin{llproof}
	Similar to the \Rule{Unif} case.
      \end{llproof}

    \proofcasederivation
      {Incr-Inst}
      {\semenv(\inst)(\tv) = \semenv(\t)}
      {\semenv \th \cpinst \inst \tv \t}

      \begin{llproof}
	Similar to the \Rule{App} case.
      \end{llproof}

  \end{proofcases}
  \end{proof}
\end{theorem}

\begin{lemma}[Inversion of suspension]
  \label{lem:susp-inversion}
  If $\semenv \th \C\where{\cmatch \t \cbrs}$ and $\Cshape \C \t \sh$,
  then\\$\semenv \th \C\where{\cmatched \t \sh \cbrs}$.

  \begin{proof}
    We use canonicalization (\cref{thm:canonicalization}) to induct on $\semenv \Th
    \C\where{\cmatch \t \cbrs}$ instead of $\semenv \th \C\where{\cmatch \t
    \cbrs}$.

    \begin{proofcases}
      \proofcasederivation
	{Can-Base}
	{\semenv \th \C\where{\cmatch \t \cbrs} \\ \C\where{\cmatch \t \cbrs} \simple}
	{\semenv \Th \C\where{\cmatch \t \cbrs}}

        The second premise is a contradiction.

      \proofcasederivation
	{Can-Match-Ctx}
	{\Cshape \Cp \tp \shp \\ \semenv \Th \Cp\where{\cmatched \tp \shp \cbrs'}}
	{\semenv \Th \underbrace{\Cp\where{\cmatch \tp \cbrs'}}_{\C\where{\cmatch {~\t~} {~\cbrs}}}}

	\begin{llproof}
	  \casesPf{\C = \Cp}
	\end{llproof}
	\begin{proofcases}
	  \proofcase{$\C = \Cp$}

	    \begin{llproof}
	      \eqPf{\C}{\Cp}{Premise}
	      \eqPf{\tp}{\t}{}
	      \eqPf{\shp}{\sh}{}
	      \eqPf{\cbrs'}{\cbrs}{}
\Hand         \VdashPf{\semenv}{\C\where{\cmatched \t \sh \cbrs}}{Premise}
	    \end{llproof}

	  \proofcase{$\C \neq \Cp$}

	    \newcommand{\Ctwo}{\C_2}
	    \begin{llproof}
	      \Eqpf{\Ctwo\where{\cmatch \t \cbrs, \cmatch \tp \cbrs'}}
                  {\C\where{\cmatch \t \cbrs}}
                  {For some 2-hole context $\Ctwo$}
	      \continueeqPf{\Cp\where{\cmatch \tp \cbrs'}}{}
	      \decolumnizePf
	      \VdashPf{\semenv}{\Ctwo\where{\cmatch \t \cbrs, \cmatched \tp \shp \cbrs'}}{Premise}
	      \decolumnizePf
	      \ForallPf{\semenvp, \gtp}{}{\hspace{30.5ex}Defn. of $\Cshape {\Ctwo\where{\hole, \cmatched \tp \shp \cbrs'}} \t \sh$}
	      \decolumnizePf
	      \vdashPf{\semenvp}{\cerase {\Ctwo\where{\cunif \t \gtp, \cmatched \tp \shp \cbrs'}}}{$\implies$I}
	      \vdashPf{\semenvp}{\cerase {\Ctwo\where{\cunif \t \gtp, \ctrue}}}{\cref{lem:cong-simple}}
	    \decolumnizePf
	      \eqPf{\cerase {\Ctwo\where{\cunif \t \gtp, \ctrue}}}{\cerase {\Ctwo\where{\cunif \t \gtp, \cerase {\cmatch \tp \cbrs'}}}}{By definition}
	      \continueeqPf{\cerase {\C\where{\cunif \t \gtp}}}{By definition}
	      \vdashPf{\semenvp}{\cerase {\C\where{\cunif \t \gtp}}}{Above}
	      \eqPf{\shape \gtp}{\sh}{$\implies$E on $\Cshape \C \t \sh$}
	      \shapePf{\Ctwo\where{\hole, \cmatched \tp \shp \cbrs'}}{\t}{\sh}{Above}
	      \VdashPf{\semenv}{\Ctwo\where{\cmatched \t \sh \cbrs, \cmatched \tp \shp \cbrs'}}{By \ih}
	      \decolumnizePf
	      \ForallPf{\semenvp, \gtp}{}{\hspace{30.5ex}Defn. of $\Cshape {\Ctwo\where{\cmatched \t \sh \cbrs, \hole}} \tp \shp$}
	      \decolumnizePf
	      \vdashPf{\semenvp}{\cerase {\Ctwo\where{\cmatched \t \sh \cbrs, \cunif \tp \gtp}}}{$\implies$I}
	      \vdashPf{\semenvp}{\cerase {\Ctwo\where{\ctrue, \cunif \tp \gtp}}}{\cref{lem:cong-simple}}
	      \eqPf{\cerase {\Ctwo\where{\ctrue, \cunif \tp \gtp}}}{\cerase {\Ctwo\where{\cerase {\cmatch \t \cbrs}, \cunif \tp \gtp}}}{By definition}
	      \continueeqPf{\cerase {\Cp\where{\cunif \tp \gtp}}}{By definition}
	      \vdashPf{\semenvp}{\cerase {\C\where{\cunif \t \gtp}}}{Above}
	      \shapePf{\Cp}{\tp}{\shp}{Premise}
	      \eqPf{\shape \gtp}{\shp}{$\implies$E on $\Cshape \Cp \tp \shp$}
	      \shapePf{\Ctwo\where{\cmatched \t \sh \cbrs, \hole}}{\tp}{\shp}{Above}
\Hand 	      \VdashPf{\semenv}{\Ctwo\where{\cmatched \t \sh \cbrs, \cmatch \tp \cbrs'}}{By \Rule{Can-Match-Ctx}}
	    \end{llproof}
	\end{proofcases}
    \end{proofcases}
  \end{proof}
\end{lemma}

%
%
%
%

\section{Properties of the constraint solver}
\label{app:proofs-solving}

The primary requirement of our constraint solver is correctness:
a constraint $\c$ is satisfiable if and only if the solver terminates with a solution.

This section decomposes this requirement into three properties: preservation,
progress, and termination---and provides proofs for each. Correctness then
follows as a corollary of these results.

\subsection{Preservation}

This section details the proof of \emph{preservation} for the solver: if $\ca
\csolve \cb$, then $\ca \cequiv \cb$.
Since rewriting may occur under arbitrary contexts, it suffices to check for
each rule, that the equivalence $\ca \cequiv \cb$ holds under all contexts
$\C$.

However, the introduction of suspended match constraints breaks congruence of
equivalence. That is, it is no longer the case that $\ca \cequiv \cb$ implies
$\C\where\ca \cequiv \C\where\cb$.
For instance, we have $\cmatch \tv \cbrs \cequiv \cfalse$, yet
$\C\where{\cmatch \tv \cbrs} \cnequiv \C\where\cfalse$ for $\C \is \hole
\cand \cunif \tv \tint$.

As a result, we must prove \emph{contextual equivalence} for each rewriting
rule explicitly. This is both non-trivial and tedious. To simplify the task, we
first present a series of auxiliary lemmas that recover contextual equivalence
for many common cases.
Whenever possible, we prefer to work with equivalences on \emph{simple}
constraints, as these retain the desired congruence properties that do not hold
generally in our system.

\begin{definition}[Contextual eqiuvalence]
  Two constraints $\ca$ and $\cb$ are contextually equivalence, written $\ca \cequivctx \cb$,
  iff:
  \begin{mathpar}
    \ca \cequivctx \cb \uad\eqdef\uad \all \C \uad \C\where\ca \cequiv \C\where\cb
  \end{mathpar}
\end{definition}

\begin{corollary}[Simple equivalence is congruent]
  \label{corollary:cong-simple-equiv}
  Given simple constraints $\ca, \cb$ and simple context $\C$. If
  $\ca \cequiv \cb$, then $\C\where\ca \equiv \C\where\cb$.
  \begin{proof}
    Follows from \cref{lem:cong-simple}.
  \end{proof}
\end{corollary}

\begin{lemma}[Simple equivalence is contextual]
  \label{lem:ctxt-equiv-simple}
  For simple constraints $\ca, \cb$. If $\ca \cequiv \cb$, then $\ca \cequivctx \cb$.
  \begin{proof}
    We proceed by induction on the number of suspended match constraints $n$ in $\C$.

    \begin{proofcases}
      \proofcase{$n$ is 0}
	Follows from \cref{corollary:cong-simple-equiv}.

      \proofcase{$n$ is $k + 1$}

	\begin{proofcases}
	  \proofcase{$\implies$}

	  \newcommand{\Ctwo}{\Cb}
	  \begin{llproof}
	    \vdashPf{\semenv}{\C\where{\ca}}{Premise}
	    \VdashPf{\semenv}{\C\where{\ca}}{\cref{thm:canonicalization}}
	    \shapePf\Cp\t\sh{Inversion of \Rule{Can-Match-Ctx}}
	    \VdashPf{\semenv}{\Cp\where{\cmatched \t \sh \cbrs}}{\ditto}
	    \eqPf{\C\where{\ca}}{\Cp\where{\cmatch \t \cbrs}}{\ditto}
	    \continueeqPf{\Ctwo\where{\cmatch \t \cbrs, \ca}}{For some two-hole context $\Ctwo$}
	    \vdashPf{\semenv}{\Ctwo\where{\cmatched \t \sh \cbrs, \cb}}{By \ih}
	    \ForallPf{\semenvp, \gt}{}{Defn of $\Cshape \Cp \t \sh$}
	    \vdashPf{\semenvp}{\cerase{\Ctwo\where{\t \is \gt, \cb}}}{Premise}
	    \vdashPf{\semenvp}{\cerase{\Ctwo\where{\t \is \gt, \ca}}}{\cref{corollary:cong-simple-equiv}}
	    \vdashPf{\semenvp}{\cerase{\Cp\where{\t \is \gt}}}{Above}
	    \eqPf{\shape \gt}{\sh}{$\implies$E on $\Cshape \Cp \t \sh$}
	    \decolumnizePf
	    \shapePf{\Ctwo\where{\hole, \cb}}{\t}{\sh}{Above}
\Hand	    \vdashPf{\semenv}{\Ctwo\where{\cmatch \t \cbrs, \cb}}{By \Rule{Match-Ctx}}
	  \end{llproof}

	  \proofcase{$\impliedby$}

	  \begin{llproof}
	    Symmetric argument.
	  \end{llproof}
	\end{proofcases}
    \end{proofcases}
  \end{proof}
\end{lemma}

\begin{lemma}[Unification is simple]
  \label{lem:unif-problem-simple}
  For all unification problems $\up$, $\up \simple$.
  \begin{proof}
    By induction on the structure of $\up$.
  \end{proof}
\end{lemma}

\begin{definition}[Context equivalence]
  Two contexts $\Ca$ and $\Cb$ are equivalent with guard $P$, written $\Ca
  \cctxequiv^P \Cb$ iff:
  \begin{mathpar}
    \Ca \cctxequiv^P \Cb \uad\eqdef\uad \all \cs \uad P(\cs) \implies
    \Ca\where\cs \cequivctx \Cb\where\c s
  \end{mathpar}
  \end{definition}

\begin{definition}[Match-closed]
  A predicate $P$ on constraints is \emph{match-closed} if, for all
  constraints $\cs, \cs'$, contexts $\C$, matches $\cmatch \t \cbrs$ and
  shapes $\sh$,
  \begin{mathpar}
    P(\cs, \C\where{\cmatch \t \cbrs}, \cs') \implies
    P(\cs, \C\where{\cmatched \t \sh \cbrs}, \cs')
  \end{mathpar}
\end{definition}

\begin{lemma}[Determines is match-closed]
  \label{lem:determines-is-match-closed}
  $\cdetermines \c \tvbs$ is match-closed. Similarly, $\th \cdetermines \c
  \tvbs$ is matched closed.
  \begin{proof}
    Follows from the definitions of $\cdetermines \c \tvbs$, $\th
    \cdetermines \c \tvbs$, and \cref{lem:cong-simple}.
  \end{proof}
\end{lemma}

\begin{lemma}[Simple context equivalence]
  \label{lem:simple-ctxt-equiv}
  For any two simple contexts $\Ca, \Cb$ and a match-closed guard $P$. If
  the two contexts $\Ca$ and $\Cb$ are equivalent under any simple
  constraints satisfying $P$, then $\Ca \cctxequiv^P \Cb$.
  \begin{proof}
    Let us assume that ($\dagger$) holds:
    \begin{mathpar}
      \all{\C, \cs \simple} P(\cs) \implies \C\where{\Ca\where\cs} \cequiv \C\where{\Cb\where\cs}
    \end{mathpar}

    We proceed by induction on the number of suspended match constraints $n$ with
    the statement $Q(n) \is  \all {\cs, \C} \cnmatches {\C} + \cnmatches \cs = n \implies P(\cs) \implies
    \C\where{\Ca\where\cs} \equiv \C\where{\Cb\where\cs}$.

    \begin{proofcases}
      \proofcase{$n$ is 0}

	\begin{llproof}
	  \simplePf{\C, \cs}{Premise ($n$ is 0)}
	  \Hand	  \equivPf{P(\cs) \implies \C\where\Ca\where\cs}{\C\where\Cb\where\cs}{$\dagger$}
	\end{llproof}

      \proofcase{$n$ is $k + 1$}

	\begin{proofcases}
	  \proofcase{$\implies$}

	  \begin{llproof}
	    \Pf{P(\cs)}{}{}{Premise}
	    \vdashPf{\semenv}{\C\where\Ca\where\cs}{Premise}
	    \VdashPf{\semenv}{\C\where\Ca\where\cs}{\cref{thm:canonicalization}}
	    \VdashPf{\semenv}{\Cp\where{\cmatched \t \sh \cbrs}}{Inversion of \Rule{Can-Match-Ctx}}
	    \shapePf{\Cp}{\t}{\sh}{\ditto}
	    \eqPf{\C\where\Ca\where\cs}{\Cp\where{\cmatch \t \cbrs}}{\ditto}
	    \commentPf{Cases on $\C, \cs$.}{}
	  \end{llproof}

	  \begin{proofcases}
	    \proofcase{$\C$ contains $\Cp$'s hole}

	      \newcommand{\Ctwo}{\Cc}
	      \begin{llproof}
		\eqPf{\C\where\Ca\where\cs}{\Ctwo\where{\cmatch \t \cbrs, \Ca\where\cs}}{For some 2-hole context $\Ctwo$}
		\VdashPf{\semenv}{\Ctwo\where{\cmatched \t \sh \cbrs, \Ca\where\cs}}{}
		\eqPf{k}{\cnmatches {\Ctwo\where{\cmatched \t \sh \cbrs, \Ca\where\cs}}}{}
		\vdashPf{\semenv}{\Ctwo\where{\cmatched \t \sh \cbrs, \Cb\where\cs}}{By \ih}
		\decolumnizePf
		\ForallPf{\semenvp, \gt}{}{}
		\vdashPf{\semenvp}{\cerase{\Ctwo\where{\cunif \t \gt, \Cb\where\cs}}}{Premise}
		\vdashPf{\semenvp}{\cerase{\Ctwo\where{\cunif \t \gt, \Ca\where\cs}}}{$\dagger$}
		\eqPf{\shape \gt}{\sh}{$\implies$E on $\Cshape \Cp \t \sh$}
		\shapePf{\Ctwo\where{\hole, \Cb\where\cs}}{\t}{\sh}{Above}
\Hand		\vdashPf{\semenv}{\Ctwo\where{\cmatch \t \cbrs, \Cb\where\cs}}{By \Rule{Match-Ctx}}
	      \end{llproof}

	    \proofcase{$\ci$ contains $\Cp$'s hole}

	      \begin{llproof}
		Similar argument to the above case, but relies on the match-closure of $P$.
	      \end{llproof}
	  \end{proofcases}

	  \proofcase{$\impliedby$}

	  \begin{llproof}
	    Symmetric argument.
	  \end{llproof}
	\end{proofcases}
    \end{proofcases}
  \end{proof}
\end{lemma}

\begin{lemma}[Simple let equivalence]
  \label{lem:simple-let-equiv}
  Given simple constraints $\ca, \cb$ and a simple context $\C$.
  Suppose that
    \begin{mathpar}
      \forall \semenv, \semenvp, \cs \simple. \uad
	\semenvp(\x) = \semenv(\cabsr \tv \tvs {\C\where\cs}) \implies
	  \semenvp \th \ca \iff \semenvp \th \cb
    \end{mathpar}
  Then, for any context $\Cp$ that does not re-bind $\x$, we have:
    \begin{mathpar}
      \cletr \x \tv \tvs {\C\where{\bar\hole}} {\Cp\where\ca}
	\cctxequiv^P \cletr \x \tv \tvs {\C\where{\bar\hole}} {\Cp\where\cb}
    \end{mathpar}
  for any match-closed guard $P$ on the holes.

  \begin{proof}
    Let us pose $\gabsr = \semenv(\cabsr \tv \tvs {\C\where\cs})$ and let us
    assume ($\dagger$):
    \begin{mathpar}
      \forall \semenv, \semenvp, \cs. \uad
	\semenvp(\x) = \gabsr \implies
	  \semenvp \th \ca \iff \semenvp \th \cb
    \end{mathpar}

    We proceed by induction on the number of suspended match constraints in
    $\Cpp, \Cp, \cs$ with the statement $P(n) \is \all {\Cpp, \Cp, \cs} \cnmatches {\Cpp, \Cp, \cs} = n \implies \Cpp\where{\cletr \x \tv \tvs {\C\where\cs} {\Cp\where{\ca}}}
    \cequiv \Cpp\where{\cletr \x \tv \tvs {\C\where\cs} {\Cp\where\cb}}$.

    \begin{proofcases}
      \proofcase{$n$ is 0}

	Thus $\Cpp, \Cp, \cs$ are simple. It suffices to show the equivalence on the let-constraint directly and use congruence
	of equivalence for simple constraints (\cref{lem:ctxt-equiv-simple}) to establish the result.

	We proceed by induction on the structure of $\Cp$ with the statement ($\ddagger$):
	\begin{mathpar}
	\forall \semenv, \semenvp. \uad
	  \semenvp(\x) = \gabsr \implies
	    \semenvp \th \Cp\where{\ca} \iff \semenvp \th \Cp\where{\cb}
	\end{mathpar}
	This holds due to the compositionality of simple equivalence using $\dagger$ as a base case.

	\begin{proofcases}
	  \proofcase{$\implies$}

	\begin{llproof}
	  \vdashPf{\semenv}{\cletr \x \tv \tvs {\C\where{\cs}} {\Cp\where{\ca}}}{Premise}
	  \vdashPf{\semenv\where{\x \is \gabsr}}{\Cp\where{\ca}}{\ditto}
	  \vdashPf{\semenv\where{\x \is \gabsr}}{\Cp\where{\cb}}{$\ddagger$}
	  \vdashPf{\semenv}{\cletr \x \tv \tvs {\C\where{\cs}} {\Cp\where{\cb}}}{By \Rule{LetR}}
	\end{llproof}
	  \proofcase{$\impliedby$}

	  \begin{llproof}
	    Symmetric argument.
	  \end{llproof}
	\end{proofcases}

      \proofcase{$n$ is $k + 1$}

      \begin{llproof}
	Analogous to the inductive step in \cref{lem:simple-ctxt-equiv}.
      \end{llproof}
    \end{proofcases}
  \end{proof}
\end{lemma}

\newcommand{\disjointPf}[3]{\Pf{#1}{\disjoint}{#2}{#3}}
\begin{lemma}
  \label{lem:unicity-soundness}
  If $\th \Cshape \C \t \sh$, then $\Cshape \C \t \sh$.

  \begin{proof}
    \begin{proofcases}
      \proofcasederivation
      {S-Uni-Type}
      {{\t \notin \TyVars}}
      {\th \Cshape \C \t {~\shape \t}}

      \begin{llproof}
        \notinPf{\t}{\TyVars}{Premise}
        \decolumnizePf
        \eqPf{\t}{\shapp[\shape \t] \tys}{For some $\tys$}
        \ForallPf{\semenv, \gt}{}{Defn. of $\Cshape \C \t {~\shape \t}$}
        \vdashPf{\semenv}{\cerase{\C\where{\cunif \t \gt}}}{Premise}
        \vdashPf{\semenva}{\cunif \t \gt}{Inversion of $\cerase \Ca$}
        \eqPf{\gt}{\semenva(\t)}{Simple inversion}
        \continueeqPf{\shapp[\shape \t] {\semenva(\tys)}}{\ditto}
\Hand   \eqPf{\shape \gt}{\shape \t}{Applying shape to both sides}
      \end{llproof}

      \proofcasederivation
        {S-Uni-Var}
        {\tv \disjoint \bvs \Cb}
        {\th \Cshape {\Ca\where{\cunif \tv {\cunif \t \ueq} \cand \Cb\where{-}}} \tv {~\shape \t}}

      \begin{llproof}
        \disjointPf{\tv}{\bvs \Cb}{Premise}
        \eqPf{\t}{\shapp[\shape \t] \tys}{For some $\tys$}
        \ForallPf{\semenv, \gt}{}{Defn. of $\Cshape \C \tv {~\shape \t}$}
        \vdashPf{\semenv}{\cerase{\Ca\where{\cunif \tv {\cunif {\shapp[\shape \t] \tys} \ueq} \cand \Cb\where{\cunif \tv \gt}}}}{Premise}
        \vdashPf{\semenva}{\cunif \tv {\cunif {\shapp[\shape \t] \tys} \ueq}}{Inversion of $\cerase \Ca$}
        \vdashPf{\semenvb}{\cunif \tv \gt}{Inversion of $\cerase \Cb$}
        \eqPf{\gt}{\semenvb(\tv)}{Simple inversion}
        \continueeqPf{\semenva(\tv)}{$\tv \disjoint \bvs \Cb$}
        \continueeqPf{\shapp[\shape \t] {\semenva(\tys)}}{Simple inversion}
  \Hand \eqPf{\shape \gt}{\shape \t}{Applying shape to both sides}
      \end{llproof}

      \proofcasederivation
        {S-Uni-BackProp}
        {\th \Cshape{\cletr \x \tv \tvs {\Ca\where{\ctrue}} {\Cb\where{\cpapp \x \tvp \tvc \inst \cand -}}} \tvc \sh \\
         \tvp \in \tv, \tvs \\
         \x \disjoint \bvs \Cb \\
         \tvp \disjoint \bvs \Ca}
        {\th \Cshape{\cletr \x \tv \tvs {\Ca\where{-}} {\Cb\where{\cpapp \x \tvp \tvc \inst}}} \tvp \sh}

        \begin{llproof}
          \inPf{\tvp}{\tv,\tvs}{Premise}
          \disjointPf{\x}{\bvs \Cb}{\ditto}
          \disjointPf{\tvp}{\bvs \Ca}{\ditto}
          \vdashPf{}{\Cshape{\cletr \x \tv \tvs {\Ca\where{\ctrue}} {\Cb\where{\cpapp \x \tvp \tvc \inst \cand -}}} \tvc \sh}{\ditto}
          \shapePf{\cletr \x \tv \tvs {\Ca\where{\ctrue}} {\Cb\where{\cpapp \x \tvp \tvc \inst \cand -}}}{\tvc}{\sh}{By \ih}
          \decolumnizePf
          \ForallPf{\semenv, \gt}{}{Defn. of $\Cshape \ldots \tv {~\shape \t}$}
          \vdashPf{\semenv}{\cerase{\cletr \x \tv \tvs {\Ca\where{\cunif \tvp \gt}} {\Cb\where{\cpapp \x \tvp \tvc \inst}}}}{Premise}
          \LetPf{\gabsr}{\semenv(\cabsr \tv \tvs {\cerase {\Ca\where{\cunif \tvp \gt}}})}{}
          \LetPf{\semenva}{\semenv[\x \is \gabsr]}{}
          \eqPf{\semenvp(\tvp)}{\gt}{For any $\greg \wild \semenvp \in \semenva(\x)$}
          \vdashPf{\semenvb}{\cpapp \x \tvp \tvc \inst}{Inversion of $\cerase \Cb$}
          \eqPf{\semenvb(\inst^\x)(\tvp)}{\semenvb(\tvc)}{Simple inversion}
          \inPf{\semenvb(\inst^\x)}{\semenvb(\x)}{Since $\cexistsi \inst \x \in \Cb$, $\semenvb$ extends $\semenva$}
          \eqPf{\semenvb(\inst^\x)(\tvp)}{\gt}{Above}
          \continueeqPf{\semenvb(\tvc)}{\ditto}
          \decolumnizePf
          \vdashPf{\semenva}{\cerase {\Cb\where{\cpapp \x \tvp \tvc \inst \cand \cunif \tvc \gt}}}{Entailment for $\cerase \Cb$}
          \vdashPf{\semenv}{\cerase {\cletr \x \tv \tvs {\Ca\where{\cunif \tvp \gt}} {\Cb\where{\cpapp \x \tvp \tvc \inst \cand \cunif \tvc \gt}}}}{By \Rule{LetR}}
          \vdashPf{\semenv}{\cletr \x \tv \tvs {\Ca\where{\ctrue}} {\Cb\where{\cpapp \x \tvp \tvc \inst \cand \cunif \tvc \gt}}}{Simple congruence}
\Hand     \eqPf{\shape \gt}{\sh}{$\implies E$ on $\Cshape \ldots \tvc \sh$}
        \end{llproof}
    \end{proofcases}
  \end{proof}
\end{lemma}

\begin{lemma}
  \label{lem:unicity-completeness}

  If $\C$ is normalized, then $\Cshape \C \t \sh$ if and only $\th \Cshape \C \t \sh$.
  \begin{proof}~
    \begin{proofcases}

      \proofcase{$\implies$}
        \newcommand{\NC}{R}

        Let us assume $\Cshape \C \t \sh$ and $\C$ is normalized.

        Given $\C$ is normalized, every constraint in $\C$ is of the form: \\

        \begin{bnfgrammar}
          \entry[]{\NC}{
            \bar{\hat\ueq}
            \cand \overline{\cmatch \tv \cbrs}
            \cand \cexists {\overline{\inst^\x}} {\overline{\cpapp \x \tvb \tvc \inst}}
            \cand \overline{\cletr \x \tvd \tvds {\NC_1} {\NC_2}}
          }
        \end{bnfgrammar}
        \\

        By assumptions, we have $\all {\phi, \gt} {\phi \th \cerase
        \C\where{\cunif \tv \gt}} \implies \shape \gt = \sh$. Hence $\cerase
        \C$ contains $\cerase \NC$ where: \\

        \begin{bnfgrammar}
          \entry[]{\cerase \NC}{
            \bar{\hat\ueq}
            \cand \cexists {\overline{\inst^\x}} {\overline{\cpapp \x \tvb \tvc \inst}}
            \cand \overline{\cletr \x \tvd \tvds {\cerase {\NC_1}} {\cerase {\NC_2}}}
          }
        \end{bnfgrammar}
        \\

        \Wlog all constraints that may determine the shape of $\tv$ are located
        with the regional binder (following the \Rule{S-Exists-Lower} and
        \Rule{S-Let-ConjLeft} rules). There are two cases:

        \begin{proofcases}

          \proofcase{$\cunif \tv {\cunif \t \ueq} \in \bar{\hat\ueq}$} Apply \Rule{S-Uni-Var}.

          \proofcase{Otherwise}

            Since $\C$ is normalized, it must be that case that no equality
            constraint determines the shape of $\tv$. Since any such equality
            would normalize to $\cunif \tv {\cunif \t \ueq}$, contradicting our
            assumption that $\C$ is normalized.

            By elimination on the structure of $\NC$, the only constraints that
            could determine the shape of $\tv$ are incremental instantiation
            constraints that copy $\tv$. So there exists a partial
            instantiation constraint $\cpapp \x \tv \tvc \inst$ such that $\Cp
            \where{\cpapp \x \tv \tvc \inst} = \C\where{\ctrue}$ and $\Cshape
            \Cp \tvc \sh$.

            By induction, we have $\th \Cshape \Cp \tvc \sh$. From
            \Rule{S-Uni-BackProp}, we have $\th \Cshape \C \tv \sh$.

        \end{proofcases}

      \proofcase{$\impliedby$} Follows from \cref{lem:unicity-soundness}.

    \end{proofcases}

  \end{proof}
\end{lemma}

\begin{lemma}[Unification preservation]
  \label{lem:unification-preservation}
  If $\upa \unif \upb$, then $\upa \equiv \upb$
  \begin{proof}
    By induction on the given derivation $\upa \unif \upb$.
    See \citet*{\BBemlti} for more details.
  \end{proof}
\end{lemma}

\preservationBIS*
\newcommand{\unifPf}[3]{\Pf{#1}{\unif}{#2}{#3}}
\newcommand{\equivctxPf}[3]{\Pf{#1}{\cequivctx}{#2}{#3}}
\newcommand{\ctxequivPf}[3]{\Pf{#1}{\cctxequiv}{#2}{#3}}
\begin{proof}
  We proceed by induction on the given derivation.
  It suffices to show that for each individual rule $R$ ($\ca \csolve_R \cb$),
  that $\ca \cequivctx \cb$.

  \begin{proofcases}
    \proofcaserewrite
      {S-Unif}
      {\upa \\ \upa \unif \upb}
      {\upb}

	\begin{llproof}
	  \unifPf{\upa}{\upb}{Premise}
	  \equivPf{\upa}{\upb}{\cref{lem:unification-preservation}}
	  \simplePf{\upa, \upb}{\cref{lem:unif-problem-simple}}
\Hand 	  \equivctxPf{\upa}{\upb}{\cref{lem:ctxt-equiv-simple}}
	\end{llproof}

    \proofcaserewrite
      {S-Exists-Conj}
      {\parens {\cexists \tv \ca} \cand \cb \\ \tv \disjoint \cb}
      {\cexists \tv \ca \cand \cb}

	\begin{llproof}
	  \disjointPf{\tv}{\cb}{Premise}
	  \decolumnizePf
	  \sufficientPf{equivalence for simple constraints}{}{\cref{lem:simple-ctxt-equiv}}
	  \supposePf{\ca, \cb \simple} {Premise}
	\end{llproof}
	\begin{proofcases}
	  \proofcase{$\implies$}

	  \begin{llproof}
	    \ForallPf{\semenv}{}{}
	    \vdashPf{\semenv}{\parens {\cexists \tv \ca} \cand \cb}{Premise}
	    \vdashPf{\semenv\where{\tv \is \gt}}{\ca}{Simple inversion}
	    \vdashPf{\semenv}{\cb}{Simple inversion}
	    \vdashPf{\semenv\where{\tv \is \gt}}{\cb}{$\tv \disjoint \cb$}
	    \vdashPf{\semenv\where{\tv \is \gt}}{\ca \cand \cb}{By \Rule{Conj}}
\Hand       \vdashPf{\semenv}{\cexists \tv \ca \cand \cb}{By \Rule{Exists}}
	  \end{llproof}
	  \proofcase{$\impliedby$}

	    \begin{llproof}
	    Symmetric argument.
	    \end{llproof}
	\end{proofcases}

      \proofcase{\Rule{S-Let}, \Rule{S-True}, \Rule{S-False},
      \Rule{S-Let-ExistsRight},
      \Rule{S-Let-ConjLeft},
      \Rule{S-Let-ConjRight}, \Rule{S-Inst-Name}, \\\Rule{S-Exists-Exists-Inst},
      \Rule{S-Exists-Inst-Conj}, \Rule{S-Exists-Inst-Let}, \Rule{S-Exists-Inst-Solve},
      \Rule{S-All-Conj}}

      \begin{llproof}
	Similar argument to the \Rule{S-Exists-Conj} case.
      \end{llproof}

    \proofcaserewrite
      {S-Let-ExistsLeft}
      {\cletr \x \tv \tvs {\cexists \tvb \ca} \cb \\ {
       \tvb \disjoint \tv, \tvs, \cb}}
      {\cletr \x \tv {\tvs, \tvb} \ca \cb}


      If we have
      $\semenv \th \cletr \x \tv \tvs {\cexists \tvb \ca} \cb$,
      then for each instance $(\inst^\x \is \semenv'_\inst) \in \semenv$
      we have $\semenv'_\inst \th \cexists \tvb \ca$.
      By simple inversion, we have $\semenv'_\inst \th \cexists \tvb \ca$
      if and only if $\semenv'_\inst \where {\tvb \is \gt_\inst} \th \ca$
      for some $\gt_\inst$, and so our derivation
      $$\semenv \where {\overline{\inst^x \is \semenv'_\inst}}
       \th \cletr \x \tv \tvs {\cexists \tvb \ca} \cb$$
      corresponds to exactly one derivation
      $$\semenv \where {\overline{\inst^x \is {\semenv'_\inst \where {\tvb \is \gt_\inst}}}}
       \th \cletr \x \tv {\tvs,\tvb} {\ca} \cb$$

    \proofcaserewrite
      {S-Match-Ctx}
      {\C\where{\cmatch \t \cbrs} \\ \th \Cshape \C \t \sh}
      {\C\where{\cmatched \t \sh \cbrs}}

	\begin{llproof}
    \vdashPf{}{\Cshape \C \t \sh}{Premise}
	  \shapePf{\C}{\t}{\sh}{\cref{lem:unicity-soundness}}
	  \sufficientPf{equivalences between constraints}{}{\cref{lem:compose-unicity}}
	\end{llproof}

	\begin{proofcases}
	  \proofcase{$\implies$}

	  \begin{llproof}
	    \ForallPf{\semenv}{}{}
	    \vdashPf{\semenv}{\C\where{\cmatch \tv \cbrs}}{Premise}
\Hand 	    \vdashPf{\semenv}{\C\where{\cmatched \tv {\shape \t} \cbrs}}{\cref{lem:susp-inversion}}
	  \end{llproof}

	  \proofcase{$\impliedby$}

	  \begin{llproof}
	    \ForallPf{\semenv}{}{}
	    \vdashPf{\semenv}{\C\where{\cmatched \tv {\shape \t} \cbrs}}{Premise}
\Hand 	    \vdashPf{\semenv}{\C\where{\cmatch \tv \cbrs}}{By \Rule{Match-Ctx}}
	  \end{llproof}
	\end{proofcases}

    \proofcaserewrite
      {S-Let-AppR}
      {\cletr \x \tv \tvs \c {\C\where{\capp \x \t}} \\ \tvc \disjoint \t \\ \x \disjoint \bvs \C}
      {\cletr \x \tv \tvs \c {\C\where{\cexistsi {\tvc, \inst} \x \cunif \tvc \t \cand \cpinst \inst \tv \tvc }}}

	\begin{llproof}
	  \disjointPf{\tvc}{\t}{Premise}
	  \disjointPf{\x}{\bvs \C}{Premise}
	  \decolumnizePf
	  \sufficientPf{a simple equivalence between}{\capp \x \t \text{ and } \cexistsi {\tvc, \inst} \x \cunif \tvc \t \cand \cpinst \inst \tv \tvc}{\cref{lem:simple-let-equiv}}
	  \supposePf{\semenvp(\x) = \semenv(\cabsr \tv \tvs \c)}{Premise}
	\end{llproof}
	\begin{proofcases}
	  \proofcase{$\implies$}

	    \begin{llproof}
	      \vdashPf{\semenvp}{\capp \x \t}{Premise}
	      \inPf{\greg {\semenvp(\t)} {\semenva}}{\semenvp(\x)}{Simple inversion}
              \eqPf{\semenvp(\t)}{\semenva(\tv)}{Defn. of $\semenv(\cabsr \tv \tvs \c)$}
	      \vdashPf{\semenvp\where{\tvc \is \semenvp(\t), \inst \is \semenva}}{\cpinst \inst \tv \tvc}{By \Rule{Incr-Inst}}
	      \vdashPf{\semenvp\where{\tvc \is \semenvp(\t), \inst \is \semenva}}{\cunif \tvc \t}{By \Rule{Unif}}
\Hand	      \vdashPf{\semenvp}{\cexistsi {\tvc, \inst} \x {\cunif \tvc \t \cand \cpinst \inst \tv \tvc}}{By \Rule{Exists}, \Rule{Exists-Inst} and \Rule{Conj}}
	    \end{llproof}

	  \proofcase{$\impliedby$}

	    \begin{llproof}
	      Symmetric argument.
	    \end{llproof}
	\end{proofcases}

    \proofcaserewrite
      {S-Inst-Copy}
      {\cletr \x \tv \tvs {\c}
	\C\where{\cpapp \x \tvb \tvc \inst}\\
	\neg \cyclic {\c} \\
  \x \disjoint \bvs \C \\
	\c = \cp \cand \cunif \tvb {\cunif {\shapp \tvbs} \ueq}\\
	\tvb \in \reg \tv \tvs \\
	\tvbs' \disjoint \tvb, \tvc, \tvbs}
      {\cletr \x \tv \tvs {\c}
  \C\where{\cexists {\tvbs'} \cunif \tvc {\shapp \tvbs'} \cand \cpapp \x {\tvbs} {\tvbs'} \inst}}

	\begin{llproof}
	  \disjointPf{\x}{\bvs \C}{Premise}
	  \disjointPf{\tvbs'}{\tvb, \tvc, \tvbs}{Premise}
	  \decolumnizePf
	  \Sufficientpf{equivalence between}{\cpapp \x \tvb \tvc \inst
          \text{ and }
          \hfil \penalty -50{}\hfill
           \cexists {\tvbs'} \cunif \tvc {\shapp \tvbs'} \cand
          \cpapp \x {\tvbs} {\tvbs'} \inst}{\cref{lem:simple-let-equiv}}
	  \supposePf{\semenvp(\x) = \semenv(\cabsr \tv {\tvs} \c)}{Premise}
	\end{llproof}
	\begin{proofcases}
	  \proofcase{$\implies$}

	    \begin{llproof}
	      \vdashPf{\semenvp} {\cpapp \x \tvb \tvc \inst}{Premise}
	      \inPf{\greg \gt \semenva}{\semenvp(\x)}{$\cexistsi \inst \x \in \C$}
	      \eqPf{\semenvp(\inst)}{\semenva}{\ditto}
	      \eqPf{\semenvp(\tvc)}{\semenv(\inst)(\tvb)}{Simple inversion}
	      \continueeqPf{\semenva(\tvb)}{Above}
	      \vdashPf{\semenva}{\cp \cand \cunif \tvb {\cunif {\shapp \tvbs} \ueq}}{Above}
	      \vdashPf{\semenva}{\cunif \tvb {\cunif {\shapp \tvbs} \ueq}}{Simple inversion}
	      \eqPf{\semenva(\tvb)}{\shapp {\semenva(\tvbs)}}{\ditto}
	      \eqPf{\semenvp(\tvc)}{\shapp {\semenva(\tvbs)}}{Above}
	      \vdashPf{\semenvp\where{\tvbs' \is \semenva(\tvbs)}}{\cunif \tvc {\shapp {\tvbs'}}}{By \Rule{Unif}}
	      \vdashPf{\semenvp\where{\tvbs' \is \semenva(\tvbs)}}{\cpapp \x \tvbs {\tvbs'} \inst}{By \Rule{Incr-Inst}}
\Hand 	      \vdashPf{\semenvp}{\cexists {\tvbs'} \cunif \tvc {\shapp \tvbs'} \cand \cpapp \x {\tvbs} {\tvbs'} \inst}{By \Rule{Exists} and \Rule{Conj}}
	    \end{llproof}

	  \proofcase{$\impliedby$}

	    \begin{llproof}
	      Symmetric argument.
	    \end{llproof}
	\end{proofcases}

    \proofcaserewrite
      {S-Inst-Unify}
      {\cpinst \inst \tvb \tvca \cand \cpinst \inst \tvb \tvcb}
      {\cpinst \inst \tvb \tvca \cand \cunif \tvca \tvcb}

      \begin{llproof}

	\Sufficientpf{equivalence between}
        {\cpinst \inst \tvb \tvca \cand \cpinst \inst \tvb \tvcb
         \text{ and }
         \hfil \penalty -50{}\hfill
         \cpinst \inst \tvb \tvca \cand \cunif \tvca \tvcb}
        {\cref{lem:simple-ctxt-equiv}}

      \end{llproof}

      \begin{proofcases}
	\proofcase{$\implies$}

	  \begin{llproof}
	    \vdashPf{\semenv}{\cpinst \inst \tvb \tvca \cand \cpinst \inst \tvb \tvcb}{Premise}
	    \vdashPf{\semenv}{\cpinst \inst \tvb \tvca}{Simple inversion}
	    \vdashPf{\semenv}{\cpinst \inst \tvb \tvcb}{\ditto}
	    \eqPf{\semenv(\tvca)}{\semenv(\inst)(\tvb)}{\ditto}
	    \eqPf{\semenv(\tvcb)}{\semenv(\inst)(\tvb)}{\ditto}
	    \eqPf{\semenv(\tvca)}{\semenv(\tvcb)}{Above}
	    \vdashPf{\semenv}{\cunif \tvca \tvcb}{By \Rule{Unif}}
\Hand	    \vdashPf{\semenv}{\cpinst \inst \tvb \tvca \cand \cunif \tvca \tvcb}{By \Rule{Conj}}
	  \end{llproof}

	\proofcase{$\impliedby$}

	  \begin{llproof}
	    Symmetric argument.
	  \end{llproof}
      \end{proofcases}

    \proofcaserewrite
      {S-Inst-Poly}
      {\cletr \x \tv {\tvs} {\ueqs \cand \c} {\C\where{\cpapp \x \tvb \tvc \inst}} \\
       \cfor \tvb \cexists {\tvbs} {\ueqs} \cequiv \ctrue \\\\
       \tvb, \tvbs \subseteq \tv, \tvs \\
       \tvb, \tvbs \disjoint \c \\
       \inst(\tvb) \disjoint \insts \C \\
       \x \disjoint \bvs \C}
      {\cletr \x \tv {\tvs} {\ueqs \cand \c} {\C\where\ctrue}}

	\begin{llproof}
	  \equivPf{\cfor \tvb \cexists {\tvbs} \ueqs}{\ctrue}{Premise}
	  \disjointPf{\tvb, \tvbs}{\c}{Premise}
    \disjointPf{\inst(\tvb)}{\insts \C}{Premise}
	  \disjointPf{\x}{\bvs \C}{Premise}
	  \decolumnizePf
	  \sufficientPf{equivalence between}
	    {\cpapp \x \tvb \tvc \inst \text{ and } \ctrue}
	    {\cref{lem:simple-let-equiv}}
	  \supposePf{\semenvp(\x) = \semenv(\cabsr \tv {\tvs} \ueqs \cand \c)}{Premise}
	\end{llproof}
	\begin{proofcases}
	  \proofcase{$\implies$}

	    \begin{llproof}
	      \vdashPf{\semenvp}{\cpapp \x \tvb \tvc \inst}{Premise}
\Hand	      \vdashPf{\semenvp}{\ctrue}{By \Rule{True}}
	    \end{llproof}

	  \proofcase{$\impliedby$}

	    \begin{llproof}
	      \vdashPf{\semenvp}{\ctrue}{Premise}
	      \inPf{\greg \gt {\semenva}}{\semenvp(\x)}{$\C = \Ca\where{\cexistsi \inst \x \Cb}$}
	      \eqPf{\semenvp(\inst)}{\semenva}{\ditto}
	      \casesPf{\semenva(\tvb)}
	    \end{llproof}
	    \begin{proofcases}
	      \proofcase{$\semenva(\tvb) = \semenvp(\tvc)$}

		\begin{llproof}
		  \eqPf{\semenva(\tvb)}{\semenvp(\tvc)}{Premise}
\Hand		  \vdashPf{\semenvp}{\cpapp \x \tvb \tvc \inst}{By \Rule{Incr-Inst}}
		\end{llproof}

	      \proofcase{$\semenva(\tvb) \neq \semenvp(\tvc)$}

		\begin{llproof}
		  \LetPf{\semenvbp}{\semenva\where{\tvb \is \semenvp(\tvc)}}{}
      \vdashPf{\semenvbp}{\cexists {\tvbs} \ueqs}{$\tfor \tvb {\cexists {\tvbs} \ueqs} \cequiv \ctrue$}
		  \LetPf{\semenvb}{\semenva\where{\tvb \is \semenvp(\tvc), \tvbs \is \gts}}{}
      \vdashPf{\semenvb}{\ueqs}{Simple inversion}
		  \vdashPf{\semenva}{\ueqs \cand \c}{By definition}
		  \vdashPf{\semenva}{\c}{Simple inversion}
		  \vdashPf{\semenvb}{\c}{$\tvb, \tvbs \disjoint \c$}
		  \vdashPf{\semenvb}{\ueqs \cand \c}{By \Rule{Conj}}
		  \inPf{\greg \gtp \semenvb}{\semenv(\x)}{By definition}
		  \supposePf{\semenvc \th \Cb\where\ctrue}{Considering entailment on $\cexistsi \inst \x$}
		  \eqPf{\semenvc(\inst)}{\semenva}{\ditto}
      \vdashPf{\semenvc\where{\inst \is \semenvb}}{\Cb\where\ctrue}{$\inst(\tvb) \disjoint \insts \Cb$}
		  \vdashPf{\deriv :: \semenvc}{\Cb\where\ctrue}{By \Rule{Exists-Inst}}
		  \commentPf{$\deriv$ is a derivation that satisfies $\semenva(\tvb) = \semenvp(\tvc)$.}{}
\Hand		  \commentPf{So this case degenerates to the former case.}{}
		\end{llproof}
	    \end{proofcases}
	\end{proofcases}

    \proofcaserewrite
      {S-Inst-Mono}
      {\cletr \x \tv \tvs \c {\C\where{\cpapp \x \tvb \tvc \inst}} \\
       \tvb \notin \reg \tv \tvs \\
       \x, \tvb \disjoint \bvs \C}
      {\cletr \x \tv \tvs \c {\C\where{\cunif \tvb \tvc}}}

	\begin{llproof}
	  \disjointPf{\tvb}{\tv, \tvs}{Premise}
	  \disjointPf{\x, \tvb}{\bvs \C}{Premise}
	  \decolumnizePf
	  \sufficientPf{equivalence between}
	    {\cpapp \x \tvb \tvc \inst \text{ and } \cunif \tvb \tvc}
	    {\cref{lem:simple-let-equiv}}
	  \supposePf{\semenvp(\x) = \semenv(\cabsr \tv {\tvs} \c)}{Premise}
	\end{llproof}

	\begin{proofcases}
	  \proofcase{$\implies$}

	  \begin{llproof}
	    \vdashPf{\semenvp}{\cpapp \x \tvb \tvc \inst}{Premise}
	    \inPf{\greg \wild \semenva}{\semenv(\c)}{$\cexistsi \inst \x \in \C$}
	    \eqPf{\semenvp(\inst)}{\semenva}{\ditto}
	    \eqPf{\semenvp(\tvc)}{\semenva(\tvb)}{Simple inversion}
	    \eqPf{\semenva(\tvb)}{\semenv(\tvb)}{$\tvb \disjoint \tv, \tvs$}
	    \eqPf{\semenvp(\tvb)}{\semenv(\tvb)}{$\tvb \disjoint \bvs \C$}
	    \eqPf{\semenvp(\tvc)}{\semenvp(\tvb)}{Above}
	    \vdashPf{\semenvp}{\cunif \tvc \tvb}{By \Rule{Unif}}
	  \end{llproof}

	  \proofcase{$\impliedby$}

	  \begin{llproof}
	    Symmetric argument.
	  \end{llproof}
	\end{proofcases}

    \proofcaserewrite
      {S-Let-Solve}
      {\cletr \x \tv \tvs \ueqs \c \\ \x \disjoint \c \\
       \cexists {\tv, \tvs} \ueqs \cequiv \ctrue}
      {\c}
	\begin{llproof}
	  \disjointPf{\x}{\c}{Premise}
	  \equivPf{\cexists {\tv, \tvs} \ueqs}{\ctrue}{}
	  \decolumnizePf
	  \sufficientPf{equivalence for simple constraints}{}{\cref{lem:simple-ctxt-equiv}}
	  \supposePf{\c \simple} {Premise}
	\end{llproof}
	\begin{proofcases}
	  \proofcase{$\implies$}

	  \begin{llproof}
	    \ForallPf{\semenv}{}{}
	    \vdashPf{\semenv}{\cletr \x \tv \tvs \ueqs \c}{Premise}
	    \vdashPf{\semenv}{\cexists {\tv, \tvs} \ueqs}{Simple inversion}
	    \vdashPf{\semenv\where{\x \is \semenv(\cabsr \tv \tvs \ueqs)}}{\c}{\ditto}
\Hand	    \vdashPf{\semenv}{\c}{$\x \disjoint \c$}
	  \end{llproof}
	  \proofcase{$\impliedby$}

	    \begin{llproof}
	    \ForallPf{\semenv}{}{}
	    \vdashPf{\semenv}{\c}{Premise}
	    \vdashPf{\semenv\where{\x \is \semenv(\cabsr \tv \tvs \ueqs)}}{\c}{$\x \disjoint \c$}
	    \vdashPf{\semenv}{\cexists {\tv, \tvs} \ueqs}{}
\Hand	    \vdashPf{\semenv}{\cletr \x \tv \tvs \ueqs \c}{By \Rule{LetR}}
	    \end{llproof}
	\end{proofcases}

  \proofcaserewrite{S-Exists-Lower}
    {\cletr \x \tv {\tvas, \tvbs} \ca \cb \\
     \cdetermines {\cexists {\tv, \tvas} \ca} \tvbs \\
     }
    {\cexists \tvbs \cletr \x \tv \tvas \ca \cb}

    \begin{llproof}
      \Pf{}{}{\cdetermines {\cexists {\tv, \tvas} \ca} \tvbs}{Premise}
      \sufficientPf{equivalence for simple constraints}{}{\cref{lem:simple-ctxt-equiv} and \cref{lem:determines-is-match-closed}}
      \supposePf{\ca, \cb \simple}{Premise}
    \end{llproof}
    \begin{proofcases}
      \proofcase{$\implies$}

      \begin{llproof}
	\vdashPf{\semenv}{\cletr \x \tv {\tvas, \tvbs} \ca \cb}{Premise}
	\vdashPf{\semenv}{\cexists {\tv, \tvas, \tvbs} \ca}{Simple inversion}
	\vdashPf{\semenv\where{\x \is \semenv(\cabsr \tv {\tvs, \tvbs} \ca)}}{\cb}{\ditto}
	\vdashPf{\semenv\where{\tv \is \gt, \tvas \is \gts, \tvbs \is \bar\gtp}}{\ca}{\ditto}
	\vdashPf{\semenv\where{\tvbs \is \bar\gtp}}{\cexists {\tv, \tvas} \ca}{By \Rule{Exists}}
	\sufficientPf{}{\semenv\where{\x \is \semenv(\cabsr \tv {\tvs, \tvbs} \ca)} = \semenv\where{\tvbs \is \bar\gt'}(\cabsr \tv \tvs \ca)}{}
      \end{llproof}
      \begin{proofcases}

      \proofcase{$\implies$}

	\begin{llproof}
	  \vdashPf{\semenv\where{\tv \is \gta, \tvs \is \bar\gta, \tvbs \is \bar\gtb}}{\ca}{Premise}
	  \vdashPf{\semenv\where{\tvbs \is \bar\gtb}}{\cexists {\tv, \tvs} \ca}{By \Rule{Exists}}
	  \eqPf{\bar\gtb}{\bar\gtp}{By definition of determines}
\Hand	  \vdashPf{\semenv\where{\tvbs \is \bar\gtp, \tv \is \gta, \tvs \is \bar\gta}}{\ca}{Above}

	\end{llproof}

      \proofcase{$\impliedby$}

      \begin{llproof}
	Symmetric argument.
      \end{llproof}

      \end{proofcases}

      \proofcase{$\impliedby$}

      \begin{llproof}
	Symmetric argument.
      \end{llproof}

    \end{proofcases}

  \proofcase{\Rule{S-Compress}, \Rule{S-Gc}, \Rule{S-Exists-All}, \Rule{S-All-Escape}, \Rule{S-All-Rigid}, \Rule{S-All-Solve}}

  \begin{llproof}
    \begin{tabular}{l}
    Similar argument. Use \cref{lem:simple-ctxt-equiv}. \\
    The simple equivalences are standard, see \citet*{\BBemlti}.
    \end{tabular}
  \end{llproof}
  \end{proofcases}
\end{proof}

\subsection{Progress}

\begin{lemma}[Unification progress]
  If unification problem $\up$ cannot take a step $\up \unif \upp$, then either:
  \begin{enumerate}[(\roman*)]
    \item $\up$ is solved.
    \item $\up$ is $\cfalse$.
    \end{enumerate}
  \begin{proof}
    This is a standard result. See \citet*{\BBemlti}.
  \end{proof}
\end{lemma}

\begin{theorem}[Progress]
  \label{thm:progress}
  If constraint $\c$ cannot take a step $\c \csolve \cp$, then either:
  \begin{enumerate}
    \item $\c$ is solved.
    \item $\c$ is stuck, it is either:
      \begin{enumerate*}
      \item
        \label {item/progress/false}
        $\cfalse$;
      \item
        \label {item/progress/scope/x}
        $\hat\C\where{\capp \x \t}$ where $\x \disjoint \hat\C$;
      \item
        \label {item/progress/scope/i}
        $\hat\C\where{\cpapp \x \tv \tvc \inst}$ where
          $\x \disjoint \hat\C$ and $\inst(\tv) \disjoint \insts {\hat\C}$;
      \item
        \label {item/progress/suspended}
      for every match constraint $\c = \hat\C\where{\cmatch \t \cbrs}$, 
          $\Cshape {\hat\C} \t \sh$ does not hold for any $\sh$.
      \end{enumerate*}
    Here, $\hat\C$ is a normalized context (\cref{def/normal-forms-appendix}).  
  \end{enumerate}
\end{theorem}
\begin{proof}
  We proceed by induction on the structure of $\c$. We
  focus on suspended match constraints, conjunctions, and $\Let$ rules.
  \begin{proofcases}
    \proofcase{$\cmatch \t \cbrs$}
      We have two cases:
      \begin{proofcases}
  \proofcase{$\t$ is a non-variable type} Apply \Rule{S-Match-Ctx} using \Rule{S-Uni-Type}
	\proofcase{$\t$ is a type variable $\tv$}

	  We have $\nCshape \hole \tv$. It suffices
	  that every match constraint in a context-reachable position
	  $\hat\C\where{\cmatch \tp \cbrs}$ satisfies $\nCshape {\hat\C} \tp$.
	  By the definition of constraint contexts, there is only one such
	  $\hat\C$, namely $\hole$, for which we already have $\nCshape \hole \t$.
	  Hence $\cmatch \t \cbrs$ is stuck.
      \end{proofcases}

    \proofcase{$\ca \cand \cb$}
    We begin by inducting on $\ca$ and $\cb$. Then we consider cases:
    \begin{proofcases}
      \proofcase{$\ca$ (or $\cb$) take a step} Apply congruence rewriting rule.
      \proofcase{$\ca$ (or $\cb$) is $\ctrue$} Apply \Rule{S-True}.
      \proofcase{$\ca$ (or $\cb$) is $\cfalse$} Apply \Rule{S-False}.
      \proofcase{$\ca$ (or $\cb$) begins with $\exists$} Apply \Rule{S-Exists-Conj}.
      \proofcase{$\ca, \cb$ are solved}

	We either apply the above $\exists$ case, or both $\ca$ and $\cb$ are solved
	multi-equations $\ueqs_1, \ueqs_2$. We perform cases on this:
	\begin{proofcases}
	  \proofcase{$\ueqs_1$ and $\ueqs_2$ are mergeable} Apply \Rule{U-Merge}.
	  \proofcase{$\cyclic {\ueqs_1, \ueqs_2}$} Apply \Rule{U-Cycle}.
	  \proofcase{Otherwise} The conjunction $\ueqs_1 \cand \ueqs_2$ is solved.
	\end{proofcases}

      \proofcase{$\ca$ and $\cb$ are stuck (and not $\cfalse$)}

	\Wlog, consider cases $\ca$.
	\begin{proofcases}
	  \proofcase{$\hat\Ca\where{\capp \x \t}$}
	    We have $\x \disjoint \bvs {\hat\Ca}$.

	    $\hat\Ca\where{\capp \x \t} \cand \cb$ is stuck as we do not bind $\x$ in $\hat\Ca \cand \cb$.
	  \proofcase{$\hat\Ca\where{\cpapp \x \tv \tvc \inst}$}
      We have $\x \disjoint \bvs {\hat \Ca}$ and $\inst(\tv) \disjoint \insts {\hat \Ca}$.

      If $\inst(\tv) \in \insts \cb$ and $\inst \disjoint \bvs {\hat\ca}$, then apply \Rule{S-Inst-Unify}.
	    It must be the case that we can apply \Rule{S-Inst-Unify}, otherwise, we could lift these instantiation
	    constraints using \Rule{S-Exists-Lower} and \Rule{S-Let-ConjLeft}, contradicting that $\hat\Ca$ is stuck.

	    Otherwise, $\x \disjoint \bvs {\hat \Ca \cand \cb}$, thus $\hat\Ca\where{\cpapp \x \tv \tvc \inst}$ is stuck.

	  \proofcase{$\hat\Ca\where{\cmatch \tp \cbrs}$}
	    We have $\nCshape \Ca \tp$.

	    Consider a match constraint $\cmatch \tp \cbrs$ in $\ca$.

      If $\th \Cshape {\where{\hat\Ca\where{-} \cand \cb}} \tp \sh$. Then we can apply \Rule{S-Match-Ctx}.

      Otherwise $\not \th \Cshape {\where{\hat\Ca\where{-} \cand \cb}} \tp \sh$. We have \cref{lem:unicity-completeness},
      so we are stuck and $\nCshape {(\Ca \cand \cb)} \tp$.
	\end{proofcases}

    \end{proofcases}

    \proofcase{$\cletr \x \tv \tvs \ca \cb$}
    We begin by inducting on $\ca$ and $\cb$. Then we consider cases:
    \begin{proofcases}
      \proofcase{$\ca$ (or $\cb$) take a step} Apply congruence rewriting rule.
      \proofcase{$\ca$ (or $\cb$) is $\cfalse$} Apply \Rule{S-False}.
      \proofcase{$\ca$ begins with $\exists$} Apply \Rule{S-Let-ExistsLeft}
      \proofcase{$\cb$ begins with $\exists$} Apply \Rule{S-Let-ExistsRight}
      \proofcase{$\cb$ begins with $\cand$ with $\x \disjoint$ from conjunct} Apply \Rule{S-Let-ConjRight}.
      \proofcase{$\ca$ begins with $\cand$ with $\tv, \tvs \disjoint$ from conjunct } Try apply \Rule{S-Let-ConjLeft}
      \proofcase{$\cb$ begins with $\cexistsi \inst \xp {}$, $\x \neq \xp$} Apply \Rule{S-Exists-Inst-Let}
      \proofcase{$\tvp \in \tvs$ is determined by $\ca$} Apply \Rule{S-Exists-Lower}
      \proofcase{$\cb$ is solved}

	Thus $\cb$ must be $\ctrue$ (due to above cases).
	\begin{proofcases}
	  \proofcase{$\ca$ is solved}
	    Thus $\ca$ must be $\ueqs$.

	    There are two cases:
	    \begin{itemize}
	      \proofcase{$\cexists {\tv, \tvs} \ueqs \cequiv \ctrue$} Apply \Rule{S-Let-Solve}.
	      \proofcase{$\cexists {\tv, \tvs} \ueqs \cnequiv \ctrue$} It must be the case there is some $\tvb$ that dominates a $\tvp$ in $\tv, \tvs$ in $\ueqs$.
		Hence $\cdetermines {\cexists {\tv, \tvs \setminus \tvp} \ueqs} \tvp$.
		So we can apply \Rule{S-Exists-Lower}.
	    \end{itemize}

	  \proofcase{$\ca$ is stuck}

	    The constraint $\cletr \x \tv \tvs \ca \cb$ remains stuck, since
	    no additional term variable bindings occur for the scope of $\ca$,
	    ruling out the instantiation cases. Additionally, we cannot apply
	    backpropagation since $\cb$ is $\ctrue$.
	\end{proofcases}

      \proofcase{$\cb$ is stuck}
	\begin{proofcases}
	  \proofcase{$\hat\C\where{\capp \x \t}$} We have $\x \disjoint \bvs {\hat\C}$.

	  Apply \Rule{S-Let-AppR}.

    \proofcase{$\hat\C\where{\cpapp \x \tvb \tvc \inst}$} We have $\x \disjoint \bvs {\hat\C}$ or $\inst(\tvb) \disjoint \insts {\hat\C}$.
	    \begin{itemize}
		\proofcase{$\tvb \in \reg \tv \tvs$}

		We can either apply \Rule{S-Inst-Copy} or \Rule{S-Compress}
		if a multi-equation involving $\tvb$ occurs in $\ca$.

		Otherwise, we consider cases where $\ca$ is solved or stuck.

		If $\ca$ is solved, then it must be of the form $\ueqs$.
		There are two cases:
		\begin{itemize}
		  \proofcase{$\cexists {\tv, \tvs} \ueqs \cequiv \ctrue$}
		  As $\tvb$ does not appear in the head position of any multi-equation in $\ueqs$,
		  it must be polymorphic. Thus $\cfor \tvb {\cexists {\tv, \tvs \setminus \tvb} \ueqs} \cequiv \ctrue$.
		  So we can apply \Rule{S-Inst-Poly}.

		  \proofcase{$\cexists {\tv, \tvs} \ueqs \cnequiv \ctrue$}
		  Apply \Rule{S-Exists-Lower} (using the same logic as above).

		\end{itemize}

		If $\ca$ is stuck, then neither stuck case regarding instantiations
		in $\ca$ is fixed, so in these cases the constraint remains
		stuck. If $\ca$ is stuck with $\hat\Cp\where{\cmatch \tp
    \cbrs'}$. Then either backpropagation (via \Rule{S-Uni-BackProp} and \Rule{S-Match-Ctx})
		applies with an equation in $\hat\C$, or the entire constraint
    is stuck (by \cref{lem:unicity-completeness}).

		\proofcase{$\tvb \notin \reg \tv \tvs$} Apply \Rule{S-Inst-Mono}.

	    \end{itemize}

	  \proofcase{For any $\hat\C\where{\cmatch \tp \cbrs}$} We have $\nCshape {\hat\C} \tp$.

	  Either $\cletr \x \tv \tvs \ca \cb$ can progress with an instantiation constraint (in the above case) to discharge
	  the match constraint or $\cletr \x \tv \tvs \ca \cb$ is stuck.
	\end{proofcases}
    \end{proofcases}
 \end{proofcases}
\end{proof}

\subsection{Termination}

This section presents a proof of termination for our solver.
Most rewrite rules, in both unification and constraint solving, are
\emph{destructive}---that is, they eliminate or modify the structure of a
constraint in a way that prevents the rule from begin applied again.
Consequently, to establish termination, it suffices to consider only those
rules that are not inherently destructive.

\begin{lemma}[Unification termination]
  \label{lem:unification-termination}
  The unifier terminates on all inputs.
  \begin{proof}
    \newcommand{\sw}[1]{\kwd{sw}{(#1)}}
    \newcommand{\iw}[1]{\kwd{iw}{(#1)}}
    \newcommand{\tw}[1]{\kwd{tw}{(#1)}}
    \newcommand{\uw}[1]{\kwd{uw}{(#1)}}

    Let every shape $\sh$ have an integer \emph{weight}
    defined by $\sw \sh \eqdef 4 + 2 \times |\sh|$, where $|\sh|$ is the
    arity of the shape $\sh$.
    The weight of a type $\tw \t$ is defined by:
    \begin{mathpar}
      \begin{tabular}{RCL}
	\tw \tv &\eqdef& 1\\
	\tw {\shapp \tys} &\eqdef& \iw {\shapp \tys} - 2\\[1ex]
	\iw \tv &\eqdef& 0\\
	\iw {\shapp \tys} &\eqdef& \sw \sh + \iw \tys\\
        \iw \tys & \eqdef & \sum\iton \iw \ti\\

      \end{tabular}
    \end{mathpar}
    The helper $\iw \t$ computes the ``internal'' weight of $\t$; in
    the common case of shallow types it is just the weight of its head
    shape.

    We define the weight of a multi-equation as the sum of the weights of its
    members. The weight of a unification problem $\uw \up$ is defined
    as the sum of the weights of its multi-equations.

    In $\up \unif \upp$, the rules \Rule{U-Decomp} and \Rule{U-Name} are not
    obviously destructive, as they may introduce new constraints that
    are structurally larger than the constraint being rewritten.

    However, we show that this is not problematic: in both cases, the unification
    weight $\uw \up$ strictly decreases. The remaining rules are obviously
    destructive and either maintain or decrease the unification weight.

    \begin{proofcases}
      \proofcaserewrite
	{U-Decomp}
	{\cunif {\pshapp \tvs} {\cunif {\pshapp \tvbs} \ueq}}
	{\cunif {\pshapp \tvs} \ueq \cand \cunif \tvs \tvbs}

	We have:
	\begin{mathpar}
	  \begin{tabular}{RRCL}
           (+)&
	    \uw {\cunif {\pshapp \tvs} {\cunif {\pshapp \tvbs} \ueq}} &=&
	      \tw {\pshapp \tvs} + \tw {\pshapp \tvbs}  + \tw \ueq \\
           (-)&
	    \uw {\cunif {\pshapp \tvs} \ueq \cand \cunif \tvs \tvbs}
	      &=&
	      \tw {\pshapp \tvs} + \tw \ueq + \tw \tvs + \tw \tvbs \Strut \\
              \hline
	       &&=&\Strut \tw {\pshapp \tvbs} - \tw \tvs - \tw \tvbs \\
               &&=&  (\sw \sh  + 0 - 2) - 2 |\sh| \\
               &&=&  (2 + 2|\sh|) - 2 |\sh| \Wide = \textbf {2}\\
	  \end{tabular}
	\end{mathpar}
	Hence $\uw {{\cunif {\pshapp \tvs} {\cunif {\pshapp \tvbs} \ueq}}} > \uw {\cunif {\pshapp \tvs} \ueq \cand \cunif \tvs \tvbs}$.

      \proofcaserewrite
	{U-Name}
      {\cunif {\pshapp \parens{\tys, \ti, \typs}} \ueq \\ \tv \disjoint \tys, \typs, \ueq \\ \ti \notin \TyVars }
      {\cexists \tv {\cunif {\pshapp \parens{\tys, \tv, \typs}} \ueq \cand \cunif \tv \ti}}

	Given $\ti \notin \TyVars$, by \cref{thm:principal-shapes},
	$\ti = \shapp[\shp] \bar\typp$ for some shape $\shp$ and types $\bar\typp$.
	So we have:
	\begin{mathpar}
	  \begin{tabular}{.R;;R;C;L.}
          (+) &
	    \uw {\cunif {\pshapp {\parens{\tys, \ti, \typs}}} \ueq}
            &=& \sw \sh + \iw \tys + \iw \ti + \iw \typs - 2 + \uw \ueq \\
          (-) &
	    \uw {\cexists \tv {\cunif \tv \ti \cand \cunif {\pshapp \parens{\tys, \tv, \typs}} \ueq}}
            &=& \sw \sh + \iw \tys + 0 + \iw \typs - 2 + \uw \ueq + 1 + \tw \ti \\
            \hline
           &&=& \iw \tyi - \iw \tv - \tw \ti - 1 \\
	    &&=& \iw \tyi - 0 - (\iw \tyi - 2) - 1 \\
	    &&=& \textbf{1}
	  \end{tabular}
	\end{mathpar}

	Hence $\uw {\cunif {\pshapp \parens{\tys, \ti, \typs}} \ueq } > \uw {\cexists \tv {\cunif {\pshapp \parens{\tys, \tv, \typs}} \ueq \cand \cunif \tv \ti}}$.
    \end{proofcases}
  \end{proof}
\end{lemma}

\terminationBIS*
\begin{proof}
\begin{local}\sloppy
  The difficulty for termination comes from the ``suspended match discharge'' rule
  \Rule{S-Match-Ctx} which can make arbitrary
  sub-constraints appear in the non-suspended part of the constraint;
  and from the instantiation rules that copy/duplicate existing
  structure in another part of the constraint, increasing its total
  size.

\end{local}
  As we argued before, the other rewrite rules are \emph{destructive},
  they strictly simplify the constraint towards a normal form and can
  only be applied finitely many times when taken together. The
  fragment without discharge rules and incremental instantiation is
  also extremely similar to the constraint language of
  \citet*{\BBemlti}, so their termination proof applies
  directly.

  \paragraph{Discharge rules} The discharge rules strictly decrease
  the number of occurrences of suspended match constraint (if we also
  count nested suspended constraints), and no rewriting rule
  introduces new suspended match constraints. So these discharge rules
  can only be applied finitely many times. To prove termination of
  constraint solving, it thus suffices to prove that rewriting
  sequences that do not contain one of the discharge rules (those that
  occur in-between two discharge rules) are always finite.

  \paragraph{Starting instantiations} By a similar argument, the number
  of non-incremental instantiations $\capp \x \t$ decreases strictly on
  \Rule{S-Let-AppR} when an incremental instantiation starts, and is
  preserved by other non-discharge rules. The rule \Rule{S-Let-AppR}
  can thus only occur finitely many times in non-discharging sequences,
  and it suffices to prove that all rewriting sequences that are
  non-discharging and do not contain \Rule{S-Let-AppR} are finite.

  \paragraph{Other instantiation rules} Among other instantiation
  rules, the rule of concern is \Rule{S-Inst-Copy}, which is not
  destructive: it introduces new instantiation constraints
  and structurally increases the size of the constraint.

  Intuitively, \Rule{S-Inst-Copy} should not endanger termination
  because the amount of copying it can perform for a given
  instantiation is bounded by the size of the types in the constraint
  $\c$ it is copying from. ($\c$ could have cyclic equations with
  infinite unfoldings, but \Rule{S-Inst-Copy} forbids copying in that
  case.) The difficulty is that rewrites to $\c$ can be interleaved
  with instantiation rules, so that the equations that are being
  copied can grow strictly during instantiation.

  To control this, we perform a structural induction: to prove that
  $(\cletr \x \tv \tvs \ca \cb)$ does not contain infinite
  non-discharging non-instance-starting rewrite rules, we can assume
  that the result holds for the strictly smaller constraint $\ca$, and
  then prove termination of the incremental instantiations of $\x$ in
  $\cb$. (The notion of structural size used here is preserved by
  non-discharging rewrite rules, as they do not affect the
  let-structure of the constraint.)

  Assuming that $\ca$ has no infinite rewriting sequence, it suffices
  to prove that only finitely many rewrites in the rest of the
  constraint (namely $\cb$) can occur between each rewrite of $\ca$.

  \newcommand{\sw}[1]{\mathprefix{sw}{(#1)}}
  \newcommand{\iw}[1]{\mathprefix{iw}{(#1)}}
  \newcommand{\stw}[1]{\mathprefix{tw}{(#1)}}
  \newcommand{\tw}[2]{\mathprefix{tw}{(#1 \in #2)}}
  \newcommand{\eqs}[1]{\mathprefix{eqs}({#1})}
  \newcommand{\cw}[1]{\mathprefix{cw}{(#1)}}

  We define a weight that captures the contribution of types
  within $\ca$ to the partial instances in~$\cb$:
  \begin{mathpar}
    \begin{tabular}{RCL}
      \stw {\shapp \tys} &\eqdef& 2 \times \sw \sh + \sum \iton \stw \ti \\[1ex]
      \stw \tv &\eqdef&
      \left\{
        \begin{array}{ll}
        \sup \Braces {\stw \t : \cunif \tv \t \in \ca } & \text{if $\ca$ is acyclic} \\
        0 & \text{otherwise}
        \end{array}
      \right.
    \end{tabular}
  \end{mathpar}
  The weight of a incremental instantiation $\cw {\cpapp \x \tv \t \inst}$ is
  defined as  the sum of $\stw \t$ and ${\stw \tv}$.
  The weight of other constraints is given using the measure
  $\mathprefix{uw}$ defined in the the
  proof of \cref{lem:unification-termination}.

  \begin{proofcases}
    \proofcaserewrite
      {S-Inst-Copy}
      {\cletr \x \tv \tvs {\c}
	\C\where{\cpapp \x \tvb \tvc \inst}\\
	\neg \cyclic {\c} \\
\x \disjoint \bvs \C \\
	\c = \cp \cand \cunif \tvb {\cunif {\shapp \tvbs} \ueq}\\
	\tvb \in \reg \tv \tvs \\
       \tvbs' \disjoint \tvb, \tvc, \tvbs
       }
      {\cletr \x \tv \tvs {\c}
	\C\where{\cexists {\tvbs'} \cunif \tvc {\shapp \tvbs'} \cand \cpapp \x {\tvbs} {\tvbs'} \inst}}

      We aim to show that the weight of the rewritten constraint
      $\cexists {\tvbs'} \cunif \tvc {\shapp \tvbs'} \cand \cpapp \x {\tvbs} {\tvbs'} \inst$
      is strictly less than the original $\cpapp \x \tvb \tvc \inst$.

      \begin{mathpar}
	\begin{tabular}{RCL}
	  \cw {\cpapp \x \tvb \tvc \inst} &=& 1 + \stw \tvb \\
	  &\geq& 1 + 2 \times \sw \sh + \sum\iton \stw \tvbi  \\[1ex]
	  \cw {\cexists {\tvbs'} \cunif \tvc {\shapp \tvbs'} \cand \cpapp \x {\tvbs} {\tvbs'} \inst}
	  &=&
	  1 + \sw \sh + \sum\iton \stw \tvbi + |\tvbs'|
	\end{tabular}
      \end{mathpar}
      To ensure a strict decrease, it suffices to show that $\sw \sh > |\tvbs'|$.
      Given that $|\tvbs'| = |\sh|$, and by the definition of $\sw \sh$, this inequality holds.
      Therefore, the weight strictly decreases under \Rule{S-Inst-Copy}.

  \end{proofcases}

  Thus the constraint solver terminates.
\end{proof}

\subsection{Correctness}

\begin{lemma}
  \label{lem:unsat-match}
  Given non-simple $\c$ constraint. If every match constraint $\C\where{\cmatch \t \cbrs} = \c$
  satisfies $\nCshape \C \t$, then $\c$ is unsatisfiable.
  \begin{proof}
    By contradiction, inverting on the canonical derivation of $\c$.
  \end{proof}
\end{lemma}

\begin{lemma}[Scope preservation]
  \label{lem:scoping-preservation}
  For all $\ca, \cb$, if $\ca \csolve \cb$, then $\fvs \ca \supseteq \fvs \cb$.
  \begin{proof}
    By induction on $\ca \csolve \cb$.
  \end{proof}
\end{lemma}

\begin{corollary}
  \label{corollary:apx-correctness}
  For the closed-term-variable constraint $\c$, $\c$ is satisfiable if and only if
  $\c \csolve^* \hat\c$ and $\hat\c$ is a solved form equivalent to $\c$.
  \begin{proof}
    We show each direction individually:
    \begin{proofcases}
      \proofcase{$\implies$}

	By transfinite induction on the well-ordering of constraints whose
	existence is shown in \cref{thm:termination}.

	We have $\c$ is satisfiable.
	By \cref{thm:progress}, we have three cases:
	\begin{proofcases}
	  \proofcase{$\c$ is solved}
	  We have $\c \csolve^* \c$ and $\c \cequiv \c$ by reflexivity.
	  So we are done.

	  \proofcase{$\c$ is stuck}
	    Given $\c$ is a closed-term-variable constraint,
	    it must be the case that either $\c$ is $\cfalse$
      \emph{or} $\hat\C\where{\cmatch \t \cbrs}$ and $\nCshape \C \t$.

	    If $\c$ is $\cfalse$, this contradicts our assumption that $\c$ is satisfiable.
	    Similarly, by \cref{lem:unsat-match}, if $\c$ is $\hat\C\where{\cmatch \t \cbrs}$,
	    then this also contradicts the satisfiability of $\c$.

	  \proofcase{$\c \csolve \cp$}

	  By \cref{thm:preservation}, we have $\c \cequiv \cp$, thus $\cp$ is satisfiable.
    Additionally, by \cref{lem:scoping-preservation}, we have $\fvs \cp = \eset$.
	  So by induction, we have $\cp \csolve^* \hat \c$ and $\hat\c$
	  is a solved form equivalent to $\cp$.
	  By transitivity of equivalence, we therefore have $\hat\c \cequiv \c$, as
	  required.

	\end{proofcases}

      \proofcase{$\impliedby$}

      By induction on the rewriting $\c \csolve^* \hat\c$.
      \begin{proofcases}
	\proofcaserewrite
	  {Zero-Step}
	  { }
	  {\hat\c \csolve^* \hat\c}

	We have $\c = \hat\c$ by inversion. All solved forms are satisfiable, thus
	$\c$ is satisfiable.

	\proofcaserewrite
	  {One-Step}
	  {\c \csolve \cp \\ \cp \csolve^* \hat\c}
	  {\c \csolve^* \hat\c}

	  By induction, we have $\cp$ is satisfiable. By \cref{thm:preservation},
	  $\c \cequiv \cp$, hence $\c$ is satisfiable.

      \end{proofcases}
    \end{proofcases}
  \end{proof}
\end{corollary}

\section{Properties of \OML}
\label{app/oml/proofs}


This section states and proves the two central metatheoretic properties of
\OML. The first is the \emph{soundness and completeness} of the constraint
generator $\cinfer \e \tv$ with respect to the \OML typing rules. The second
is the existence of \emph{principal types}, which follows as a consequence
of soundness and completeness: every closed well-typed term $\e$ admits a
most general type.


Throughout this section, we restrict our attention to \emph{closed
terms}. This is because the typing context $\G$ can contain bindings to
terms whose type is ``guessed''. When we generate constraints for a term
$\e$ under a context $\G$, we encode the type schemes in $\G$ as part of the
constraint itself using let-constraints. However, these schemes are
treated as known within the constraint! As a result, we assume terms are
closed from the outside to avoid $\G$ leaking any guessed type information.

\subsection{Simple syntax-directed system}

As a first step towards proving soundness and completeness of constraint
generation, we first present a variant of the \OML type system for
\emph{simple terms}. For this system, the syntax tree completely determines
the derivation tree.

We use the standard technique of removing the \Rule{Inst} and \Rule{Gen} rules,
and always apply instantiations in \Rule{Var} (\Rule{Var-SD}) and always
generalize at let-bindings (\Rule{Let-SD}). We can show that this system is
sound and complete with respect to the declarative rules.

\begin{theorem}[Soundness of the syntax directed rules]
  \label{thm:soundness-sd}
  Given the simple term $\e$.
  If $\G \thsimplesd \e : \t$ then we also have $\G \thsimple \e : \t$
  \begin{proof}
    Induction on the given derivation.
  \end{proof}
\end{theorem}

\begin{theorem}[Completeness of the syntax directed rules]
  \label{thm:completeness-sd}
  Given the simple term $\e$.
  If $\G \thsimple \e : \ts$, then $\G \thsimple \e : \t$ for any instance $\t$ of $\ts$.
  \begin{proof}
    Induction on the given derivation.
  \end{proof}
\end{theorem}

\paragraph{Inversion} On a simple syntax-directed derivation $\G \thsimplesd \e : \t$, we have the usual inversion
principle:
\begin{lemma}[Simple inversion]
  \label{lem:simple-inversion-sd}
  ~
  \begin{enumerate}[(\roman*)]
    \item If $\G \thsimplesd \x : \t$, then $\x : \tfor \tvs \tp \in \G$ and $\t = \tp\where{\tvs \is \tys}$.
    \item If $\G \thsimplesd \efun \x \e : \t$, then $\G, \x : \ta \thsimplesd \e : \tb$ and $\t = \ta \to \tb$.
    \item If $\G \thsimplesd \eapp \ea \eb : \t$, then $\G \thsimplesd \ea : \tp \to \t$ and $\G \thsimplesd \eb : \tp$.
    \item If $\G \thsimplesd \eunit : \t$, then $\t = \tunit$.
    \item If $\G \thsimplesd \elet \x \ea \eb : \t$, then $\G \thsimplesd \ea : \tp$, $\tvs \disjoint \fvs \G$, and $\G, \x : \tfor \tvs \tp \thsimplesd \eb : \t$.
    \item If $\G \thsimplesd \eannot \e \tvs \tp : \t$, then $\G \thsimplesd \e : \tp\where{\tvs \is \tys}$ and $\t = \tp\where{\tvs \is \tys}$.
   \item If $\G \thsimplesd \etuple {\ea, \ldots, \en} : \t$, then $\G \thsimplesd \ei : \ti$ for all $1 \leq i \leq n$ and $\t = \tProd \ti$.
    \item If $\G \thsimplesd \exproj \e j n : \t$, then $\G \thsimplesd \e : \tProd \ti$ and $\t = \tj$, with $n \geq j$.
    \item If $\G \thsimplesd \expoly \e \tvs {\tfor \tvbs \tp} : \t$, then $\G \thsimplesd \e : \t\where{\tvs \is \tys}$, $\tvbs \disjoint \G$ and
      $\t = \tpoly {\tfor \tvbs \tp}\where{\tvs \is \tys}$.
    \item If $\G \thsimplesd \exinst \e \tvs \ts : \t$, then $\G \thsimplesd \e : \tpoly \ts\where{\tvs \is \tys}$ and $\ts\where{\tvs \is \tys} \leq \t$.
    \item If $\G \thsimplesd \emagic \es : \t$, then $\G \thsimplesd \ei : \tip$ for all $1 \leq i \leq n$.
    \item If $\G \thsimplesd \exrecord \T {\overline{\elab = \e}} : \t$, then $\G \thsimplesd \ei : \ti$ and ${\labfrom \elab \T} \leq \t \to \ti$ for $1 \leq i \leq n$ and $\Dom {\labfrom \labenv \T} = \elabs$.
    \item If $\G \thsimplesd \erecord {\overline{\elab = \e}} : \t$, then $\labsuni \elabs \T$ and $\G \thsimplesd \exrecord \T {\overline{\elab = \e}} : \t$.
    \item If $\G \thsimplesd \exfield \e \T \elab : \t$, then $\G \thsimplesd \e : \tp$, ${\labfrom \elab \T} \leq \tp \to \t$.
    \item If $\G \thsimplesd \efield \e \elab : \t$, then $\labuni \elab \T$ and $\G \thsimplesd \exfield \e \T \elab : \t$.
  \end{enumerate}
\end{lemma}

\subsection{Canonicalization of typability}
Our system satisfies a similar canonicalization theorem to constraint satisfiability.

\begin{lemma}[Composability of unicity]
  ~
  \label{lem:comp-unicity-typing}
  \begin{enumerate}[(\roman*)]
    \item If $\Eshape \Ea \es \sh$, then $\Eshape {\Eb\where\Ea} \es \sh$.
    \item If $\eshape \Ea \e \sh$, then $\eshape {\Eb\where\Ea} \e \sh$.
  \end{enumerate}
  \begin{proof}
    By induction on $\Eb$.
  \end{proof}
\end{lemma}

\begin{lemma}[Decanonicalization]
  \label{lem:decanonicalization-typing}
  If $\Th \e : \t$, then $\eset \th \e : \t$.
  \begin{proof}
    By induction on the given derivation $\Th \e : \t$.
  \end{proof}
\end{lemma}

\newcommand{\enimplicit}[1]{{\#\mathprefix[\mathsf]{implicit} {#1}}}
\begin{theorem}[Canonicalization]
  \label{thm:canonicalization-typing}
  If $\th \e : \ts$, then $\Th \e : \t$ for any instance $\t$ of $\ts$.
  \begin{proof}
    By induction on the following measure of $\e$:
    \begin{mathpar}
      \| \e \| \uad\eqdef\uad \angles {\enimplicit \e, |\e|}
    \end{mathpar}
    where $\angles \ldots$ denotes a lexicographically ordered pair, and
  \begin{enumerate}

    \item $\enimplicit \e$ is the number of implicit constructs in $\e$ \ie overloaded tuple projections $\eproj \e j$,
      implicit non-unique field projections $\efield \e \elab$, implicit non-unique records $\erecord {\overline{\elab = \e}}$, polytype instantiations $\einst \e$
      and polytype boxing $\epoly \e$.

    \item the last component $|\e|$ is a structural measure of terms \ie a
      application $\eapp \ea \eb$ is larger than the two terms $\ea, \eb$.
  \end{enumerate}
  This measure is analogous to the measure $\cmeasure \c$ for constraints.
  \end{proof}
\end{theorem}

\subsection{Unifiers}

A substitution $\sub$ is an idempotent function from type variables to types.
The (finite) domain of $\sub$ is the set of type variables such that $\sub(\tv)
\neq \tv$ for any $\tv \in \dom \sub$, while the codomain consists of the free
type variables of its range.
We use the notation $\where{\tvs \is \tys}$ for the substitution $\sub$ with
domain $\tvs$ and $\sub(\tvs) = \tys$.

The constraint induced by a substitution $\sub$, written $\exists \sub$, is
$\cexists {\tvbs} \tvs = \tys$ where $\tvbs = \rng \sub$, $\tvs = \dom \sub$
and $\sub(\tvs) = \tys$.

\begin{definition}[Unifier]
  A substitution $\sub$ is a unifier of $\c$ if $\exists \sub$ entails $\c$.
  A unifier $\sub$ of $\c$ is \emph{most general} when $\exists \sub$ is equivalent
  to $\c$.
\end{definition}

\begin{lemma}[Simple inversion of unifiers]
  \label{lem:unifier-simple-inversion}
  ~
  \begin{itemize}
    \item If $\sub$ is a unifier of $\cunif \ta \tb$, then $\sub(\ta) = \sub(\tb)$.
    \item For simple $\ca, \cb$, if $\sub$ is a unifier of $\ca \cand \cb$, then $\sub$ is a unifier of $\ca$ and $\cb$.
    \item For simple $\c$, if $\sub$ is a unifier of $\cexists \tv \c$, then $\sub\where{\tv \is \t}$ is a unifier of $\c$ for some $\t$.
    \item For simple $\c$, if $\sub$ is a unifier of $\cfor \tv \c$, then $\sub$ is a unifier of $\c$.
  \end{itemize}
  \begin{proof}
    Follows by simple inversion.
  \end{proof}
\end{lemma}

\begin{lemma}
  \label{lem:unifier-abs-equiv}
  If $\sub$ unifies $\cexists \tv \c$, then there exists a unifier $\subp$ that extends $\sub$ with $\tv$,
  where $\subp$ is most general unifier of $\exists \sub \cand \c$.

  Then $\cabs \tv \c$ is equivalent to $\cabs \tv \sigma \leq \tv$ under $\sub$, where $\ts = \tfor \tvbs \subp(\tv)$ and
  $\tvbs = \fvs {\subp(\tv)} \setminus \rng \sub$. We write this equivalent constraint abstraction as $\csem {\cabs \tv \c}_\sub$.
  \begin{proof}
    See \citet*{\BBemlti}.
  \end{proof}
\end{lemma}

\begin{lemma}[Let inversion of unifiers]
  For simple $\ca, \cb$.
  If $\sub$ unifies $\clet \x \tv \ca \cb$, then
  $\sub$ unifies $\cexists \tv \ca$ and
  $\sub$ unifies $\cletin \x {\csem {\cabs \tv \ca}_\sub} \cb$
  \begin{proof}
    Follows from \cref{lem:unifier-abs-equiv} and simple inversion.
  \end{proof}
\end{lemma}

\begin{lemma}
  \label{lem:mgus}
  For two substitutions $\sub$, $\subp$. If $\exists \sub \centails \exists
  \subp$, there exists  $\subpp$ such that $\sub = \subpp \compose \subp$.
  \begin{proof}
    Standard result, follows from definition of $\exists \sub$.
  \end{proof}
\end{lemma}

\subsection{Soundness and completeness of constraint generation}

\begin{lemma}
  \label{lem:ctxt-gen-correctness}
  For any term context $\E$, term $\e$, $\ctxinfer \E \t \tp \where{\cinfer
  \e \t} = \cinfer {\E\where{\e}} \tp$.
  \begin{proof}
    By induction on the structure of $\E$.
  \end{proof}
\end{lemma}

\begin{lemma}
  \label{lem:erasure-constraint-gen}
  For any term $\e$, $\cerase {\cinfer \e \t} = \cinfer {\eerase \e} \t$.
  \begin{proof}
    By induction on $\e$.
  \end{proof}
\end{lemma}

\begin{lemma}[Simple soundness and completeness]
  \label{lem:simple-soundness-completeness}
  For simple terms $\e$.  $\sub(\G) \thsimplesd \e : \sub(\t)$ if and only
  if $\sub$ is a unifier of $\csem {\G \th \e : \t}$.
  \begin{proof}
    By induction on $\e \simple$.
  \end{proof}
\end{lemma}

\begin{theorem}[Soundness and completeness of constraint generation]
  \label{thm:soundness-and-completeness}
  $\Th \e : \sub(\tv)$ if and only if $\sub$ is a unifier of $\csem {\e : \tv}$
  \begin{proof}
    By induction on the number $n$ of implicit terms in $\e$.
    \begin{proofcases}
      \proofcase{$n$ is $0$}

	\begin{llproof}
	  \simplePf{\e}{Premise}
	  \iffPf{\eset\thsimplesd \e : \sub(\tv)}{\sub \text{ unifies } \csem {\e : \tv}}{\cref{lem:simple-soundness-completeness}}
	  \iffPf{\eset \thsimplesd \e : \sub(\tv)}{\Th \e : \sub(\tv)}{When $\e \simple$}
\Hand	  \iffPf{\Th \e : \sub(\tv)}{\sub \text{ unifies } \csem {\e : \tv}}{Above}
	\end{llproof}

      \proofcase{$n$ is $k + 1$}

	\begin{proofcases}
	  \proofcase{$\implies$}

	  \begin{proofcases}

	    \proofcasederivation
	      {Can-Proj-I}
	      {\eshape \E \e {\any \tvcs \Pi\iton \tvcs} \\ \sub(\G) \Th \E\where{\exproj \e j n} : \sub(\tv)}
	      {\Th \E\where{\eproj \e j} : \sub(\tv)}

	      \begin{llproof}
		\ThTypPf{\sub(\G)}{\E\where{\exproj \e j n}}{\sub(\tv)}{Premise}
		\UnifierPf{\sub}{\csem {\G \th \E\where{\exproj \e j n} : \tv}}{By \ih}
		\eqPf{\csem {\G \th \E\where{\exproj \e j n } : \tv}}{\cletG {\cinfer {\E\where{\exproj \e j n}} \tv}}{By definition}
		\continueeqPf{\cletG {\ctxinfer \E \tvb \tv}\where{\cinfer {\exproj \e j n} \tvb}}{\cref{lem:ctxt-gen-correctness}}
		\equivPf{\cinfer {\exproj \e j n} \tvb}{\cexists {\tvaa \tvcs} {\cinfer \e \tvaa \cand \cunif \tvaa {\Pi\iton \tvcs} \cand \cunif \tvb \tvc_j}}{By definition}
		\continueequivPf
                  {\begin{tabular}[t]{L}\cexists {\tvaa} \cinfer \e \tvaa
                  \\ \quad \cand \cmatched \tvaa {\any \tvcs \Pi\iton
                \tvcs}{\cbranch {\cpatprod \tvc j} {\cunif \tvb \tvc}}
                        \end{tabular}}{\ditto}
		\UnifierPf{\sub}{\cletG \ctxinfer \E \tvb \tv \where{\cexists \tvaa \cinfer \e \tvaa \cand \ldots}}{Above}
		\eshapePf{\E}{\e}{\any \tvcs \Pi\iton \tvcs}{Premise}
		\LetPf{\C}{\cletG \ctxinfer\E \tvb \tv \where{\cexists \tvaa \cinfer \e \tvaa \cand \hole}}{}
		\vdashPf{\semenv}{\cerase{\C\where{\cunif \tvaa \gt}}}{Premise}
		\eqPf{\cexists \tvaa \cinfer \e \tvaa \cand \cunif \tvaa \gt}{\cexists \tvaa \cinfer {\eannot \e {} \gt} \tvaa}{By definition}
		\continueeqPf{\cinfer {\emagic {\eannot \e {} \gt}} \tvb}{\ditto}
		\eqPf{\cerase {\C\where{\cunif \tvaa \gt}}}{\cerase {\cletG \ctxinfer \E \tvb \tv \where{\cinfer {\emagic {\eannot \e {} \gt}} \tvb}}}{\ditto}
		\continueeqPf{\cerase {\cletG \cinfer {\E\where{\emagic {\eannot \e {} \gt}}} \tv}}{\cref{lem:ctxt-gen-correctness}}
		\continueeqPf{\cletG \cerase {\cinfer {\E\where{\emagic {\eannot \e {} \gt}}} \tv}}{By definition}
		\continueeqPf{\cletG \cinfer {\eerase {\E\where{\emagic {\eannot \e {} \gt}}}} \tv}{\cref{lem:erasure-constraint-gen}}
		\UnifierPf{\semenv}{\cletG \cinfer {\eerase {\E\where{\emagic {\eannot \e {} \gt}}}} \tv}{Above}
		\ThTypPf{}{\eerase {\E\where{\emagic {\eannot \e {} \gt}}}}{\semenv(\tv)}{By \ih}
		\thTypPf{\eset}{\eerase {\E\where{\emagic {\eannot \e {} \gt}}}}{\semenv(\tv)}{\cref{lem:decanonicalization-typing}}
		\decolumnizePf
		\eqPf{\shape \gt}{\any \tvcs \Pi\iton \tvcs}{$\implies$E}
		\shapePf \C \tvaa {\any \tvcs \Pi\iton \tvcs}{Above}
		\UnifierPf{\sub}{\C\where{\cmatch \tvaa {\cbranch {\cpatprod \tvc j} {\cunif \tvb \tvc}}}}{By \Rule{Match-Ctx}}
		\eqPf{\cinfer {\eproj \e j} \tvb}{\cexists \tvaa \cinfer \e \tvaa \cand \cmatchdots \tvaa}{By definition}
		\eqPf{\C\where{\cmatchdots \tvaa}}{\cletG \ctxinfer \E \tvb \tv\where{\cexists \tvaa \cinfer \e \tvaa \cand \ldots}}{\ditto}
		\continueeqPf{\cletG \ctxinfer \E \tvb \tv \where{\cinfer {\eproj \e j} \tvb}}{Above}
		\continueeqPf{\cletG \cinfer {\E\where{\eproj \e j}} \tv}{\cref{lem:ctxt-gen-correctness}}
		\continueeqPf{\csem {\E\where{\eproj \e j} : \tv}}{}
\Hand  		\UnifierPf{\sub}{\csem {\E\where{\eproj \e j} : \tv}}{}
	      \end{llproof}

      \proofcase{\Rule{Can-Poly-I}, \Rule{Can-Use-I}, \Rule{Can-Rcd-I}, \Rule{Can-Rcd-Proj-I}}

	    \begin{llproof}
	      Similar arguments.
	    \end{llproof}
	  \end{proofcases}

	  \proofcase{$\impliedby$}

	  \begin{proofcases}
	    \proofcasederivation
	      {Can-Match-Ctx}
	      {\Cshape \C \tvaa {\any \tvcs \Pi\iton \tvcs} \\
	       \sub \text{ unifies } \C\where{\cmatched \tvaa {\any \tvcs \Pi\iton \tvcs} {\ldots}}}
	      {\sub \text{ unifies } \underbrace{\C\where{\cmatch \tvaa {\cbranch {\cpatprod \tvc j} {\cunif \tvb \tvc}}}}_{\cinfer \e \tv}}

	      \begin{llproof}
		\eqPf{\csem {\e : \t}}{\cletG \cinfer {\E\where{\eproj \e j}} \tv}{Premise}
    \decolumnizePf
		\eqPf{\C}{\cletG \ctxinfer \E \tvb \tv\where{\cexists \tv \cinfer \e \tv \cand \hole}}{Premise}
		\UnifierPf{\sub}{\C\where{\cmatched \tvaa {\any \tvcs \Pi\iton \tvcs} {\ldots}}}{Premise}
		\UnifierPf{\sub}{\csem {\E\where{\exproj \e j n} : \tv}}{Above (See $\implies$ direction)}
		\ThTypPf{}{\E\where{\exproj \e j n}}{\sub(\tv)}{By \ih}
		\thTypPf{\Gp}{\E\where{\emagic {\eannot \e {} \gt}}}{\tp}{Premise}
		\eqPf{\Gp}{\eset}{$\E\where{\emagic {\eannot \e {} \gt}}$ is closed}
		\ThTypPf{}{\E\where{\emagic {\eannot \e {} \gt}}}{\tp}{\cref{lem:decanonicalization-typing}}
		\UnifierPf{\where{\tv \is \tp}}{\csem {\E\where{\emagic {\eannot \e {} \gt}} : \tv}}{By \ih}
		\vdashPf{\semenv\where{\tv \is \semenv(\tp)}}{\cinfer {\E\where{\emagic {\eannot \e {} \gt}}} \tv}{By definition}
		\shapePf{\C}{\tvaa}{\any \tvcs \Pi\iton \tvcs}{Premise}
		\eqPf{\shape \gt}{\any \tvcs \Pi\iton \tvcs}{$\implies$E}
		\eshapePf{\E}{\e}{\any \tvcs \Pi\iton \tvcs}{Above}
		\ThTypPf{}{\E\where{\eproj \e j}}{\sub(\tv)}{By \Rule{Can-Proj-I}}
	      \end{llproof}

      \proofcase{$\epoly \e$, $\einst \e$, $\erecord {\overline{\elab = \e}}$, $\efield \e \elab$}

	    \begin{llproof}
	      Similar arguments.
	    \end{llproof}
	  \end{proofcases}
	\end{proofcases}
    \end{proofcases}
  \end{proof}
\end{theorem}

\subsection{Principal types}
\label{apx/principal-types}

\principalTypesBIS*
\begin{proof}

Let $\e$ be an arbitrary closed well-typed term; that is, there exists a
type $\t$ such that $\th \e : \t$.  By 
\cref{thm:soundness-and-completeness}, the constraint $\cinfer \e \tv$ is
satisfiable (specifically under the unifier $\cunif \tv \t$). By
\cref{corollary:apx-correctness}, there exists a solved constraint $\hat\c$ such
that $\hat\c \cequiv \cinfer \e \tv$.  From $\hat\c$, we extract a
unifier~$\sub$. Since $\hat\c \cequiv \exists \sub$, it follows that $\sub$
is  \emph{most general}.

We claim that $\sub(\tv)$ is the principal type of $\e$. This amounts to showing:
\begin{enumerate}[(\roman*)]

\item
  \label{proof:principal-types:1}
  $\th \e : \sub(\tv)$

\item
  \label{proof:principal-types:2}
  For any other typing $\th \e : \tp$, then $\tp = \theta(\sub(\tv))$ for
  some $\theta$.

\end{enumerate}
Since $\sub$ is a unifier of $\cinfer \e \tv$, it follows immediately from
\cref{thm:soundness-and-completeness} that $\th \e : \sub(\tv)$, proving
\ref{proof:principal-types:1}.  For \ref{proof:principal-types:2}, suppose
$\th \e : \tp$ for some $\tp$. Then by \cref{thm:soundness-and-completeness}
again, there exists a unifier $\subp$ of $\cinfer \e \tv$ such that
$\subp(\tv) = \tp$. Since $\sub$ is most general, we have $\exists \subp
\centails \exists \sub$, and by \cref{lem:mgus}, this implies the existence
of a substitution $\subpp$ such that $\subp = \subpp \compose \sub$.  Hence,
$\tp = \subp(\tv) = \subpp(\sub(\tv))$, witnessing that $\tp$ is an instance
of $\sub(\tv)$, as required \ref{proof:principal-types:2}.
\end{proof}

\Draft{\input {draft}}{}

\clearpage
\setcounter{tocdepth}{1}
\tableofcontents

\end{document}

